\documentclass[12pt]{article}
\usepackage{graphicx}

\def\hybrid{\topmargin 0pt      \oddsidemargin 0pt
        \headheight 0pt \headsep 0pt
       \voffset-1cm
        \textwidth 6.25in       
       \textheight 9.5in       
        \marginparwidth 0.0in
        \parskip 5pt plus 1pt   \jot = 1.5ex}
\catcode`\@=11
\def\marginnote#1{}

\newcount\hour
\newcount\minute
\newtoks\amorpm
\hour=\time\divide\hour by60
\minute=\time{\multiply\hour by60 \global\advance\minute by-\hour}
\edef\standardtime{{\ifnum\hour<12 \global\amorpm={am}%
        \else\global\amorpm={pm}\advance\hour by-12 \fi
        \ifnum\hour=0 \hour=12 \fi
        \number\hour:\ifnum\minute<10 0\fi\number\minute\the\amorpm}}
\edef\militarytime{\number\hour:\ifnum\minute<10 0\fi\number\minute}

\def\draftlabel#1{{\@bsphack\if@filesw {\let\thepage\relax
   \xdef\@gtempa{\write\@auxout{\string
      \newlabel{#1}{{\@currentlabel}{\thepage}}}}}\@gtempa
   \if@nobreak \ifvmode\nobreak\fi\fi\fi\@esphack}
        \gdef\@eqnlabel{#1}}
\def\@eqnlabel{}
\def\@vacuum{}
\def\draftmarginnote#1{\marginpar{\raggedright\scriptsize\tt#1}}

\def\draftlabel#1{{\@bsphack\if@filesw {\let\thepage\relax
   \xdef\@gtempa{\write\@auxout{\string
      \newlabel{#1}{{\@currentlabel}{\thepage}}}}}\@gtempa
   \if@nobreak \ifvmode\nobreak\fi\fi\fi\@esphack}
        \gdef\@eqnlabel{#1}}
\def\@eqnlabel{}
\def\@vacuum{}
\def\draftmarginnote#1{\marginpar{\raggedright\scriptsize\tt#1}}

\def\draft{\oddsidemargin -.5truein
        \def\@oddfoot{\sl preliminary draft \hfil
        \rm\thepage\hfil\sl\today\quad\militarytime}
        \let\@evenfoot\@oddfoot \overfullrule 3pt
        \let\label=\draftlabel
        \let\marginnote=\draftmarginnote
   \def\@eqnnum{(\theequation)\rlap{\kern\marginparsep\tt\@eqnlabel}%
\global\let\@eqnlabel\@vacuum}  }

\def\numberbysection{\@addtoreset{equation}{section}
        \def\theequation{\thesection.\arabic{equation}}}

\def\underline#1{\relax\ifmmode\@@underline#1\else
        $\@@underline{\hbox{#1}}$\relax\fi}

\def\titlepage{\@restonecolfalse\if@twocolumn\@restonecoltrue\onecolumn
     \else \newpage \fi \thispagestyle{empty}\c@page\z@
        \def\thefootnote{\fnsymbol{footnote}} }

\def\endtitlepage{\if@restonecol\twocolumn \else  \fi
        \def\thefootnote{\arabic{footnote}}
        \setcounter{footnote}{0}}  
\relax

\numberbysection
\hybrid

\newfont{\Bbb}{msbm10 scaled 1\@ptsize00}
\newfont{\Bbbb}{msbm7 scaled 1\@ptsize00}
\newcommand{\CC}{\mbox{\Bbb C}}

\newcommand{\DDD}{\raise-1pt\hbox{$\mbox{\Bbbb D}$}}

\newcommand{\RR}{\mbox{\Bbb R}}

\newcommand{\UUU}{\raise-1pt\hbox{$\mbox{\Bbbb U}$}}

\newcommand{\ZZ}{\mbox{\Bbb Z}}
\newcommand{\z}{\raise-1pt\hbox{$\mbox{\Bbbb Z}$}}

\newcommand{\sss}{\raise-1pt\hbox{$\mbox{\Bbbb S}$}}

\def\beq{\begin{equation}}
\def\eeq{\end{equation}}
\def\p{\partial}
\def\t{\tau}

\def\t{{\sf t}}

\newtheorem{theorem}{Theorem}[section]

\newtheorem{lemma-definition}{Lemma-Definition}[section]

\newtheorem{example}{Example}[section]
\newtheorem{remark}{Remark}[section]

\newtheorem{proposition}{Proposition}[section]

\def\normord{ {\scriptstyle {{\bullet}\atop{\bullet}}} }
\def\lbr{\left <}
\def\rbr{\right >}

\def\square{\hfill
{\vrule height6pt width6pt depth1pt} \break \vspace{.01cm}}

\begin{document}

\begin{titlepage}

\title{Integrable hierarchies with zero dispersion \\ and 
elliptic curves}

\author{A. Savchenko\thanks{
Skolkovo Institute of Science and Technology, 143026, Moscow, Russia
and National Research University Higher School of Economics,
20 Myasnitskaya Ulitsa,
Moscow 101000, Russia,
e-mail: ksanoobshhh@gmail.com}
\and
A.~Zabrodin\thanks{
National Research University Higher School of Economics,
20 Myasnitskaya Ulitsa,
Moscow 101000, Russia and
NRC ``Kurchatov institute'', Moscow, Russia;
e-mail: zabrodin@itep.ru}}

\date{May 2026}
\maketitle

\vspace{-7cm} \centerline{ \hfill ITEP-TH-14/26}\vspace{7cm}

\vspace{0.1cm}

\begin{abstract}

We consider integrable hierarchies such as KP, modified KP, 
2D Toda lattice, BKP (small and large), DKP, Pfaff-Toda and 
their multi-component generalizations. 
We work in the framework of 
the bilinear formalism in which the universal
dependent variable is a tau-function satisfying bilinear equations 
of the Hirota-Miwa type. Our principal interest in this paper is
the dispersionless versions of the hierarchies.
In the limit of zero dispersion 
the main object is an $F$-function, which
is the limit of properly re-scaled logarithm of the tau-function.
We show that in all the cases there exists
an algebraic curve built into the structure of the hierarchy.
We call it the {\it dynamical curve}.
For the KP, modified KP and Toda lattice hierarchies, as well
as for their multi-component generalizations, the curve 
is rational (of genus 0) and can be uniformized by rational or
trigonometric functions. 
For hierarchies of the Pfaff type (DKP and Pfaff-Toda) the dynamical 
curve is in general a smooth elliptic curve (of genus 1), with its
modular parameter being a dynamical variable. It is also shown that
the large BKP hierarchy admits two different 
dispersionless versions. In one of them the dynamical curve degenerates
to a rational curve while in the other one it remains to be elliptic.
We show that a reformulation of the hierarchies based on uniformization 
of the dynamical curves by elliptic (or trigonometric) 
functions makes their structure nice and clear, 
especially in the multi-component case.

\end{abstract}

\end{titlepage}

\vspace{3cm}

%

\tableofcontents

\newpage

\section{Introduction}

The main purpose of this article is to provide, from a unified 
viewpoint and within a unified approach, a detailed 
account of the new developments in the theory of integrable hierarchies
(such as KP, mKP, BKP, DKP,  Pfaff-Toda
and their multi-component generalizations 
in the zero dispersion limit) related to their algebraic structures
and initiated in our previous works \cite{SZ24}-\cite{Z24a}.
The main new phenomenon is a rather unexpected emergence of rational or 
elliptic curves in this context. Since parameters of these curves
are dynamical variables, we call them {\it dynamical curves}. 
It turns out that
a special change of variables motivated by uniformization of 
the dynamical curve by means of trigonometric or elliptic
functions makes the structure of the hierarchy in question 
nice and clear, especially in the multi-component case. 

In this article, we have tried to systematize somehow the rather
scattered results mentioned above. 
Along with revisiting already known results
(which have been obtained here
by more rigorous and systematic methods than it was done 
in the earlier works), some new interesting manifestations 
of the dynamical curve phenomenon are found on this basis.

\subsection{Algebraic curves in integrable systems: a bird's-eye view}

Emergence of complex algebraic curves (Riemann surfaces) in 
mathematical physics, and, in particular, in the
theory of integrable systems, both classical and quantum, is of course not
a new phenomenon. Nowadays, there are several seemingly unrelated 
branches of the theory, where algebraic
curves and the related mathematical apparatus play a pivotal role.
To put the subject of this article in a broader context, some
of them are listed below. 

\begin{itemize}
\item[a)] The most striking manifestation of effectiveness 
of algebraic-geometrical methods in the theory of integrable systems
is the construction of periodic and quasi-periodic solutions
to the classical nonlinear differential equations and their 
integrable difference analogues. The literature on these issues 
is enormously vast, see, e.g., reviews \cite{DMN76,Dubrovin81},
original papers \cite{Krichever77,Krichever77a} and 
references therein. In a nutshell, this theory claims the
following: given an algebraic curve $\Gamma$ of any finite genus 
${\sf g}>0$ and some additional data 
(such as marked points $P_{\alpha}\in \Gamma$ and
local parameters in their neighborhoods), one can construct 
quasi-periodic solutions to nonlinear integrable equations
and their hierarchies. In general, the solutions are expressed
through the Riemann (or Prim) theta-functions associated
with the curve. The curve $\Gamma$ (which often can be realized as
spectral curve of some differential or difference operator) 
is an integral of motion, and the dynamics
takes place in its Jacobian.

\item[b)] Another branch of the theory 
is the universal Whitham hierarchy of dynamical systems on algebraic
curves in Krichever's formulation \cite{Krichever92}.
In the nutshell, this is a method aimed to (at least,
approximate) description
of more general solutions to the same integrable equations
in terms of the algebraic-geometrical solutions. The theory
can be formulated as a dynamical system for certain meromorphic
differentials on an algebraic curve $\Gamma$ but now the curve 
is no longer an integral of motion. It becomes a dynamical variable,
i.e., depends on the independent variables (``times'') 
of the hierarchy. Figuratively speaking, the curve 
(which again can be of arbitrary
finite genus ${\sf g}$) ``breathes'', 
i.e., its moduli depend on the times.
In particular, the Whitham theory for curves of genus ${\sf g}=0$ is
precisely the theory of integrable hierarchies in the zero
dispersion limit \cite{TT95}.

\item[c)] The previous two issues are related to 
classical integrability.
In the theory of quantum integrable models (and closely 
related models of statistical
mechanics on the 2D lattice), there is an intrinsic
variable called {\it spectral parameter}, which lives on an
algebraic curve \cite{book1,Slavnov,Baxter}. In the most popular 
(and simple) examples, the curve
is in general of genus ${\sf g}\leq 1$ (a smooth rational or 
elliptic curve), which is an integral
of motion. However, there are rather exotic 
models, such as the chiral Potts model \cite{chiralPotts},
in which the spectral parameter may live on certain curves of 
genus ${\sf g}>1$.

\item[d)] A closely related example
of algebraic curves in quantum integrable models is their
appearance as {\it vacuum curves} of quantum $L$-operators
\cite{Krichever81,KZ99}. In the cases when they have been described
explicitly, their genus is ${\sf g}\leq 1$,
and they are, generally speaking, singular curves (for example,
a ``bouquet'' of elliptic or rational curves, see 
\cite{KZ99}).
The vacuum curves, as well as spectral curves, are integrals 
of motion.

\item[e)] Elliptic curves (${\sf g}=1$) are known to 
emerge in the theory
of dispersionless integrable hierarchies of the Pfaff type,
such as DKP and Pfaff-Toda. This is a rather unexpected 
phenomenon because
dispersionless hierarchies, in accordance with the Whitham theory 
in Kri\-che\-ver's formulation, were usually associated with rational
curves of genus $0$ rather than $1$.

\end{itemize}

The last relatively new item is just what the present paper is devoted to.
The story has begun with the observation made by Takasaki in
\cite{Takasaki07,Takasaki09} that the dispersionless versions
of integrable hierarchies of the Pfaff type 
contain a hidden (in general smooth) elliptic 
curve naturally built in the structure
of the hierarchy. We suggest to call it the dynamical curve,
since its parameters depend on the hierarchical times.

\subsection{Dynamical curves: elliptic and rational}

First of all,
a few words about the Pfaff hierarchies are in order.
Along with the well-known hierarchies of integrable 
Kadomtsev-Petviashvili (KP) and 2D Toda type equations, 
there are more general (and still less studied) Pfaffian versions of them, 
so named because some of their exact solutions are expressed not in terms of determinants (like for KP and Toda), but in terms of Pfaffians. 
One such hierarchy, first briefly mentioned in 
the work by Jimbo and Miwa in 1983 \cite{JM83}, 
was rediscovered several times later and is now known in the literature 
under more than one name (DKP, coupled KP, etc.), 
see \cite{HO}--\cite{Kodama}. 
The Pfaff-Toda hierarchy was introduced by Takasaki in \cite{Takasaki09}.
Recently, in our papers \cite{SZ24,SZ25a} multi-component generalizations
of these hierarchies were suggested.

The theory of their dispersionless versions
was further developed in \cite{AZ14,AZ15}, where it was
shown that passing to new variables motivated by uniformization
of the elliptic curve via elliptic functions, makes the structure
of the hierarchy nice and clear (at the cost of having to deal with 
non-elementary functions). 

The advantage of the elliptic parametrization
is especially evident in studying multi-component hierarchies
\cite{DJKM81}-\cite{TZ25}. Their original formulation contains many
seemingly different equations which are
hardly amenable to order and control. As is shown in \cite{SZ24,SZ25a},
the elliptic parametrization reduces all this set of many 
equations to just a single one, which is proved to be 
equivalent to the whole hierarchy. It is important to note
that the elliptic modular parameter $\tau$ of elliptic functions
is a dynamical variable, which in general
depends on all hierarchical times.

Moreover, in \cite{Z24}
it was shown that this approach is applicable to hierarchies 
of the type A (i.e., 
``usual'', not of the Pfaff type, such as KP or mKP) in the
limit of zero dispersion, including their 
multi-component generalizations, and allows one to significantly clarify
their structure. In these cases, there is an underlying algebraic
curve, too, but instead of being elliptic it turns out 
to be rational and admits uniformization by elementary functions
(rational, trigonometric or hyperbolic). 
Such a parametrization
makes the structure of multi-component hierarchies
as clear as in the one-component case. From this point of view,
dispersionless
hierarchies of the Pfaff type look like ``elliptic deformations''
of the standard dispersionless KP or Toda hierarchies.
Or, equivalently, the latter can be thought of as a degeneration
of Pfaff hierarchies as the modular parameter $\tau$ tends
to infinity: $\tau \to +i\infty$.

\subsection{General structure of the paper and its main contents}

Our attention is directed to different integrable
hierarchies which we divide into two big groups: ``usual'' 
(of type $A$) KP, mKP, Toda hierarchies and their multi-component
generalizations, and hierarchies of the Pfaff type:
small BKP, large BKP, DKP, Pfaff-Toda as well as 
their multi-component generalizations. The main object for us
in each case is the corresponding tau-function $\tau ({\bf T})$
depending on
an infinite set ${\bf T}$ of independent variables (``times'').
This set may include both continuous ($\RR$- or $\CC$-valued) 
and discrete ($\ZZ$-valued) times, 
specific for each particular hierarchy. For example,
in the simplest case of the one-component KP hierarchy
$$
{\bf T}=\{t_1, t_2, t_3, \ldots \}, \quad t_i\in \CC ,
$$
while for the $N$-component KP and DKP hierarchies it is
$$
\begin{array}{l}
{\bf T}={\bf n} \cup {\bf t}_{1} \cup \ldots {\bf t}_{N},  \quad
\mbox{where}
\\ \\
{\bf t}_{\alpha}=\{t_{\alpha ,1}, t_{\alpha ,2}, t_{\alpha ,3}, \ldots \}, 
\quad t_{\alpha ,i}\in \CC ,
\\ \\
{\bf n}=\{ n_1, \ldots , n_N\}, \quad n_i \in \ZZ .
\end{array}
$$
The tau-function, as a function of all the times, serves
as a universal dependent variable.
Our starting point is the bilinear integral functional relation
for it. Schematically, its general form is
\beq\label{int0}
\frac{1}{2\pi i}
\sum_{\gamma =1}^N\oint_{C_{\infty}}e^{\xi_{\gamma}({\bf T}-{\bf T'}, z)}
e^{-\nabla_{\gamma}(z)} \tau ({\bf T}) \cdot 
e^{\nabla'_{\gamma}(z)} \tau ({\bf T'}) \, dz =0
\eeq
for hierarchies of the type $A$ (KP, modified KP and their
$N$-component generalizations) and
\beq\label{int01}
\begin{array}{l}
\displaystyle{
\phantom{+\,}
\frac{1}{2\pi i} 
\sum_{\gamma =1}^N\oint_{C_{\infty}}
e^{\xi_{\gamma}({\bf T}-{\bf T'}, z)}
e^{-\nabla_{\gamma}(z)} \tau ({\bf T}) \cdot 
e^{\nabla'_{\gamma}(z)} \tau ({\bf T'}) \, dz }
\\ \\
\displaystyle{
+\, \, \frac{1}{2\pi i} \sum_{\gamma =1}^N
\oint_{C_{\infty}}e^{-\xi_{\gamma}({\bf T}-{\bf T'}, z)}
e^{\nabla_{\gamma}(z)} \tau ({\bf T}) \cdot 
e^{-\nabla'_{\gamma}(z)} \tau ({\bf T'}) \, dz} 
\\ \\
\, =\,\, 
\varepsilon ({\bf T}, {\bf T'})\,
\tau ({\bf T})\tau ({\bf T'})
\end{array}
\eeq
for the $N$-component hierarchies of the Pfaff type (DKP and large BKP). 
These relations hold for arbitrary
${\bf T}$ and ${\bf T'}$. The integration contour is a big circle
of radius $R\to \infty$. The other notations are: 
$\xi_{\gamma}({\bf T},z)$ is a certain linear function of the times,
$\nabla_{\gamma}(z)$ is a differential operator acting
to functions of ${\bf T}$, $\nabla_{\gamma}'(z)$ is the same
operator acting to functions of ${\bf T'}$. For example,
in the $N$-component KP hierarchy they are
\beq\label{int03}
\xi_{\gamma}({\bf T}, z)=n_{\gamma} \log z +
\sum_{k\geq 1}t_{\gamma ,k} z^{k}, \quad
\nabla_{\gamma}(z)=\p_{n_{\gamma}}+
\sum_{k\geq 1}\frac{z^{-k}}{k}\, \p_{t_{\gamma ,k}}.
\eeq 
At last, $\varepsilon ({\bf T}, {\bf T'})$ in (\ref{int01}) 
can be $0$ or $1$ depending on values of the discrete variables
from the sets ${\bf T}$ and ${\bf T'}$.
More details on the notation are given in the
corresponding sections of the main text.

For each case, we give a representation of the tau-function as
an expectation value of certain operators made of free fermions
(except for the BKP tau-functions). This is done mostly for
illustrative purposes, assuming that the reader is to some extent
familiar with the fermionic operator approach developed 
by the Kyoto school (see, e.g., \cite{JM83,DJKM83,AZ13}).
This approach
provides a very clear and convenient language for analysis 
of integrable hierarchies. 
It seems to us that having in mind the fermionic picture helps
a lot to understand internal structures of the 
hierarchies and interrelations
between them.

In this paper we are mostly interested in the dispersionless limit
of the hierarchies mentioned above. 
Following the approach developed in \cite{TT95}, in order to 
pass to the limit, one should
introduce an extra parameter $\hbar \to 0$ 
and re-scale the times ${\bf T}$
as ${\bf T}\to {\bf T}/\hbar$. Next,
introduce the function $F$ of the 
re-scaled times
$F({\bf T}; \hbar )$ related to the 
tau-function by the
formula
\beq\label{d1b}
\tau \Bigl ({\bf T}/\hbar \Bigr )=
\exp \left ( \frac{1}{\hbar^2}\,
F({\bf T}; \hbar )\right )
\eeq
and consider the limit  
\beq \label{d1b1}
F=F({\bf T})=
\lim\limits_{\hbar \to 0}
F({\bf T}; \hbar ),
\eeq 
if it exists. (It is known that for a large class of solutions 
for $\tau ({\bf T})$, coming, for example, from the theory
of random matrices and logarithmic gases, the limit does exist.) 
The function $F$
plays the role of the 
tau-function in the dispersionless limit $\hbar \to 0$, meaning that
it serves as a universal dependent variable.
The $F$-function satisfies 
an infinite number of highly nonlinear differential equations which are
limiting cases of the bilinear equations for the 
tau-function. In the main text, the hierarchies with zero
dispersion will be abbreviated by adding 
the small letter ``d'' to their abbreviated names:
dKP, dBKP, dmKP, etc.

Unfortunately, it is not clear to us whether it is possible to
obtain equations  of the dispersionless hierarchies 
for the $F$-function directly from the general integral bilinear
relations of the form (\ref{int0}) or (\ref{int01}). 
A way out is to use for this purpose the
infinite set of bilinear Hirota-Miwa equations each of which
contains a finite number of terms. They are known to be
equivalent to the integral relations and can be derived from them with the
help of the so-called Miwa substitution first suggested in \cite{Miwa82}.
The bilinear equations obtained in this way are already 
suitable for taking the dispersionless limit. In the limit $\hbar \to 0$,
they convert into highly nonlinear equations for the $F$-function.
We call them dispersionless Hirota-Miwa equations. 

Note that
after the re-scaling $n_{\alpha}=t_{\alpha ,0}/\hbar$, 
the limit $\hbar \to 0$ makes the former discrete times $n_{\alpha}$
continuous (i.e., $ t_{\alpha ,0}\in \RR$),
and all the equations for the $F$-function 
are differential rather than
differential-difference or purely difference. The key point is that
combining them, one can obtain certain equations for some functions
containing second order partial derivatives of $F$ and an additional
complex variable $z$. To be more concrete, we outline here what
happens for the $N$-component DKP hierarchy. In this case the
functions mentioned above can be chosen as
\beq\label{ww2}
w_{\alpha}(z)=z^{-1}e^{\nabla_{\alpha}(z)\p_{\alpha}F}, \quad
w_{\alpha \beta}(z)=e^{\nabla_{\alpha}(z)\p_{\beta}F} \quad
(\alpha \neq \beta ),
\eeq
where $\p_{\alpha}\equiv \p_{t_{\alpha , 0}}$ and
$\displaystyle{\nabla_{\alpha}(z)=\p_{\alpha}+
\sum_{k\geq 1}\frac{z^{-k}}{k}\, \p_{t_{\alpha ,k}}}$ is the 
operator of the form (\ref{int03}). 
We prove the following key proposition.

\begin{proposition}
Let $F$ be a
solution to the whole set of the dispersionless Hirota-Miwa
equations of the dDKP hierarchy. Then
the functions $w_{\alpha}(z)$, $w_{\alpha \beta}(z)$
are constrained by the
following relation:
\beq\label{int04}
R_{\alpha \beta}^2 \Bigl (w_{\alpha}^2(z)w_{\alpha \beta}^2(z)+1
\Bigr ) -\Bigl (w_{\alpha}^2(z)+w_{\alpha \beta}^2(z)\Bigr )+
V_{\alpha \beta}w_{\alpha}(z)w_{\alpha \beta}(z)=0.
\eeq
\end{proposition}

\noindent
Here $R_{\alpha \beta}$ and $V_{\alpha \beta}$ are some functions
depending on the times (but not on $z$), which are expressed through
second order partial derivatives of the $F$-function in a known way. 
Equation (\ref{int04}), being quadratic in each of the
variables $x=w_{\alpha}$, $y=w_{\alpha \beta}$,
defines an elliptic curve $P(x,y)=0$, where
$
P(x,y)=R^2(x^2y^2 +1)-(x^2 +y^2)+Vxy
$.
It is the dynamical curve for the dDKP hierarchy.
In general, the curve is smooth, with the elliptic modulus being
a function of the times. The functions (\ref{ww2}) are meromorphic
functions on this curve, and $z^{-1}$ plays
the role of a local parameter in a neighborhood of infinity.

It is well known that any algebraic curve defined by the equation 
$P(x,y)=0$
with a bi-quadratic polynomial $P(x,y)$ can be uniformized by
elliptic functions (or Jacobi's theta-functions).
In the explicit form, the uniformization of the curve (\ref{int04})
is given by
\beq\label{uni111}
w_{\alpha}(z)=
\frac{\theta_1(u_{\alpha}(z)-\eta_{\alpha})}{\theta_4(u_{\alpha}(z)-
\eta_{\alpha})}, \qquad
w_{\alpha \beta}(z)=\epsilon_{\beta \alpha}
\frac{\theta_1(u_{\alpha}(z)-\eta_{\beta})}{\theta_4(u_{\alpha}(z)-
\eta_{\beta})},
\eeq
where $\epsilon_{\beta \alpha}=\pm 1$ is some sign factor and
$\theta_{i}(u)=\theta_{i}(u|\tau )$ are Jacobi's theta-functions
with a modular parameter $\tau$ to be fixed using some additional
arguments (it turns out to be a dynamical variable, i.e., 
in general depends on the times). The functions $u_{\alpha}(z)$ 
are expanded in series 
of the form $\displaystyle{u_{\alpha}(z)=\eta_{\alpha}({\bf t}) + \sum_{k\geq 1}c_k^{(\alpha )} 
({\bf t})z^{-k}}$ as $z\to \infty$,
and the coefficients $\eta_{\alpha}({\bf t}),
c_{k}^{(\alpha )}({\bf t})$ are 
new dynamical variables.
The advantage of this change of variables is obvious: instead of
$N^2$ functions $w_{\alpha},\, w_{\alpha \beta}$ with complicated
relations between them
it is enough to 
deal with only $N$ ones. Formulas (\ref{uni111}) should be
supplemented with the following expressions for $R_{\alpha \beta}$ and
$V_{\alpha \beta}$:
\beq\label{RV111}
R_{\alpha \beta}=\epsilon_{\beta \alpha}
\frac{\theta_1(\eta_{\alpha \beta})}{\theta_4(\eta_{\alpha \beta})}, 
\qquad
V_{\alpha \beta}=
2\epsilon_{\beta \alpha} \frac{\theta_4^2(0)\, \theta_2(\eta_{\alpha \beta})
\, \theta_3(\eta_{\alpha \beta})}{\theta_2(0)\, 
\theta_3(0)\,\theta_4^2(\eta_{\alpha \beta})}, \quad
\eta_{\alpha \beta}\equiv \eta_{\alpha}-\eta_{\beta}.
\eeq
The uniformization means that after substitution of (\ref{uni111}) and
(\ref{RV111}) into (\ref{int04}), the latter equation 
holds identically.

The main advantage of the  elliptic
parametrization is that a plethora of equations that arise in the
traditional approach (which are hardly suitable for putting 
them in order) are
replaced by only one equation with a clear structure.
We prove the following theorem.

\begin{theorem}
The whole $N$-component dDKP hierarchy is equivalent to the
single equation
\beq\label{E111}
\epsilon_{\beta \alpha}(a^{-1}-b^{-1})^{\delta_{\alpha \beta}} e^{\nabla_{\alpha}(a)\nabla_{\beta}(b)F}
= \frac{\theta_1(u_{\alpha}(a)- 
u_{\beta}(b)|\tau )}{\theta_4(u_{\alpha}(a)-
u_{\beta}(b)|\tau )}.
\eeq
The modular parameter $\tau =\tau ({\bf t})$ of the theta-functions
is a dynamical variable. Its dependence on times is implicitly 
determined by the equation
\beq\label{E11b}
e^{2\p_{\alpha}\p_{\beta}F}+e^{-2\p_{\alpha}\p_{\beta}F}-
e^{-2\p_{\alpha}^2F}(\p_{\beta}\p_{t_{\alpha ,1}}F)^2 =
\frac{\theta_2^2 (0|\tau )}{\theta_3^2 (0|\tau )}+
\frac{\theta_3^2 (0|\tau )}{\theta_2^2 (0|\tau )}.
\eeq
\end{theorem}

Moreover, we also prove similar statements for the
so-called large BKP hierarchy, including its multi-component
generalization. 
This hierarchy was originally suggested
by Kac and van de Leur in \cite{KL98} under another name,
was further studied in \cite{OST12,OST16,vLO15} and called there
``large BKP hierarchy'', as opposed to the 
``small BKP hierarchy'' \cite{DJKM82}-\cite{Z21}.
It is closely connected with the KP hierarchy of type D.
Recently this hierarchy was rediscovered under the name
of the $B$-Toda hierarchy \cite{KZ22,PZ23}, and its 
previously unknown close connection with the 2D Toda hierarchy
was established. Its multi-component generalization suggested
in the present paper within the framework of the bilinear
approach seems to be new.

In short, the new results on the large BKP hierarchy obtained in the 
present work are as follows. We have shown that it admits two
essentially different dispersionless versions: version I and version II. 
The simpler version I was already discussed in \cite{Z24a} for the
one-component hierarchy. We have extended the analysis to the
multi-component setup. 
The result is that the version I of the limit corresponds to 
a degeneration of the elliptic curve to a rational one
in the limit $\tau \to +i0$. 
The equation of the dynamical curve (\ref{int04}) acquires the form
\beq\label{split3}
R^2_{\alpha \beta} (w_{\alpha}w_{\alpha \beta}+1)^2=
(w_{\alpha} +w_{\alpha \beta} )^2,
\eeq
i.e., the curve becomes singular and splits into two rational components
$
R_{\alpha \beta} (w_{\alpha}w_{\alpha \beta}+1)=\pm
(w_{\alpha} +w_{\alpha \beta} )$,
which can be uniformized by means of elementary (hyperbolic) functions:
\beq\label{uni112}
w_{\alpha}(z)=
\tanh (u_{\alpha}(z)-\eta_{\alpha}), \qquad
w_{\alpha \beta}(z)=\epsilon_{\beta \alpha}
\tanh (u_{\alpha}(z)-\eta_{\beta}),
\eeq
with $R_{\alpha \beta}=
\epsilon_{\beta \alpha}\tanh (\eta_{\alpha} -\eta_{\beta})$.

\begin{theorem}
The whole $N$-component 
large BKP hierarchy in the dispersionless 
version I is equivalent to the single equation
\beq\label{E112}
\epsilon_{\beta \alpha}(a^{-1}-b^{-1})^{\delta_{\alpha \beta}} e^{\nabla_{\alpha}(a)\nabla_{\beta}(b)F}
= \tanh (u_{\alpha}(a)- u_{\beta}(b)).
\eeq
\end{theorem}

The zero dispersion version II is performed in a
somewhat sophisticated way and is essentially ``elliptic'', i.e., 
can be parametrized by means of non-degenerate 
elliptic functions. To be more precise, the result for the 
$N$-component large BKP hierarchy turns out to be basically the
same as the $(N+1)$-component dispersionless DKP hierarchy.

Amusingly, the similar approach to the $N$-component dispersionless
KP and mKP hierarchies (with $N\geq 2$) 
formally leads to equations that correspond
to another degeneration
of elliptic functions ($\tau \to +i\infty$ rather than
$\tau \to +i0$). Compared to the BKP case, it consists in replacing the
$\tanh$-function by the $\sin$-function.
So, one can say that the two cases are
connected by the (limiting form of) 
modular transformation $\tau \leftrightarrow -1/\tau$.
The explicit formulas are:
\beq\label{uni113}
w_{\alpha}(z)=
\sin (u_{\alpha}(z)-\eta_{\alpha}), \qquad
w_{\alpha \beta}(z)=\epsilon_{\beta \alpha}
\sin (u_{\alpha}(z)-\eta_{\beta})
\eeq
instead of (\ref{uni112}) and
\beq\label{t4a}
\epsilon_{\beta \alpha}(a^{-1}-b^{-1})^{\delta_{\alpha \beta}}
e^{\nabla_{\alpha}(a)\nabla_{\beta}(b)F}=
\sin (u_{\alpha}(a)-u_{\beta}(b))
\eeq
instead of (\ref{E112}).

\subsection{Organization of the paper}

Sections 2--4 are devoted to hierarchies of type A.
In Section 2 we consider the simplest and most 
familiar examples of the KP and
mKP hierarchies. The procedure of passing to the 
zero dispersion limit 
is discussed there in detail. Although some key structures that 
appear in what follows in more complicated cases (such as the curve)
in these examples 
are ``in their infancy'' and rather implicit, it nevertheless
seems to be instructive to start with such simple and well known 
matters.
Section 3 is devoted to the same hierarchies of the type A
(KP and mKP), but in their multi-component incarnation.
Now the algebraic curve appears on the scene in an explicit way.
However, it is still only rational. Nevertheless, its uniformization 
by means of trigonometric functions helps a lot in clarifying the
structure of the hierarchies for $N>2$.
In Section 4 these results are translated to the language of
the (dispersionless) 
$N$-component Toda lattice hierarchy, which is known to be
equivalent to the $2N$-component KP.

Sections 5 and 6 are devoted to hierarchies of the Pfaff type:
one-component (Section 5) and multi-component (Section 6).
The presentation begins with the simplest case of the small BKP
hierarchy, which can be attributed to this class only somewhat 
conditionally, on the grounds that some of 
its solutions are expressed through Pfaffians. 
Like for the one-component KP case, the curve does not appear
in this example (or appears in a trivial way, being defined by
an equation of the form $x-y=0$). The second example is already
rather serious and non-trivial: we have tried to show in detail
how the elliptic curve emerges in the dispersionless DKP hierarchy,
and how to benefit from its uniformization by elliptic functions.
The third example in this section is the large BKP hierarchy,
which in some sense can be regarded as a difference analogue 
of its small sister.
It is closely related to DKP, but contains more equations and
admits two different zero dispersion limits. In one of them
(the simplest one)
the additional equations imposed to the same 
$F$-function cause the elliptic curve to degenerate. This leads
to parametrization in terms of hyperbolic functions, which in the
one-component case is not that too meaningful and informative. 
The analysis of the other limit requires involvement of multi-component
hierarchies and is postponed to the next section.
Section 6 is devoted to the multi-component DKP 
and large BKP hierarchies, for which the elliptic parametrization
considerably clarifies and simplifies their structure.

Section 7 contains the conclusion.
Some unsolved problems are also mentioned there.

Lastly, there are four appendices. Appendix A contains the basic
definitions and formulas of the free fermions operator approach.
In Appendix B the definitions and main properties of the 
Jacobi's theta-functions necessary in the main text are given.
In Appendix C it is explained how uniformization of rational and
elliptic curves, by trigonometric and elliptic 
functions respectively, works in practice.
In Appendix D it is shown how the approach to the $N$-component dmKP
hierarchy can be specified to the case $N=1$.

\section{The simplest examples: 
rational curves in one-component hierarchies}

\subsection{A warm-up exercise: the dKP hierarchy}
\label{section:dKP}

We start with the simplest possible case of the dKP hierarchy, 
where any algebraic curve
is not appearing yet\footnote{Better to say, 
the curve exists in this case, too,
but is the most trivial rational curve: if desired, it formally 
can be defined by the equation $P(x,y)=0$ with $P(x,y)=x-y$.}. 
Nevertheless, it seems instructive 
to start with this particular case and work out some 
important details of calculations 
(which will be encountered later in more complex and meaningful cases)
using this simple example. 
Most of the technical details are common for 
more general and complicated cases addressed later.

\subsubsection{The KP hierarchy}
\label{section:KP}

In the KP hierarchy,
the set of independent variables is ${\bf t}=\{t_1, t_2, t_3, \ldots \}$.
The universal dependent variable is the 
tau-function $\tau ({\bf t})$ that satisfies the integral bilinear
relation
\beq\label{kp1}
\oint_{C_{\infty}} dz \, e^{\xi ({\bf t}-{\bf t'}, z)} 
\tau \Bigl ({\bf t}-[z^{-1}]\Bigr )
\tau \Bigl ({\bf t'}+[z^{-1}]\Bigr ) \, =0.
\eeq
Here
\beq\label{xi1}
\xi (z, {\bf t})=\sum_{k\geq 1} t_k z^k, \quad
{\bf t}\pm [z^{-1}]=\Bigl \{ t_1 \pm z^{-1}, t_2 \pm \frac{1}{2}
z^{-2}, t_3 \pm \frac{1}{3} z^{-3}, \ldots \Bigr \}.
\eeq
The integration contour $C_{\infty}$ 
is a big circle of radius
$R\to \infty$ such that all singularities 
coming from the exponential function
$e^{\xi ({\bf t}-{\bf t'} z)}$
are outside it and all singularities 
coming from the $\tau$-factors are inside it (the 
$\tau$-factors, as functions of $z$, are assumed to be regular in some
neighborhood of infinity.) This equation is valid for any choice of
the times ${\bf t}$, ${\bf t'}$. 
The simplest solution is $\tau ({\bf t})=1$.
Note that the tau-function is defined up to a transformation of the 
form
$$
\tau ({\bf t})\longrightarrow e^{\ell \, ({\bf t})}\tau ({\bf t}),
$$
where $\ell ({\bf t})$ is an arbitrary linear function of the times.

In the operator approach (outlined in Appendix A) the tau-function is
represented as the following vacuum expectation value:
\beq\label{vackp}
\tau ({\bf t})=\lbr 0 \bigl |e^{J({\bf t})}g \bigr |0\rbr ,
\eeq
where $g$ is a neutral Clifford group element of the form
$$
g=\exp \Bigl (\sum_{i,j\in \z}A_{ij}\psi_i \psi^{*}_j\Bigr ).
$$

To proceed,
one should choose ${\bf t}-{\bf t}'$ in (\ref{kp1})
in such a way that the integral
could be evaluated by means of residue calculus. 
This is equivalent to the Miwa change of
variables first introduced by Miwa in \cite{Miwa82}.
Namely, we employ the substitution of the following
general form:
\beq\label{kp2}
\begin{array}{l}
\displaystyle{
{\bf t}-{\bf t}' =\sum_{i=1}^{P^+}[a_i^{-1}]-
\sum_{k=1}^{P^-}[b_k^{-1}].}
\end{array}
\eeq
Here $P^+, P^-$ are non-negative integer numbers,
$a_i , b_k \in \CC$ are arbitrary 
parameters\footnote{Technically 
it is much easier to make the calculations
assuming that they are
distinct and, if necessary, consider various limits
when some of them coincide afterwards.} 
(the Miwa variables) 
belonging to a neighborhood of infinity. 
After this substitution the factor $e^{\xi ({\bf t}-{\bf t'}, z)}$ becomes
\beq\label{kp3}
\begin{array}{lll}
e^{\xi ({\bf t}-{\bf t'}, z)} & = &
\displaystyle{\exp \left ( \sum_{i=1}^{P^+}\sum_{k\geq 1}
\frac{1}{k}(z/a_i)^k -\sum_{j=1}^{P^-}\sum_{k\geq 1}
\frac{1}{k}(z/b_j)^k\right )}
\\ && \\
& = &
\displaystyle{\exp \left ( -\sum_{i=1}^{P^+}\log 
\Bigl (1-\frac{z}{a_i}\Bigr ) +
\sum_{k=1}^{P^-}\log \Bigl (1-\frac{z}{b_k}\Bigr )\right )}
\\ && \\
&=& \displaystyle{
\left ( \prod_{i=1}^{P^+} \frac{a_i}{a_i-z}\right )
\left ( \prod_{k=1}^{P^-} \frac{b_k-z}{b_k}\right ).}
\end{array}
\eeq
So, the Miwa change generates simple 
poles at the points $a_i$ instead of an essential 
singularity at $\infty$, and
the integral in (\ref{kp1})
can be calculated by residue calculus. 
Calculating the residues, one should take into account
that all the points $a_i$ are {\it outside} the contour, so 
the residues at these points should be taken with the sign
``$-$''; and there may be also a contribution from $\infty$.
From (\ref{kp3}) it is clear that if
\beq\label{kp4}
P^+ -P^- \geq 2,
\eeq
the residue at $\infty$ is zero. Below in this section we assume that
inequality (\ref{kp4}) holds. We also assume that 
all $a_i$'s are distinct, so the poles are simple. 
Calculating the residues and shifting the variables, we represent 
(\ref{kp1}) in the following form:
\beq\label{kp5}
\sum_{s=1}^{P^+}
\left ( \prod_{i=1, \neq s}^{P^+} (a_i-a_s)\right )^{-1}\!\!
\left ( \prod_{k=1}^{P^-} (b_k-a_s)\right )
\tau \Bigl ( {\bf t} +\sum_{i=1, \neq s}^{P^+}[a_i^{-1}]\Bigr )
\tau \Bigl ({\bf t}+[a_s^{-1}]+
\!\! \sum_{k=1}^{P^-}[b_k^{-1}]\Bigr )=0.
\eeq
The simplest non-trivial case is $(P^+, P^-)=(3,0)$:
\beq\label{kp5a}
\begin{array}{l}
(a_1-a_2)\tau \Bigl ({\bf t}+[a_1^{-1}]+[a_2^{-1}]\Bigr )
\tau \Bigl ({\bf t}+[a_3^{-1}]\Bigr )+
(a_2-a_3)\tau \Bigl ({\bf t}+[a_2^{-1}]+[a_3^{-1}]\Bigr )
\tau \Bigl ({\bf t}+[a_1^{-1}]\Bigr )
\\ \\
\phantom{aaaaaaaaaaaaa}
+\, (a_3-a_1)\tau \Bigl ({\bf t}+[a_3^{-1}]+[a_1^{-1}]\Bigr )
\tau \Bigl ({\bf t}+[a_2^{-1}]\Bigr )=0.
\end{array}
\eeq
This is the famous 3-term bilinear relation for the KP tau-function
first obtained 
by Miwa \cite{Miwa82}.
It holds for all $a_1, a_2, a_3$. Expanding the right-hand side
in inverse powers of $a_i\to \infty$, and equating the coefficients
to zero, one obtains an infinite set of partial differential
equations for the tau-function.

\begin{theorem}\cite{TT95,Shigyo13}\label{theorem:kp}
Equation (\ref{kp5a}) is equivalent to the whole KP hierarchy
defined by (\ref{kp1}).
\end{theorem}

\noindent
The direct proof can be found in \cite{Shigyo13}.
Below we shall see that it drastically simplifies for 
the hierarchy with zero dispersion (Proposition
\ref{proposition:dkp-eq}).

\subsubsection{The $\hbar$-KP hierarchy}
\label{section:hbarKP}

As an intermediate step for passing to the dispersionless version,
we include into play an extra parameter $\hbar$ (a formal dispersion
parameter)
and re-scale the times ${\bf t}$ as
$t_{k}\to t_{k}/\hbar$ for all $k \geq 1$.
Introduce the function 
$F({\bf t}; \hbar )$ related to the tau-function 
by the formula
\beq\label{kp6}
\tau \Bigl (\hbar^{-1}{\bf t}  \Bigr )=
\exp \left ( \frac{1}{\hbar^2}\,
F({\bf t}; \hbar )\right ).
\eeq
The $F$-function satisfies 
an infinite number of highly nonlinear differential equations which are
obtained from the bilinear equations for the 
tau-function rewritten in terms of the
$F$-function\footnote{This form of the KP hierarchy is sometimes
called the $\hbar$-KP hierarchy, see, e.g., \cite{TT95,NZ15}.}.
Of course for any $\hbar \neq 0$ the $\hbar$-KP hierarchy is 
equivalent to the original one corresponding to $\hbar =1$.

To represent the equations from Section \ref{section:KP}
in the form suitable for the
dispersionless limit, we introduce the operators
\beq\label{kp7}
D(z)=\sum_{k\geq 1}\frac{z^{-k}}{k}\,
\p_{t_{k}}, \qquad \Delta_{\hbar}(z)=\frac{e^{\hbar D(z)}-1}{\hbar},
\eeq
then 
\beq\label{kp8}
\tau \Bigl ({\bf t}/\hbar \pm [a^{-1}]\Bigr )=
\exp \left (\frac{1}{\hbar^2}\, e^{\pm \hbar D(a)}
F ({\bf t}; \hbar )\right ),
\eeq
and, more generally,
\beq\label{kp9}
\tau \left (\frac{{\bf t}}{\hbar} +
\sum_{i=1}^{P^+}[a_i^{-1}]-
\sum_{k=1}^{P^-}[b_k^{-1}]\right )
\Bigr )=
\exp \left (\frac{1}{\hbar^2}\, e^{\hbar 
\sum_{i=1}^{P^+}D(a_i)- \hbar \sum_{k=1}^{P^-}D(b_i)}
F ({\bf t}; \hbar )\right ).
\eeq
So, the $\hbar$-version of the general equation (\ref{kp5}) reads
\beq\label{kp5x}
\sum_{s=1}^{P^+}
\left ( \prod_{i=1, \neq s}^{P^+} (a_i-a_s)\right )^{-1}\!\!
\left ( \prod_{k=1}^{P^-} (b_k-a_s)\right )
\exp \left (\frac{1}{\hbar^2}\Bigl (
e^{\hbar (S^+-D(a_s))} + e^{\hbar (S^-+D(a_s))}\Bigr )F\right )=0,
\eeq
where the operators $S^{\pm}$ are
\beq\label{Spm}
S^+ = \sum_{i=1}^{P^+} D(a_i), \qquad
S^- = \sum_{k=1}^{P^-} D(b_k).
\eeq
In particular, equation (\ref{kp5a}) can be rewritten as
\beq\label{kp5b}
(a_1-a_2)e^{\Delta_{\hbar}(a_1)\Delta_{\hbar} (a_2)F} + 
(a_2-a_3)e^{\Delta_{\hbar} (a_2)\Delta_{\hbar}(a_3)F} +
(a_3-a_1)e^{\Delta_{\hbar}(a_3)\Delta_{\hbar}(a_1)F}=0
\eeq
(see also \cite{NZ15}).
Letting here $a_3\to \infty$, we obtain the equation
\beq\label{kp5d}
(a_1-a_2)e^{\Delta_{\hbar}(a_1)\Delta_{\hbar} (a_2)F}=
{\sf p}(a_1)-{\sf p}(a_2),
\eeq
where
$$
{\sf p}(z)=z-\Delta_{\hbar}(z)\p_{t_1}F.
$$
It is equivalent to (\ref{kp5b}), which, in its turn, is equivalent
to (\ref{kp5x}).

\subsubsection{The dispersionless limit}
\label{section:lessKP} 

As is well known, there exists a large class of solutions to the
$\hbar$-KP hierarchy for which the $F$-function is regular 
at $\hbar =0$ and admits an 
expansion in powers of $\hbar$ as $\hbar \to 0$ of the 
following form:
\beq\label{exp1}
F({\bf t}; \hbar )=F_0 ({\bf t}) + \hbar F_1 ({\bf t})
+ \hbar^2 F_2 ({\bf t}) + O(\hbar^3), \qquad \hbar \to 0.
\eeq
The dispersionless limit corresponds to $\hbar =0$, 
when only the leading
term 
$$
F_0 =F_0({\bf t})=
\lim\limits_{\hbar \to 0}
F({\bf t}; \hbar )
$$
of the series survives.
Noting that 
$
\displaystyle{
\lim_{\hbar \to 0}\Delta_{\hbar}(z)=D(z),}
$
we immediately obtain the dispersionless version of equation 
(\ref{kp5a}):
\beq\label{kp11}
(a_1-a_2)e^{D(a_1)D(a_2)F_0} + (a_2-a_3)e^{D(a_2)D(a_3)F_0} +
(a_3-a_1)e^{D(a_3)D(a_1)F_0}=0.
\eeq
In the general case, the $\hbar \to 0$ limit of (\ref{kp5x})
is as follows:
\beq\label{kp10}
\sum_{i=1}^{P^+}
\left (\prod_{k=1, \neq i}^{P^+} (a_k-a_i)e^{D(a_i)D(a_k)F_0}\right )^{-1}
\left (\prod_{j=1}^{P^-} (b_j-a_i)e^{D(a_i)D(b_j)F_0}\right )=0.
\eeq
Equation (\ref{kp11}) is the simplest non-trivial case 
$(P^+, P^-)=(3,0)$ of it.
Letting $a_3\to \infty$ in (\ref{kp11}), we have:
\beq\label{kp12}
(a_1-a_2)e^{D(a_1)D(a_2)F_0} =p(a_1)-p(a_2),
\eeq
where
\beq\label{kp13}
p(z)=z-D(z)\p_{t_1}F_0.
\eeq
In fact (\ref{kp12}) is equivalent to (\ref{kp11}): summing 
equations of the form (\ref{kp12}) written for the pairs of points
$\{a_1, a_2\}$, $\{a_2, a_3\}$ and
$\{a_3, a_1\}$, we get (\ref{kp11}). 

\begin{proposition}\label{proposition:dkp-eq}
Equation (\ref{kp12}) itself 
is equivalent to the whole hierarchy
(\ref{kp10}).
\end{proposition}

\noindent
{\it Proof.}
Plugging (\ref{kp12}) into (\ref{kp10}), we 
have the identity
$$
\sum_{i=1}^{P^+}
\frac{\prod\limits_{j=1}^{P^-} 
(p(b_j)-p(a_i))}{\prod\limits_{k=1, \neq i}^{P^+} (p(a_k)-p(a_i))}
=0
$$
which does hold for all $P^{+}$, $P^{-}$ such that 
$P^{+}-P^- \geq 2$ because its left-hand side is the sum of residues
of the rational function
$$
f(p)=(-1)^{P^+ -P^-}\, 
\frac{\prod\limits_{j=1}^{P^-}
(p-p(b_j))}{\prod\limits_{k=1}^{P^+}(p-p(a_k))}
$$
(for $P^{+}-P^- \geq 2$ there is no residue at $\infty$).
\square

\subsubsection{The $F_1$-function}
\label{section:F1}

Later, in the more complicated case of the 
large BKP hierarchy, we will need some information about 
the next-to-leading term of the series
(\ref{exp1}), $F_1$. 

\begin{proposition}\label{proposition:F1KP}
In the case of the KP hierarchy, the function $F_1$ satisfies
the homogeneous linear equation
\beq\label{exp3}
\sum_{s=1}^{P^+}A_s 
e^{(S^- - S^+ +D(a_s))D(a_s) F_0}(S^- \! - \! S^+ +D(a_s))D(a_s)F_1=0,
\eeq
where 
$$
A_s=
\left (\prod_{i=1, \neq s}^{P^+} (a_i-a_s)\right )^{-1}\!\!
\left ( \prod_{k=1}^{P^-} (b_k-a_s)\right )
$$
and
the operators $S^{\pm}$ are given in (\ref{Spm}).
\end{proposition}

\noindent
{\it Proof.}
To prove the  equation for $F_1$, we should
expand the general equation (\ref{kp5x}) up to the next-to-leading
order as $\hbar \to 0$. The procedure is straightforward and the
result is
\beq\label{exp2}
\begin{array}{l}
\displaystyle{
\sum_{s=1}^{P^+}A_s 
\exp \left [ (S^- \! -\!  S^+ \! +\! D(a_s))D(a_s) F_0 + 
\frac{\hbar}{2}R_s F_0 \right. }
\\ \\
\displaystyle{
\phantom{aaaaaaaaaaa}\left. \vphantom{\frac{\hbar}{2}}+\, \hbar 
(S^-\!  - \! S^+ \! +\! D(a_s))
D(a_s) F_1  + O(\hbar^2) \right ]=0,}
\end{array}
\eeq
where $R_s$ is the operator
$$
R_s=
(S^+ \! +\!  S^-)(D(a_s)+ S^-\! - \! S^+) D(a_s).
$$
Expanding (\ref{exp2}) up to terms of order $\hbar$, we get:
$$
\begin{array}{l}
\displaystyle{
\sum_{s=1}^{P^+}A_s 
e^{(S^- - S^+ +D(a_s))D(a_s) F_0}+ \frac{\hbar}{2}
\sum_{s=1}^{P^+}A_s 
e^{(S^- - S^+)D(a_s) F_0} R_s F_0}
\\ \\
\phantom{aaaaaaa}
\displaystyle{
+\, \hbar \sum_{s=1}^{P^+}A_s 
e^{(S^- - S^+ +D(a_s))D(a_s) F_0}(S^- - S^+ +D(a_s))D(a_s)F_1}=0.
\end{array}
$$
The first term is zero because it is the left-hand side of
equation (\ref{kp10}). The second term vanishes, too, because
it is the result of acting to the
first one  by the operator $S^+ + S^-$.  The rest is just the 
linear equation (\ref{exp2}).
\square

\begin{remark}
Let $v$ be any parameter on which a general solution $F_0$ 
to the dispersionless hierarchy depends: 
$F_0=F_0({\bf t}; v )$. Then the function
$F_1=\p_{v} F_0$ solves equation (\ref{exp3}) since its
left-hand side is just the $v$-derivative of (\ref{kp10}).
Presumably, it is a general solution for $F_1$, i.e., all solutions
of (\ref{exp3}) are of this form with a suitable parameter $v$.
\end{remark}

\subsection{The dmKP hierarchy}
\label{section:dmKP}

Next we consider the mKP hierarchy and its dispersionless limit (dmKP).
Our aim in this (still rather simple) example 
is to show how a less trivial rational curve emerges.

\subsubsection{The mKP hierarchy}
\label{section:mKP}

In the mKP hierarchy,
the independent variables are the same ${\bf t}=\{t_1, t_2, t_3, \ldots \}$
and a discrete variable $m\in \ZZ$.
The tau-function $\tau (m, {\bf t})$ satisfies the integral bilinear
relation
\beq\label{mkp1}
\oint_{C_{\infty}} \frac{dz}{z^2} \, z^{m-m'}
e^{\xi ({\bf t}-{\bf t'}, z)} 
\tau \Bigl (m-1, {\bf t}-[z^{-1}]\Bigr )
\tau \Bigl (m'+1, {\bf t'}+[z^{-1}]\Bigr ) \, =0,
\eeq
which is valid for all ${\bf t}$, ${\bf t}'$ and $m, m'\in \ZZ$ such that
$m-m'\geq 2$. The simplest solution is $\tau (m, {\bf t})=1$.

In the fermionic approach, the tau-function is defined as the
following expectation value:
\beq\label{taummkp1}
\tau (m, {\bf t})=
\lbr m\bigl | \, e^{J({\bf t})} g \bigr | \, m \rbr,
\eeq
where 
\beq\label{g01}
g=\exp \Bigl (\sum_{i,j\in \z} A_{ij} \psi_{i}\psi_j^*\Bigr )
\eeq
is a Clifford group element with zero charge. For the
simplest solution $g=1$.

The analog of the substitution (\ref{kp2}) in this case is
\beq\label{mkp2}
\left \{\begin{array}{l}
\displaystyle{
m-m'=P^+ -P^-,}
\\ \\
\displaystyle{
{\bf t}-{\bf t}' =\sum_{i=1}^{P^+}[a_i^{-1}]-
\sum_{k=1}^{P^-}[b_k^{-1}],}
\end{array}\right.
\eeq
where we again assume that all the points $a_i, b_j$ are 
distinct.
Then we have:
$$
z^{P^+ -P^-} e^{\xi ({\bf t}-{\bf t'}, z)} =
\prod_{i=1}^{P^+} \left (z^{-1}-a_i^{-1}\right )^{-1}
\prod_{j=1}^{P^-} \left (z^{-1}-b_j^{-1}\right )
$$
and the integral in (\ref{mkp1}) is reduced to sum of residues at
the simple poles $z=a_i$ (note that the substitution (\ref{mkp2}) 
implies that the residue at $\infty$ vanishes). The result is
\beq\label{mkp3}
\begin{array}{l}
\displaystyle{
\sum_{s=1}^{P^+} \prod_{i=1, \neq s}^{P^+} E^{-1}(a_s, a_i) \,
\prod_{k=1}^{P^-} E(a_s, b_k)\,
\tau \Bigl (m+P^+ \! -\! 1, {\bf t}+\sum_{i\neq s}^{P^+}
[a_i^{-1}]\Bigr )}
\\ \\
\phantom{aaaaaaaaaaaaaaaaaaaaaaaaa}\displaystyle{\times \, 
\tau \Bigl (m+P^- \! +\! 1, {\bf t}+[a_s^{-1}]+ \sum_{k=1}^{P^-}
[b_k^{-1}]\Bigr )=0},
\end{array}
\eeq
where the convenient short-hand notation
\beq\label{mkp4}
E(a,b)=a^{-1}-b^{-1}
\eeq
is used. 
Note that although this equation was derived under assumption that
all the points $a_i$ are distinct and finite, it still holds in  
degenerate cases, too, when some of the points merge or tend to infinity.
In such cases, some terms of the equation may
become singular, if one considers them separately. In the full
expression, the singularities 
can be resolved, and, as a result, 
derivatives of the tau-functions with respect to 
continuous times arise in this way.

If $P^+ - P^{-}=2$, then the $m$-variables 
in both tau-functions in (\ref{mkp3}) are the same, and so 
equation (\ref{mkp3}) for $\tau (m, {\bf t})$ coincides (for each $m$)
with the generating bilinear equation (\ref{kp5}) for the KP hierarchy.
The Hirota-Miwa equations specific for the mKP hierarchy are obtained 
if $P^+ - P^{-}>2$.

For $(P^+, P^-)=(3,1)$ we have the equation
\beq\label{mkp5}
\begin{array}{l}
\phantom{+\,}E^{-1}(a_1, a_2)E^{-1}(a_1, a_3)E(a_1, b_1)
\tau \Bigl (m, {\bf t}+[a_2^{-1}]+[a_3^{-1}]\Bigr )
\tau \Bigl (m, {\bf t}+[a_1^{-1}]+[b_1^{-1}]\Bigr )
\\ \\
+\, E^{-1}(a_2, a_1)E^{-1}(a_2, a_3)E(a_2, b_1)
\tau \Bigl (m, {\bf t}+[a_1^{-1}]+[a_3^{-1}]\Bigr )
\tau \Bigl (m, {\bf t}+[a_2^{-1}]+[b_1^{-1}]\Bigr )
\\ \\
+\,E^{-1}(a_3, a_1)E^{-1}(a_3, a_2)E(a_3, b_1)
\tau \Bigl (m, {\bf t}+[a_1^{-1}]+[a_2^{-1}]\Bigr )
\tau \Bigl (m, {\bf t}+[a_3^{-1}]+[b_1^{-1}]\Bigr )\, = \, 0
\end{array}
\eeq
(which is equivalent to (\ref{kp5a}) and to the whole KP hierarchy),
while for $(P^+, P^-)=(3,0)$ the equation is
\beq\label{mkp6}
\begin{array}{l}
\phantom{+\,}E^{-1}(a_1, a_2)E^{-1}(a_1, a_3)
\tau \Bigl (m+1, {\bf t}+[a_2^{-1}]+[a_3^{-1}]\Bigr )
\tau \Bigl (m, {\bf t}+[a_1^{-1}]\Bigr )
\\ \\
+\, E^{-1}(a_2, a_1)E^{-1}(a_2, a_3)
\tau \Bigl (m+1, {\bf t}+[a_1^{-1}]+[a_3^{-1}]\Bigr )
\tau \Bigl (m, {\bf t}+[a_2^{-1}]\Bigr )
\\ \\
+\,E^{-1}(a_3, a_1)E^{-1}(a_3, a_2)
\tau \Bigl (m+1, {\bf t}+[a_1^{-1}]+[a_2^{-1}]\Bigr )
\tau \Bigl (m, {\bf t}+[a_3^{-1}]\Bigr )\, = \, 0.
\end{array}
\eeq

\begin{theorem}\cite{Shigyo13}\label{theorem:mkp}
Equation (\ref{mkp6}) is equivalent to the whole mKP hierarchy
defined by (\ref{mkp1}).
\end{theorem}

\noindent
The direct proof is given in \cite{Shigyo13}.

\subsubsection{The dispersionless limit}
\label{section:ddmKP}

As before, the passage to the dispersionless version of the 
hierarchy consists in the re-scaling
$t_{k}\to t_{k}/\hbar$ for all $k \geq 1$,
$m\to \tilde t_0/\hbar$ and letting $\hbar \to 0$: 
$$
F_0=F(\tilde t_0, {\bf t})=
\lim\limits_{\hbar \to 0} \Bigl (
\hbar^2 \log 
\tau \Bigl (\hbar^{-1}\tilde t_0, \hbar^{-1}{\bf t}  \Bigr )\Bigr ).
$$
(The tilde above $t_0$ is introduced here in order to 
distinguish it from similar variables that will appear later.)
Note that in the dispersionless limit the variable $\tilde t_0$ becomes 
continuous, and the former difference equations containing it become
differential\footnote{
To simplify the notation, in what follows, 
we will denote the extended set of times $\{\tilde t_0, {\bf t}\}$ simply 
as ${\bf t}$.}. 
It is convenient to introduce the differential operator
\beq\label{mkp7}
\tilde \nabla (z)=\p_{\tilde t_0}+D(z)=
\p_{\tilde t_0}+\sum_{k\geq 1}\frac{z^{-k}}{k}\,
\p_{t_{k}},
\eeq
then 
\beq\label{mkp8}
\tau \Bigl (\tilde t_0/\hbar \pm 1, {\bf t}/\hbar \pm [a^{-1}]\Bigr )=
\exp \left (\frac{1}{\hbar^2}\, e^{\pm \hbar \tilde \nabla (a)}F 
({\bf t}; \hbar )\right ).
\eeq
The limit of equation ({\ref{mkp3}) is
\beq\label{mkp10}
\sum_{s=1}^{P^+}
\left (\prod_{i=1, \neq s}^{P^+} E(a_s, a_i)
e^{\tilde \nabla (a_i)\tilde \nabla (a_s)F}\right )^{-1}
\left (\prod_{k=1}^{P^-}E(a_s, b_k)e^{\tilde \nabla (a_s) 
\tilde \nabla (b_k)F}\right )=0.
\eeq
(Here and below in this section we, for simplicity 
of the notation, write $F$ instead of $F_0$. This can not lead
to a misunderstanding since $F_1$ will not appear in this section.)
In particular, the limiting forms of equations (\ref{mkp5}), (\ref{mkp6})
are:
\beq\label{mkp5a}
\begin{array}{l}
E(a_1, a_2)E(a_3, b_1) e^{\tilde \nabla (a_1)\tilde \nabla (a_2)F
+\tilde \nabla (a_3)\tilde \nabla (b_1)F}
+ E(a_2, a_3)E(a_1, b_1) e^{\tilde \nabla (a_2)\tilde \nabla (a_3)F
+\tilde \nabla (a_1)\tilde \nabla (b_1)F}
\\ \\
\phantom{aaaaaaaaaaaaaaaaaaaaa}
+ \, E(a_3, a_1)E(a_2, b_1) e^{\tilde \nabla (a_1)\tilde \nabla (a_3)F
+\tilde \nabla (a_2)\tilde \nabla (b_1)F}=0,
\end{array}
\eeq
\beq\label{mkp6a}
\begin{array}{l}
E(a_1, a_2)e^{\tilde \nabla (a_1)\tilde \nabla (a_2)F}
+ E(a_2, a_3) e^{\tilde \nabla (a_2)\tilde \nabla (a_3)F}
+ E(a_3, a_1)e^{\tilde \nabla (a_1)\tilde \nabla (a_3)F}=0,
\end{array}
\eeq
These equations hold for all $a_1, a_2, a_3, b_1 \in \CC$.

\subsubsection{The dynamical curve: first appearance}

To represent equations (\ref{mkp5a}), (\ref{mkp6a}) 
in a more visual form, we introduce
the notation
\beq\label{mkp11}
\tilde g(a,b) =E(a,b)e^{\tilde \nabla (a)\tilde \nabla (b)F},
\quad
\tilde w(z)=z^{-1}e^{\tilde \nabla (z)\p_{t_0}F}=\tilde g(z, \infty ).
\eeq
Letting $b\to \infty$ in (\ref{mkp5a}) and putting $a_1=a$, $a_2=b$,
$a_3=c$, we write the system of equations as
\beq\label{mkp12}
\left \{
\begin{array}{l}
\tilde g(a,b)\tilde w(c)+\tilde g(b,c)\tilde w(a)+\tilde 
g(c,a)\tilde w(b)=0,
\\ \\
\tilde g(a,b)+\tilde g(b,c)+\tilde g(c,a)=0.
\end{array}
\right.
\eeq
The limit $c\to \infty$ in the first equation yields:
\beq\label{mkp13}
\left \{
\begin{array}{l}
\tilde g(a,b)=\tilde R_0^{-1}
\tilde w(a)\tilde w(b)\Bigl (\tilde p(b)-\tilde  
p(a)\Bigr ), 
\\ \\
\tilde g(a,b)=\tilde w(a)-\tilde w(b),
\end{array}
\right.
\eeq
where $\tilde R_0=e^{\tilde \p_0^2 F}$, 
$\tilde \p_0 =\p_{\tilde t_0}$ and
\beq\label{mkp14}
\tilde p(z)=z-\tilde \nabla (z)\p_{t_1}F.
\eeq

\begin{remark}
The first equation in this system is equivalent
to the dKP hierarchy. Indeed, substituting $\tilde w(z)$ from
(\ref{mkp11}), one can see that terms with
$\tilde t_0$-derivatives cancel and 
this equation becomes the same as
(\ref{kp12}). This means that the $F$-function with any fixed 
$\tilde t_0$ is a solution of the dKP hierarchy 
(as a function of all other times). Moreover, the second
equation in (\ref{mkp13}) can be equivalently rewritten
as
\beq\label{mkp13d}
(a-b)e^{D(a)D(b)F}= e^{\tilde \p_0^2F}\Bigl (\tilde w^{-1}(a)-
\tilde w^{-1}(b)\Bigr ).
\eeq
Summing such equations for the pairs $(a,b)$, $(b,c)$, $(c,a)$,
one excludes the $\tilde t_0$-derivatives and thus 
obtains the dKP hierarchy
in the form (\ref{kp11}).
\end{remark}

Differential equations of the dmKP hierarchy are obtained from
expansion of the second equation in (\ref{mkp13}) in inverse powers
of $a,b \to \infty$. For example, the simplest such equation reads
\beq\label{L7a}
F_{\tilde 02}-2F_{11}-F_{\tilde 01}^2=0,
\eeq
where we denote $F_{mn}\equiv \p_{t_m}\p_{t_n}F$ (in particular,
$F_{\tilde 01}=\p_{\tilde t_0 t_1}F$, and similarly for $F_{\tilde 02}$).

Now we are ready to show how the dynamical curve appears in this
simple example.
Equating the right-hand sides of equations (\ref{mkp13}), we have:
$$
\tilde p(a)-\tilde R_0 \tilde w^{-1}(a)=\tilde 
p(b)-\tilde R_0 \tilde w^{-1}(b)
$$
from which it follows that $p(z)-\tilde R_0 \tilde w^{-1}(z)\equiv C$ does 
not depend on $z$. Letting $z\to \infty$, 
we conclude that $C=0$, i.e.,
it holds
\beq\label{mkp15}
\tilde w(z) \tilde p(z)-\tilde R_0=0.
\eeq
Set $x=\tilde w(z) $, $y=\tilde p(z)$, then the equation is
$P(x,y)=0$ with the polynomial $P(x,y)=xy-\tilde R_0$ of degree 2.
This polynomial equation
defines a rational smooth algebraic curve. 
The curve is equipped with 
the local parameter $z^{-1}$ in a neighborhood of infinity.
This is the simplest example of how dynamical curves 
arise in dispersionless hierarchies.
In this particular case, there is no evident benefit from the curve.
Its key role will be revealed later in less trivial examples.

\section{Multi-component dKP and dmKP:
rational dynamical curve and 
tri\-go\-no\-met\-ric uniformization}

\subsection{The multi-component KP and mKP hierarchies}
\label{section:mmKP}

Our next example is the multi-component mKP hierarchy introduced
in \cite{Z19}\footnote{Here we deal with a slightly more general 
version of it.}.
In the $N$-component mKP hierarchy the independent variables are
$N$ infinite sets 
\beq\label{mmkp1}
{\bf t}_{\alpha}=\{t_{\alpha , 1}, t_{\alpha , 2}, t_{\alpha , 3}, 
\ldots \, \}, \quad \alpha =1, \ldots , N
\eeq
of continuous ``times'' (in general complex numbers) 
and two finite sets of integer variables
$$
{\bf m}=\{m_1, \ldots , m_N\}, \quad
{\bf n}=\{n_1, \ldots , n_N\}, \quad m_{\alpha},  n_{\alpha}
\in \ZZ 
$$
such that
\beq\label{mmkp1a}
|{\bf n}|\equiv \sum_{\alpha =1}^N n_{\alpha} =0.
\eeq
In what follows we abbreviate the full set of continuous times as
$
{\bf t}=\{{\bf t}_1, {\bf t}_2, \ldots , {\bf t}_N\}.
$

In the fermionic approach, the tau-function is defined as the
following expectation value:
\beq\label{taummkp}
\tau ({\bf n}, {\bf m}, {\bf t})=
\lbr {\bf n}\! +\! {\bf m}\bigl | \, e^{J({\bf t})} 
g \bigr | \, {\bf m}\rbr,
\eeq
where 
\beq\label{g0}
g=\exp \Bigl (\sum_{i,j\in \z} \sum_{\alpha , \beta} 
A_{ij}^{(\alpha \beta )}
\psi_{i}^{(\alpha )}\psi_j^{*(\beta )}\Bigr )
\eeq
is a neutral Clifford group element.

The integral bilinear relation for the tau-function has the form
\beq\label{mmkp2}
\begin{array}{l}
\displaystyle{
\sum_{\gamma =1}^N\epsilon_{\gamma}({\bf n}\! +\! {\bf m})
\epsilon_{\gamma}({\bf n'}\! +\! {\bf m'})
\oint_{C_{\infty}} \frac{dz}{z^2}\, z^{n_{\gamma}-n_{\gamma}' +
m_{\gamma}-m_{\gamma}'}\,
e^{\xi ({\bf t}_{\gamma}-{\bf t'}_{\gamma}, z)}}
\\ \\
\phantom{aaaaaaaaaaaaaaa}\displaystyle{
\times \, \tau \Bigl ({\bf n}\! -\! {\bf e}_{\gamma}, 
{\bf m}, {\bf t}-[z^{-1}]_{\gamma} \Bigr )
\tau \Bigl (
{\bf n'}\! +\! {\bf e}_{\gamma}, {\bf m'}, 
{\bf t'}+[z^{-1}]_{\gamma} \Bigr )=0.}
\end{array}
\eeq
Hereafter, ${\bf e}_{\gamma}$ is the $N$-component vector whose 
$\gamma$th component is 1 and all other are 0, 
$\epsilon_{\gamma}({\bf n})$ is the following sign factor:
\beq\label{epsilon}
\epsilon_{\gamma}({\bf n})=
(-1)^{n_{\gamma +1}+\ldots +n_N}.
\eeq
The notation ${\bf t}\pm [z^{-1}]_{\gamma}$ means the set of times
in which the times ${\bf t}_{\gamma}$ are shifted in the standard way as
$\{ t_{\gamma , 1}\pm z^{-1}, \,  t_{\gamma , 2}\pm \frac{1}{2}z^{-1},\,
t_{\gamma , 3}\pm \frac{1}{3}z^{-1}, \, \ldots \}$, and the other times
are intact. The equation (\ref{mmkp2}) is valid for all 
${\bf t}, \, {\bf t'}$ and integer ${\bf n}, \, {\bf n'}$,  
${\bf m}, \, {\bf m'}$ such that 
\beq\label{mmkp3}
|\, {\bf n}\, |=1,
\quad  |\, {\bf n'}\, |=-1
\quad \mbox{and \, $m_{\gamma}\geq m_{\gamma}'$ for all $\gamma$}.
\eeq
Note that at ${\bf m}={\bf m'}$ equation (\ref{mmkp2}) states that
$\tau ({\bf n}, {\bf m}, {\bf t})$ as a function
of ${\bf n}, {\bf t}$ is, at any fixed ${\bf m}$, a tau-function
of the $N$-component KP hierarchy.

The transition to Miwa variables (an analog of the substitution (\ref{kp2}))
is as follows:
\beq\label{m1}
\left \{ \begin{array}{l}
\displaystyle{
{\bf n}-{\bf n}' +{\bf m}-{\bf m}'
=\sum_{i=1}^{P^+}{\bf e}_{\alpha_i}-
\sum_{j=1}^{P^-}{\bf e}_{\beta_j},  
}
\\ 
\displaystyle{
{\bf m}-{\bf m'}=\sum_{k=1}^{Q}{\bf e}_{\gamma_k}, 
}
\\ 
\displaystyle{
{\bf t}-{\bf t}' =\sum_{i=1}^{P^+}[a_i^{-1}]_{\alpha_i}-
\sum_{j=1}^{P^-}[b_j^{-1}]_{\beta_j}.}
\end{array} \right.
\eeq
Here $\alpha_i, \beta_j, \gamma_k$ are arbitrary indices
from the set $\{1, \ldots , N\}$ (they may enter with multiplicities, i.e.,
the cases when, say, $\alpha_i =\alpha_j$ for $i\neq j$ are allowed), and
$a_i , b_k \in \CC$ are arbitrary 
parameters belonging to a neighborhood of infinity 
(and, again, we allow the
cases when some of them coincide). If some of these points tend to infinity,
a nonzero residue at infinity arises. As is easy to see, the two procedures
(taking residue at some $a_i$ and tending $a_i$ to infinity,  
obtaining a non-zero residue there) commute.
Therefore, we can start, without any loss of generality, with the case
when all complex variables are distinct and 
finite and, if necessary, tend 
some of them to infinity afterwards. The numbers $P^+, P^-$ and $Q$ 
are arbitrary non-negative integer numbers such that
\beq\label{mmkp4}
P^+-P^-=2+Q.
\eeq
The latter condition guarantees that if all the points $a_i$ are
finite, after the substitution of (\ref{m1}) into 
(\ref{mmkp2}) any contribution from infinity vanishes. 
We will refer to the case when all $a_i, b_j$ are distinct and
condition (\ref{mmkp4}) is satisfied as {\it non-degenerate case}, and
to the corresponding Hirota-Miwa equations as {\it non-degenerate} ones.
The case $Q=0$ corresponds to the $N$-component KP hierarchy.

Applying the residue calculus to (\ref{mmkp2}) 
after the substitution (\ref{m1}), one obtains the 
following general equation of the 
Hirota-Miwa type:
\beq\label{mmkp5}
\begin{array}{c}
\displaystyle{
\sum_{s=1}^{P^+} \prod_{i=1, \neq s}^{P^+} E^{-1}_{\alpha_i \alpha_s}(a_s, a_i)
\prod_{k=1}^{P^-} E_{\beta_k \alpha_s}(a_s, b_k)\,
\tau \Bigl ({\bf n}\! -\! {\bf e}_{\alpha_s},
{\bf m}+\sum_{j=1}^Q {\bf e}_{\gamma_j}, {\bf t}-[a_s^{-1}]_{\alpha_s}\Bigr )
}
\\ \\
\displaystyle{
\times \, 
\tau \Bigl ({\bf n}\! -\! \sum_{i\neq s}^{P^+}
{\bf e}_{\alpha_i}+ \sum_{j=1}^{P^-}{\bf e}_{\beta_j}
+\sum_{j=1}^Q {\bf e}_{\gamma_j}, {\bf m}, 
{\bf t}-\sum_{i\neq s}^{P^+}[a_i^{-1}]_{\alpha_i}+
\sum_{j=1}^{P^-}[b_j^{-1}]_{\beta_j}\Bigr )=0,}
\end{array}
\eeq
where
\beq\label{E1}
E_{\alpha \beta}(a,b)=\epsilon_{\alpha \beta} 
\Bigl (a^{-1}-b^{-1}\Bigr )^{\delta_{\alpha \beta}},
\eeq
and $\epsilon_{\alpha \beta}$ is a sign factor equal to 1 if $\alpha \leq \beta$
and $-1$ otherwise. Note the obvious property of this function:
$$
E_{\alpha \beta}(a,b)=-E_{ \beta \alpha}(b,a).
$$
Equation (\ref{mmkp5}) contains $P^+$ 
bilinear terms, each of which is product of two tau-functions with
various shifts of the arguments. The coefficients are rational functions
of $a_i , b_k$. We will call it
the (non-degenerate) $(P^+ +P^-)$-point relation, according to the total
number of the points. Note that the number of terms in non-degenerate
relations may be either equal or less than the number of points. 
Besides, in the non-degenerate 
case, when all the points $a_i$ are distinct, the relation does not contain any derivatives of the tau-function with respect to the 
continuous times.

\begin{remark}
The general relation (\ref{mmkp5}) 
still holds in 
all degenerate cases, when some of the points merge or tend to infinity.
In such cases, some terms of equation (\ref{mmkp5}) 
become singular, if one considers them separately. However, in the full
expression the singularities 
can be resolved, and, as a result, 
derivatives of the tau-functions with respect to 
continuous times arise in this way.
\end{remark}

The simplest possible case $(P^+, P^-, Q)=(2,0,0)$ is trivial: as it can be
easily seen, (\ref{mmkp5}) becomes an identity (of the form $0=0$). 
The simplest meaningful 
cases are:
\begin{itemize}
\item[I)] $\; (P^+, P^-, Q)=(3,1,0)$:

\beq\label{caseI}
\begin{array}{l}
{\bf n}- {\bf n'}={\bf e}_{\alpha_1}+{\bf e}_{\alpha_2}+{\bf e}_{\alpha_3}-
{\bf e}_{\beta_1}, \; {\bf m}={\bf m'},
\\ \\
{\bf t}- {\bf t'}=[a_1^{-1}]_{\alpha_1}+ [a_2^{-1}]_{\alpha_2}+
[a_3^{-1}]_{\alpha_3}-[b_1^{-1}]_{\beta_1},
\end{array}
\eeq

\item[II)] $\; (P^+, P^-, Q)=(3,0,1)$:
\beq\label{caseII}
\begin{array}{l}
{\bf n}- {\bf n'}={\bf e}_{\alpha_1}+{\bf e}_{\alpha_2}+{\bf e}_{\alpha_3},
 \; {\bf m}-{\bf m'}={\bf e}_{\beta_1},
\\ \\
{\bf t}- {\bf t'}=[a_1^{-1}]_{\alpha_1}+ [a_2^{-1}]_{\alpha_2}+
[a_3^{-1}]_{\alpha_3}.
\end{array}
\eeq
\end{itemize}

\noindent
Case I) corresponds to the Hirota-Miwa equation for the multi-component
KP hierarchy. Here we will not write down these equations 
explicitly because we are mostly interested in their dispersionless
versions.

\subsection{Multi-component dKP and dmKP hierarchies}
\label{section:dmmKP}

The next step is to pass to the dispersionless limit.
Re-scaling the variables as
$$
t_{\alpha , k}\to \frac{t_{\alpha , k}}{\hbar}, 
\quad n_{\alpha}\to \frac{t_{\alpha , 0}}{\hbar}, \quad
m_{\alpha}\to \frac{\tilde t_{\alpha , 0}}{\hbar},
$$
we introduce the differential operators
\beq\label{mmkp11}
\nabla_{\alpha}(z)=\p_{t_{\alpha ,0}}+\sum_{k\geq 1}\frac{z^{-k}}{k}\,
\p_{t_{\alpha ,k}}, \quad 
\p_{\alpha}=\p_{t_{\alpha ,0}}, \quad
\tilde \p_{\alpha}=\p_{\tilde t_{\alpha ,0}}.
\eeq
The function
\beq\label{mmkp11a}
F({\bf t}_0, \tilde {\bf t}_0, {\bf t})=
\lim_{\hbar \to 0}\left [\vphantom{\frac{a}{b}}
\, \hbar^2 \log \tau \Bigl 
(\hbar^{-1}{\bf t}_0, \hbar^{-1} \tilde {\bf t}_0, \hbar^{-1}{\bf t}\Bigr )\right ]
\eeq
plays the role of the tau-function in the limit
$\hbar \to 0$. The limiting form of the general Hirota-Miwa
equation (\ref{mmkp5}) can be obtained in basically the same way 
as this was done in Sections \ref{section:dKP}, \ref{section:dmKP}, so
we omit the details. The result is:
\beq\label{mmkp12}
\begin{array}{l}
\displaystyle{
\sum_{s=1}^{P^+}
\left (\prod_{i=1, \neq s}^{P^+} E_{\alpha_i \alpha_s}^{-1}(a_s, a_i)
e^{-\nabla_{\alpha_s} (a_s)\nabla_{\alpha_i} (a_i)F}\right )
\left (\prod_{k=1}^{P^-}E_{\beta_k \alpha_s}(a_s, b_k)
e^{\nabla_{\alpha_s} (a_s)\nabla_{\beta_k} (b_k)F}\right )}
\\ \\
\phantom{aaaaaaaaaaaaaaaaaaaaaaaaaaaa}\displaystyle{
\times \, \prod_{j=1}^Q e^{\nabla_{\alpha_s} (a_s)(\p_{\gamma_j} -
\tilde \p_{\gamma_j})F}=0.}
\end{array}
\eeq
In case I)  (with $b_1=\infty$) this general 
relation converts (after
extracting a common multiplier) into the equation
\beq\label{caseI1}
\begin{array}{l}
\;\, \epsilon_{\beta_1 \alpha_1} E_{\alpha_2 \alpha_3}(a_3, a_2)
a_1^{-\delta_{\alpha_1 \beta_1}} \, e^{\nabla_{\alpha_2}(a_2)
\nabla_{\alpha_3}(a_3)F + \nabla_{\alpha_1}(a_1)\p_{\beta_1}F}
\\ \\
+\, \epsilon_{\beta_1 \alpha_2} E_{\alpha_3 \alpha_1}(a_1, a_3)
a_2^{-\delta_{\alpha_2 \beta_1}} \, e^{\nabla_{\alpha_3}(a_3)
\nabla_{\alpha_1}(a_1)F + \nabla_{\alpha_2}(a_1)\p_{\beta_1}F}
\\ \\
+\, \epsilon_{\beta_1 \alpha_3} E_{\alpha_1 \alpha_2}(a_2, a_1)
a_3^{-\delta_{\alpha_3 \beta_1}} \, e^{\nabla_{\alpha_1}(a_1)
\nabla_{\alpha_2}(a_2)F + \nabla_{\alpha_3}(a_1)\p_{\beta_1}F}=0,
\end{array}
\eeq
which is the simplest non-trivial Hirota-Miwa equation of the
dKP hierarchy.
In case II) the equation is
\beq\label{caseII1}
\begin{array}{l}
\;\, E_{\alpha_2 \alpha_3}(a_3, a_2)
\, e^{\nabla_{\alpha_2}(a_2)
\nabla_{\alpha_3}(a_3)F + \nabla_{\alpha_1}(a_1)(\p_{\beta_1}-
\tilde \p_{\beta_1})F}
\\ \\
+\,E_{\alpha_3 \alpha_1}(a_1, a_3)
 \, e^{\nabla_{\alpha_3}(a_3)
\nabla_{\alpha_1}(a_1)F + \nabla_{\alpha_2}(a_1)(\p_{\beta_1}-\tilde 
\p_{\beta_1})F}
\\ \\
+\, E_{\alpha_1 \alpha_2}(a_2, a_1)
 \, e^{\nabla_{\alpha_1}(a_1)
\nabla_{\alpha_2}(a_2)F + \nabla_{\alpha_3}(a_1)(\p_{\beta_1}-
\tilde \p_{\beta_1})F}=0.
\end{array}
\eeq
These are 3-point non-degenerate relations. Our next task is to
derive from them (degenerate) 2-point relations by considering different
ways of degeneration. 

For each of the equations (\ref{caseI1}),  
(\ref{caseII1}) four essentially different types of 2-point degeneration 
are possible. In all of them we put
$$
a_1=a, \quad a_2=b, \quad a_3 = \infty ,
$$
and after that the possible types of degeneration are obtained by
different choices of the $\alpha_i$'s and $\beta_1$.
They are:
\begin{itemize}
\item[1)] 
$\alpha_1=\alpha_2=\alpha_3=\beta_1=\alpha$; 
\item[2)] $\alpha_1=\alpha_2=\alpha_3=\alpha$, 
$\beta_1=\beta \neq \alpha$;
\item[3)] $\alpha_1=\alpha$,
$\alpha_2=\alpha_3=\beta_1 =\beta \neq \alpha$;
\item[4)] $\alpha_1=\alpha_3=\alpha$,
$\alpha_2=\beta_1 =\beta \neq \alpha$.
\end{itemize}
The corresponding equations obtained from (\ref{caseI1}) 
are as follows:

\beq\label{mmkp16-1}
\hspace{-3.5cm}\begin{array}{ll}
\mbox{I$\vphantom{a}_{1}$:}&
(a^{-1}-b^{-1})e^{\nabla_{\alpha}(a)\nabla_{\alpha}(b)F
+\p_{\alpha}^2 F}=
e^{\nabla_{\alpha}(a)\p_{\alpha}F +\nabla_{\alpha}(b)\p_{\alpha}F }
\Bigl (a^{-1}-b^{-1}
\\ &\\
& \phantom{aaaaaaaaaaaa}+\, 
(ab)^{-1} (\nabla_{\alpha}(a)-\nabla_{\alpha}(b))
\p_{t_{\alpha , 1}}F\Bigr ),
\end{array}
\eeq

\beq\label{mmkp16-2}
\hspace{-0.2cm}\begin{array}{ll}
\mbox{I$\vphantom{a}_{2}$:}&
(a^{-1}-b^{-1})e^{\nabla_{\alpha}(a)\nabla_{\alpha}(b)F +
\p_{\alpha}\p_{\beta}F}=
a^{-1}e^{\nabla_{\alpha}(a)\p_{\alpha}F +\nabla_{\alpha}(b)\p_{\beta}F }
-b^{-1}e^{\nabla_{\alpha}(b)\p_{\alpha}F +\nabla_{\alpha}(a)\p_{\beta}F }
\\ & \\
&\mbox{(here $\alpha , \beta$ can be arbitrary, including the 
case $\beta =\alpha$)},
\end{array}
\eeq

\beq\label{mmkp16-3}
\begin{array}{ll}
\mbox{I$\vphantom{a}_{3}$:}&
e^{\nabla_{\alpha}(a)\nabla_{\beta}(b)F}=b^{-1} e^{-\p_{\beta}^2 F +
\nabla_{\beta}(b)\p_{\beta}F +\nabla_{\alpha}(a)\p_{\beta} F}
\Bigl (b-\nabla_{\beta}(b) \p_{t_{\beta , 1}}F +
\nabla_{\alpha}(a) \p_{t_{\beta , 1}}F\Bigr )
\\ & \\
&\mbox{(here $\alpha \neq \beta$),}
\end{array}
\eeq

\beq\label{mmkp16-4}
\hspace{-1.2cm}\begin{array}{ll}
\mbox{I$\vphantom{a}_{4}$:}&
e^{\nabla_{\alpha}(a)\nabla_{\beta}(b)F+\p_{\alpha}\p_{\beta}F}=
e^{\nabla_{\beta}(b)\p_{\alpha}F +\nabla_{\alpha}(a)\p_{\beta} F}+
(ab)^{-1}
e^{\nabla_{\alpha}(a)\p_{\alpha}F +\nabla_{\beta}(b)\p_{\beta} F}
\\ & \\
&\mbox{(here $\alpha \neq \beta$)}.
\end{array}
\eeq

\noindent
The equations obtained from (\ref{caseII1}) 
are as follows:

\beq\label{mmkp17-12}
\begin{array}{ll}
\mbox{II$\vphantom{a}_{1,2}$:}&
(a^{-1}-b^{-1})e^{\nabla_{\alpha}(a)\nabla_{\alpha}(b)F +
\p_{\alpha}\p_{\beta}F}=
a^{-1}e^{\nabla_{\alpha}(a)\p_{\alpha}F +\nabla_{\alpha}(b)
(\p_{\beta}-\tilde \p_{\beta})F }\phantom{aaaaaaaaaaaa}
\\ & \\
&\phantom{aaaaaaaaaaaaaaaaaaaaaaaaaaaaaa}
-\, b^{-1}e^{\nabla_{\alpha}(b)\p_{\alpha}F +\nabla_{\alpha}(a)
(\p_{\beta}-\tilde \p_{\beta})F }
\\ & \\
&\mbox{(here $\alpha , \beta$ can be arbitrary, including the 
case $\beta =\alpha$)},
\end{array}
\eeq

\beq\label{mmkp17-3}
\begin{array}{ll}
\mbox{II$\vphantom{a}_{3}$:}&
-b^{-1}
e^{\nabla_{\beta}(b)\p_{\beta}F+ \nabla_{\alpha}(a)
(\p_{\beta}-\tilde \p_{\beta})F}+
\epsilon_{\beta \alpha}
e^{\nabla_{\alpha}(a)\p_{\beta}F+ \nabla_{\beta}(b)
(\p_{\beta}-\tilde \p_{\beta})F}\phantom{aaaaaaaaaaaa}
\\ & \\
& \phantom{aaaaaaaaaaaaaaaaaaaaaaaaaaaaaa}+\, \epsilon_{\alpha \beta}
e^{\nabla_{\alpha}(a)\nabla_{\beta}(b)F +\p_{\beta}
(\p_{\beta}-\tilde \p_{\beta})F}=0,
\\ & \\
&\mbox{(here $\alpha \neq \beta$),}
\end{array}
\eeq

\beq\label{mmkp17-4}
\begin{array}{ll}
\mbox{II$\vphantom{a}_{4}$:}&
\epsilon_{\beta \alpha}
e^{\nabla_{\alpha}(a)\nabla_{\beta}(b)F +\p_{\alpha}(\p_{\beta}-
\tilde \p_{\beta})F}
=a^{-1} e^{\nabla_{\alpha}(a)\p_{\alpha}F +
\nabla_{\beta}(b) (\p_{\beta}-\tilde \p_{\beta})F}\phantom{aaaaaaaaaaaa}
\\ & \\
& \phantom{aaaaaaaaaaaaaaaaaaaaaaaaaaaaaa}-\, \epsilon_{\alpha \beta}
e^{\nabla_{\beta}(b)\p_{\alpha}F +
\nabla_{\alpha}(a) (\p_{\beta}-\tilde \p_{\beta})F}
\\ & \\
&\mbox{(here $\alpha \neq \beta$)}.
\end{array}
\eeq

In order to write down these equations in a more compact 
and suggestive form,
we introduce the differential operator
\beq\label{mmkp17}
\bar \p_{\alpha}=\p_{\alpha}-\tilde \p_{\alpha}
\eeq
and the following functions:
\beq\label{mmkp15}
\begin{array}{ll}
w_{\alpha}(z)=z^{-1}e^{\nabla_{\alpha}(z)\p_{\alpha}F}, \quad
& w_{\alpha \beta}(z)=e^{\nabla_{\alpha}(z)\p_{\beta}F}
\; \; \; (\alpha \neq \beta),
\\ & \\
\bar w_{\alpha}(z)=e^{\nabla_{\alpha}(z)\bar \p_{\alpha}F},
\quad
& \bar w_{\alpha \beta}(z)=e^{\nabla_{\alpha}(z)\bar \p_{\beta}F} 
\; \; \; (\alpha \neq \beta).
\end{array}
\eeq
The function
\beq\label{mmkp18}
v_{\alpha \beta}(z)=\frac{w_{\alpha \beta}(z)}{w_{\alpha}(z)},
\quad \alpha \neq \beta ,
\eeq
will be useful, too. We also need the functions
\beq\label{mmkp20}
\begin{array}{ll}
p_{\alpha}(z)=z-\nabla_{\alpha}(z)\p_{t_{\alpha , 1}}F, \quad
&
p_{\alpha \beta}(z)=-\nabla_{\alpha}(z)\p_{t_{\beta , 1}}F
\; \; \; (\alpha \neq \beta).
\end{array}
\eeq
Besides, the following notation for certain $z$-independent 
quantities will be used:
\beq\label{mmkp19}
\begin{array}{ll}
R_{\alpha}=e^{\p^2_{\alpha}F}, \quad
& R_{\alpha \beta}=e^{\p_{\alpha}\p_{\beta}F}=w_{\alpha \beta}(\infty ),
\; \; \; (\alpha \neq \beta),
\\ & \\
\bar R_{\alpha}=e^{\p_{\alpha}\bar \p_{\alpha}F},
\quad
& \bar R_{\alpha \beta}(z)=e^{\p_{\alpha}\bar \p_{\beta}F}
=\bar w_{\alpha \beta}(\infty ) 
\; \; \; (\alpha \neq \beta).
\end{array}
\eeq

In this notation, equations (\ref{mmkp16-1})--(\ref{mmkp16-4}) 
acquire the following form:

\beq\label{mmkp16-1a}
\hspace{-1.5cm}\begin{array}{ll}
\mbox{I$\vphantom{a}_{1}$:}&
(a^{-1}-b^{-1})e^{\nabla_{\alpha}(a)\nabla_{\alpha}(b)F}=R_{\alpha}^{-1}
w_{\alpha}(a)w_{\alpha}(b)
\Bigl (p_{\alpha}(b)-p_{\alpha}(a)\Bigr ),
\phantom{\quad \alpha \neq \beta}
\end{array}
\eeq

\beq\label{mmkp16-2a}
\hspace{2mm}\begin{array}{ll}
\mbox{I$\vphantom{a}_{2}$:}&
(a^{-1}-b^{-1})e^{\nabla_{\alpha}(a)\nabla_{\alpha}(b)F}=R_{\alpha \beta}^{-1}
\Bigl (w_{\alpha}(a) w_{\alpha \beta}(b)-w_{\alpha}(b) w_{\alpha \beta}(a)
\Bigr ), \quad \alpha \neq \beta ,
\end{array}
\eeq

\beq\label{mmkp16-3a}
\hspace{-3cm}\begin{array}{ll}
\mbox{I$\vphantom{a}_{3}$:}&
e^{\nabla_{\alpha}(a)\nabla_{\beta}(b)F}=R_{\beta}^{-1}
w_{\alpha \beta}(a)w_{\beta}(b)
\Bigl (p_{\beta}(b)-p_{\alpha \beta}(a)\Bigr ), \quad \alpha \neq \beta ,
\end{array}
\eeq

\beq\label{mmkp16-3aprime}
\hspace{-2.8cm}\begin{array}{ll}
\mbox{I${}'\vphantom{a}_{3}$:}&
e^{\nabla_{\alpha}(a)\nabla_{\beta}(b)F}=R_{\alpha}^{-1}
w_{\beta \alpha}(b)w_{\alpha}(a)
\Bigl (p_{\alpha}(a)-p_{\beta \alpha}(b)\Bigr ), \quad \alpha \neq \beta ,
\end{array}
\eeq

\beq\label{mmkp16-4a}
\hspace{-2.9cm}\begin{array}{ll}
\mbox{I$\vphantom{a}_{4}$:}&
e^{\nabla_{\alpha}(a)\nabla_{\beta}(b)F}=R_{\alpha \beta}^{-1}\,
\Bigl (w_{\alpha \beta}(a)w_{\beta \alpha}(b)+w_{\alpha}(a)
w_{\beta}(b) \Bigr ), \quad \alpha \neq \beta .
\end{array}
\eeq

\noindent
Note that equations I$\vphantom{a}_{3}$ and 
I${}'\vphantom{a}_{3}$ are obtained from each other 
by the simultaneous permutations
$\alpha \leftrightarrow \beta$, $a\leftrightarrow b$.

Equations
(\ref{mmkp17-12})--(\ref{mmkp17-4}) acquire the form

\beq\label{mmkp17-1a}
\hspace{-0.3cm}\begin{array}{ll}
\mbox{II$\vphantom{a}_{1}$:}&
(a^{-1}-b^{-1})e^{\nabla_{\alpha}(a)\nabla_{\alpha}(b)F} =
\bar R_{\alpha}^{-1}\Bigl ( w_{\alpha}(a) \bar w_{\alpha}(b)-
w_{\alpha}(b) \bar w_{\alpha}(a)\Bigr ), \phantom{\quad \alpha \neq \beta ,}
\end{array}
\eeq

\beq\label{mmkp17-1aa}
\begin{array}{ll}
\mbox{II$\vphantom{a}_{2}$:}&
(a^{-1}-b^{-1})e^{\nabla_{\alpha}(a)\nabla_{\alpha}(b)F} =
\bar R_{\alpha \beta}^{-1}\Bigl ( w_{\alpha}(a) \bar w_{\alpha \beta}(b)-
w_{\alpha}(b) \bar w_{\alpha \beta}(a)\Bigr ),
\quad \alpha \neq \beta ,
\end{array}
\eeq

\beq\label{mmkp17-3a}
\hspace{-0.4cm}\begin{array}{ll}
\mbox{II$\vphantom{a}_{3}$:}&
\epsilon_{\beta \alpha}
e^{\nabla_{\alpha}(a)\nabla_{\beta}(b)F} =\bar R_{\beta}^{-1} \, \Bigl (
\epsilon_{\beta \alpha}w_{\alpha \beta}(a)\bar w_{\beta}(b)-
w_{\beta}(b)\bar w_{\alpha \beta}(a)\Bigr ), \quad \alpha \neq \beta ,
\end{array}
\eeq

\beq\label{mmkp17-3aprime}
\hspace{-0.4cm}\begin{array}{ll}
\mbox{II${}'\vphantom{a}_{3}$:}&
\epsilon_{\beta \alpha}
e^{\nabla_{\alpha}(a)\nabla_{\beta}(b)F} =\bar R_{\alpha}^{-1} \, \Bigl (
w_{\alpha}(a)\bar w_{\beta \alpha}(b)-\epsilon_{\alpha \beta}
w_{\beta \alpha}(b)\bar w_{\alpha}(a)\Bigr ), \quad \alpha \neq \beta ,
\end{array}
\eeq

\beq\label{mmkp17-4a}
\hspace{-0.2cm}\begin{array}{ll}
\mbox{II$\vphantom{a}_{4}$:}&
\epsilon_{\beta \alpha}
e^{\nabla_{\alpha}(a)\nabla_{\beta}(b)F}=
\bar R_{\alpha \beta}^{-1}\, \Bigl (
w_{\alpha}(a)\bar w_{\beta}(b) -\epsilon_{\alpha \beta}
\bar w_{\alpha \beta}(a) w_{\beta \alpha}(b)\Bigr ), \quad
\alpha \neq \beta ,
\end{array}
\eeq

\beq\label{mmkp17-4aprime}
\hspace{-0.2cm}\begin{array}{ll}
\mbox{II${}'\vphantom{a}_{4}$:}&
\epsilon_{\beta \alpha}
e^{\nabla_{\alpha}(a)\nabla_{\beta}(b)F}=
\bar R_{ \beta \alpha}^{-1}\, \Bigl (\epsilon_{\beta \alpha}
w_{\alpha \beta}(a)\bar w_{\beta \alpha}(b) -
\bar w_{\alpha}(a)w_{\beta}(b)\Bigr ), \quad
\alpha \neq \beta .
\end{array}
\eeq

\noindent
Equations II$\vphantom{a}_{3}$ and 
II${}'\vphantom{a}_{3}$, as well as II$\vphantom{a}_{4}$ and 
II${}'\vphantom{a}_{4}$, are obtained from each other 
by the simultaneous permutations
$\alpha \leftrightarrow \beta$, $a\leftrightarrow b$.

\subsection{The dynamical curve: second appearance}

Now we are ready to recover the dynamical curve. 

\begin{theorem}\label{theorem:curve2}
For any distinct $\alpha , \beta =1, \ldots , N$ the functions 
$p_{\alpha}(z)$, $p_{\alpha \beta}(z)$ defined in (\ref{mmkp20})
satisfy the equation
\beq\label{mmkp24}
p_{\alpha \beta}(z)p_{\alpha}(z) -p_{\alpha \beta}(z)
p_{\beta \alpha}(\infty )
-p_{\alpha}(z)p_{\alpha \beta}(\infty )+
p_{\alpha \beta}(\infty)p_{\beta \alpha}(\infty )+
\frac{R_{\alpha}R_{\beta}}{R^2_{\alpha \beta}}=0.
\eeq
\end{theorem}

\noindent
{\it Proof.}
The idea of the proof is to 
note that the left-hand sides of some of the equations 
obtained above are the same. This allows one to eliminate
$e^{\nabla_{\alpha}(a)\nabla_{\beta}(b)F}$ from a part of the equations
thus obtaining some constraints on the $w$- and $p$-functions.
First of all, equating the right-hand sides
of equations (\ref{mmkp16-1a}) and (\ref{mmkp16-2a}), we obtain the relation
\beq\label{mmkp22}
p_{\alpha}(a)-
\frac{R_{\alpha}}{R_{\alpha \beta}} v_{\alpha \beta}(a)=
p_{\alpha}(b)-
\frac{R_{\alpha}}{R_{\alpha \beta}} v_{\alpha \beta}(b),
\eeq
where the function $v_{\alpha \beta}(z)$ is defined in (\ref{mmkp18}).
Since the left-hand side of this relation depends only on $a$ while
the right-hand side depends only on $b$, it follows that
$p_{\alpha}(z)-
R_{\alpha}R^{-1}_{\alpha \beta} v_{\alpha \beta}(z)\equiv
C_{\alpha \beta}$ is a $z$-independent quantity.
Letting $z$ to $\infty$, we can express it through derivatives
of the $F$-function:
\beq\label{mmkp23}
p_{\alpha}(z)-
\frac{R_{\alpha}}{R_{\alpha \beta}} v_{\alpha \beta}(z)=
p_{\beta \alpha}(\infty )=-\p_{\beta}\p_{t_{\alpha ,1}}F, \quad
\alpha \neq \beta .
\eeq
Next, from
(\ref{mmkp16-3a}), (\ref{mmkp16-3aprime}) and
(\ref{mmkp16-4a}) we have:
\beq\label{mmkp21}
\begin{array}{lll}
1+v_{\alpha \beta}(a)v_{\beta \alpha}(b)&=&
\displaystyle{
\frac{R_{\alpha \beta}}{R_{\alpha}}\, 
v_{\beta \alpha}(b)\Bigl (p_{\alpha}(a)-p_{\beta \alpha}(b)\Bigr )}
\\ && \\
&=&
\displaystyle{\frac{R_{\beta \alpha}}{R_{\beta}}\, 
v_{\alpha \beta}(a)\Bigl (p_{\beta}(b)-p_{\alpha \beta}(a)\Bigr )},
\end{array}
\eeq
Plugging 
\beq\label{v}
v_{\alpha \beta}(z)=\frac{R_{\alpha \beta}}{R_{\alpha}}
\, \Bigl (p_{\alpha}(z)-p_{\beta \alpha}(\infty )\Bigr )
\eeq
from (\ref{mmkp23}) into (\ref{mmkp21}), we obtain the following relation
containing the $p$-functions only:
$$
\begin{array}{c}
\displaystyle{
\frac{R_{\alpha}R_{\beta}}{R^2_{\alpha \beta}}+
\Bigl (p_{\alpha}(a)-p_{\beta \alpha}(\infty )\Bigr )
\Bigl (p_{\beta}(b)-p_{\alpha \beta}(\infty )\Bigr )}
\\ \\
=\, \Bigl (p_{\alpha}(a)-p_{\beta \alpha}(b )\Bigr )
\Bigl (p_{\beta}(b)-p_{\alpha \beta}(\infty )\Bigr )
=\, \Bigl (p_{\beta}(b)-p_{\alpha \beta}(a)\Bigr )
\Bigl (p_{\alpha}(a)-p_{\beta \alpha}(\infty )\Bigr ).
\end{array}
$$
After opening the brackets and some cancellations, 
we obtain the constraint (\ref{mmkp24})
for the functions $p_{\alpha}(z)$ and $p_{\alpha \beta}(z)$
that we are looking for.
\square

Putting $x=p_{\alpha \beta}(z)$, $y=p_{\alpha}(z)$, we rewrite
equation (\ref{mmkp24}) in the form 
\beq\label{curve2}
P(x,y)=xy +Ax +By +C=0.
\eeq
It defines 
a rational algebraic curve which is the dynamical curve for the
multi-component dmKP hierarchy. 
The functions $p_{\alpha \beta}(z)$ 
and $p_{\alpha}(z)$ are rational functions
on this curve, and
$z^{-1}$ is the local parameter in a neighborhood
of infinity.

\subsection{The trigonometric parametrization}
\label{section:trig}

Being a rational curve of degree 2, the curve (\ref{mmkp24}) can be 
uniformized by trigonometric functions. (See Appendix C for details.)
As we shall see, the uniformization of this curve 
allows one to reduce the plethora of equations from 
Section \ref{section:dmmKP} to a system of just
two equations 
having a nice compact form.

The key step is to introduce new dependent variables
$\eta_{\alpha}({\bf t})$, $c_{\alpha , k}
({\bf t})$ combined into
generating functions
\beq\label{t1}
u_{\alpha}(z)=\eta_{\alpha}({\bf t}) +\sum_{k\geq 1}c_{\alpha , k}
({\bf t})z^{-k}
\eeq
and treat $u_{\alpha}(z)$ as the uniformizing variable
$u$ from Appendix C. 
According to equations (\ref{curve1a})--(\ref{curve1c}) we uniformize 
the curve (\ref{mmkp24}) (i.e., (\ref{curve2})), as follows:
\beq\label{t2}
p_{\alpha}(z)=\gamma_{\alpha}
\frac{\cos (u_{\alpha}(z)-\eta_{\alpha})}{\sin (u_{\alpha}(z)
-\eta_{\alpha})},  \qquad
p_{\alpha \beta}(z)=\gamma_{\beta}
\frac{\cos (u_{\alpha}(z)-\eta_{\beta})}{\sin (u_{\alpha}(z)
-\eta_{\beta})},  
\eeq
so
\beq\label{t2a}
p_{\alpha \beta}(\infty )=\gamma_{\beta}
\frac{\cos (\eta_{\alpha}-\eta_{\beta})}{\sin (\eta_{\alpha}
-\eta_{\beta})}. 
\eeq
Together with
\beq\label{t3}
R_{\alpha}=\gamma_{\alpha}, \qquad
R_{\alpha \beta}=\epsilon_{\beta \alpha}
\sin (\eta_{\alpha}-\eta_{\beta})
\eeq
these substitutions convert equation (\ref{mmkp24}) into
an identity (see Appendix C for details). Here the $\gamma_{\alpha}$'s
are dynamical variables: $\gamma_{\alpha}=\gamma_{\alpha}({\bf t})$.

Now we can obtain the trigonometric parametrization of the
$w$-functions. From (\ref{v}) we have:
$$
\begin{array}{lll}
v_{\alpha \beta}& =& \displaystyle{\frac{w_{\alpha \beta}}{w_{\alpha}}=
\epsilon_{\beta \alpha} \sin (\eta_{\alpha}-\eta_{\beta})
\Bigl (\cot (u_{\alpha}-\eta_{\alpha}) +
\cot (\eta_{\alpha}-\eta_{\alpha})\Bigr )}
\\ && \\
& =& \displaystyle{
\epsilon_{\beta \alpha} \, 
\frac{\sin (u_{\alpha}-
\eta_{\beta})}{\sin (u_{\alpha}-\eta_{\alpha})}},
\end{array}
$$
where $v_{\alpha \beta}=v_{\alpha \beta}(z)$, etc. Therefore,
we can put
\beq\label{w}
w_{\alpha}(z)=\sin (u_{\alpha}(z)-\eta_{\alpha}), \qquad
w_{\alpha \beta}(z)=\epsilon_{\beta \alpha}
\sin (u_{\alpha}(z)-\eta_{\beta}).
\eeq

It is easy to check that with this parametrization equations
(\ref{mmkp16-1a}) and (\ref{mmkp16-4a}) acquire the form
\beq\label{t4}
\epsilon_{\beta \alpha}(a^{-1}-b^{-1})^{\delta_{\alpha \beta}}
e^{\nabla_{\alpha}(a)\nabla_{\beta}(b)F}=
\sin (u_{\alpha}(a)-u_{\beta}(b)).
\eeq
Expanding both sides of (\ref{t4}) at $\beta =\alpha$ in powers
of $z^{-1}$ and comparing the coefficients at the leading term,
we obtain:
\beq\label{t5}
R_{\alpha}({\bf t})=\gamma_{\alpha}({\bf t})=
c_{\alpha , 1}({\bf t})
\eeq
($c_{\alpha , 1}$ is the coefficient at $z^{-1}$ 
in the series (\ref{t1})).

It remains to find the trigonometric parametrization for the
$\bar w$-functions. 
To this end, substitute first the equation
$\mbox{II}{}_1$ (\ref{mmkp17-1a}) into (\ref{t4}) at $\beta =\alpha$.
After some algebra, this leads to the relation
$$
\bar R_{\alpha}^{-1} \, 
\frac{\bar w_{\alpha}(a)}{w_{\alpha}(a)} -\cot (u_{\alpha}(a)
-\eta_{\alpha})=
\bar R_{\alpha}^{-1} \, 
\frac{\bar w_{\alpha}(b)}{w_{\alpha}(b)} -\cot (u_{\alpha}(b)
-\eta_{\alpha})
$$
from which it follows that
\beq\label{t8}
\bar R_{\alpha}^{-1} \, 
\frac{\bar w_{\alpha}(z)}{w_{\alpha}(z)} -\cot (u_{\alpha}(z)
-\eta_{\alpha})
=U_{\alpha}
\eeq
does not depend on $z$. Letting $z\to \infty$, we find:
\beq\label{t7}
U_{\alpha}=e^{-\p_{\alpha}^2F}\bar \p_{\alpha} \p_{t_{\alpha ,1}}F.
\eeq
In a similar way, substituting 
$\mbox{II}{}_2$ (\ref{mmkp17-1aa}) into (\ref{t4}),
we find that
\beq\label{t8a}
\bar R_{\alpha \beta}^{-1} \, 
\frac{\bar w_{\alpha \beta}(z)}{w_{\alpha}(z)} -\cot (u_{\alpha}(z)
-\eta_{\alpha})
=U_{\alpha \beta}
\eeq
does not depend on $z$. Letting $z\to \infty$, we have:
\beq\label{t7a}
U_{\alpha \beta}=e^{-\p_{\alpha}^2F}\bar \p_{\beta} \p_{t_{\alpha ,1}}F.
\eeq
Now one can see that relations (\ref{t8}), (\ref{t8a}) are 
identically
satisfied if one puts
\beq\label{t9}
\bar w_{\alpha}(z)=\sin (u_{\alpha}(z)-\bar \eta_{\alpha}),
\qquad
\bar w_{\alpha \beta}(z)=\sin (u_{\alpha}(z)-\bar \eta_{\beta})
\eeq
and
\beq\label{t9a}
\begin{array}{l}
\bar R_{\alpha}=\bar w_{\alpha}(\infty )=
\sin (\eta_{\alpha}-\bar \eta_{\alpha}),
\qquad
\bar R_{\alpha \beta}=\bar w_{\alpha \beta}(\infty )=
\sin (\eta_{\alpha}-\bar \eta_{\beta})
\\ \\
\displaystyle{
U_{\alpha}=\frac{\cos (\eta_{\alpha}-\bar \eta_{\alpha})}{\sin(
\eta_{\alpha}-\bar \eta_{\alpha})}, \qquad
\phantom{aaaaaaaa}
U_{\alpha \beta}=\frac{\cos (\eta_{\alpha}-\bar \eta_{\beta})}{\sin(
\eta_{\alpha}-\bar \eta_{\beta})}.}
\end{array}
\eeq
Here $\bar \eta_{\alpha}=\bar \eta_{\alpha}({\bf t})$ are 
additional dependent variables.

Comparing equations 
(\ref{mmkp17-1aa}) and (\ref{mmkp16-2a}), we conclude, in a similar way,
that
\beq\label{t12}
\frac{R_{\alpha \beta}}{\bar R_{\alpha \beta}}\,
\frac{\bar w_{\alpha \beta}(z)}{w_{\alpha}(z)}-
\frac{w_{\alpha \beta}(z)}{w_{\alpha}(z)}=D_{\alpha \beta}
\eeq
does not depend on $z$. The limit $z\to \infty$ yields:
\beq\label{t11}
D_{\alpha \beta}=e^{-\p_{\alpha}^2F +\p_{\alpha}\p_{\beta}F}
(\bar \p_{\beta}-\p_{\beta})\p_{t_{\alpha , 1}}F.
\eeq
Plugging into (\ref{t12}) the trigonometric expressions 
(\ref{t3}), (\ref{w}), (\ref{t9}), (\ref{t9a}), one can see
that the left-hand side is equal to $\epsilon_{\alpha \beta}$,
so we conclude that
\beq\label{t13}
D_{\alpha \beta}=\epsilon_{\alpha \beta}.
\eeq
The mutual compatibility of the rest of the equations from
the lists (\ref{mmkp16-1a})--(\ref{mmkp16-4a}) and
(\ref{mmkp17-1a})--(\ref{mmkp17-4aprime}) can be checked
in a similar way. 

Summarizing, we have proved the following statement.

\begin{theorem}
In the trigonometric parametrization (\ref{t2}), (\ref{w}), 
(\ref{t9}), (\ref{t9a}) 
the set of 11 equations
$$
\mbox{I$\vphantom{a}_{1}$}, \; 
\mbox{I$\vphantom{a}_{2}$}, \;
\mbox{I$\vphantom{a}_{3}$}, \;
\mbox{I${}'\vphantom{a}_{3}$}, \;
\mbox{I$\vphantom{a}_{4}$}, \;
\mbox{II$\vphantom{a}_{1}$}, \; 
\mbox{II$\vphantom{a}_{2}$}, \;
\mbox{II$\vphantom{a}_{3}$}, \;
\mbox{II${}'\vphantom{a}_{3}$}, \;
\mbox{II$\vphantom{a}_{4}$},
\mbox{II${}'\vphantom{a}_{4}$}
$$
(see (\ref{mmkp16-1a})--(\ref{mmkp16-4a}) and
(\ref{mmkp17-1a})--(\ref{mmkp17-4aprime})) reduces to the following
two compact equations 
\beq\label{t15}
\left \{
\begin{array}{l}
\epsilon_{\beta \alpha}(a^{-1}-b^{-1})^{\delta_{\alpha \beta}}
e^{\nabla_{\alpha}(a)\nabla_{\beta}(b)F}=
\sin \Bigl (u_{\alpha}(a)-u_{\beta}(b)\Bigr ),
\\ \\
e^{\nabla_{\alpha}(a)\bar \p_{\beta}F}=\sin \Bigl (u_{\alpha}(a)
-\bar \eta_{\beta}\Bigr ),
\end{array} \right.
\eeq
where $u_{\alpha}(z)$ is of the form (\ref{t1}).
\end{theorem}

\noindent
The reduction of a plethora of equations to just two is 
the main advantage of introducing the dynamical curve and its 
uniformization. 
Another advantage is given by the following proposition.

\begin{proposition}
Solutions to
the two equations (\ref{t15})
solve the whole dispersionless multi-component 
mKP hierarchy. 
\end{proposition}

\noindent
{\it Proof.}
It is enough to check that imposing the two equations (\ref{t15}), one
solves the general Hirota-Miwa equation (\ref{mmkp12}).
Indeed, 
the substitution of (\ref{t15}) into (\ref{mmkp12})
leads to the relation
\beq\label{t16}
\sum_{s=1}^{P^+}
\frac{\prod\limits_{k=1}^{P^-}\sin \Bigl (u_{\alpha_s}(a_s)-
u_{\beta_k}(b_k)\Bigr )}{\prod\limits_{i=1, \neq s}^{P^+}\sin \Bigl (u_{\alpha_s}(a_s)- u_{\alpha_i}(a_i)\Bigr )}\,
\prod_{j=1}^Q \sin 
\Bigl (u_{\alpha_s}(a_s)- \bar \eta_{\gamma_j}\Bigr )\, =\, 0,
\eeq
which is satisfied identically since the left-hand side 
is the sum of residues of the $\pi$-periodic function
$$
f(u)=\frac{\prod\limits_{k=1}^{P^-}\sin \Bigl (u-u_{\beta_k}(b_k)
\Bigr )}{\prod\limits_{i=1}^{P^+}\sin \Bigl (u-u_{\alpha_i}(a_i)
\Bigr )}\, \prod_{j=1}^Q \sin 
\Bigl (u- \bar \eta_{\gamma_j}\Bigr )
$$
in the fundamental domain (the strip $0\leq {\rm Re} \, u <\pi$).
It is indeed equal to zero because the 
contour integral over the boundary of the fundamental domain 
vanishes:
the integrals along the vertical lines 
${\rm Re}\, u=0$ and ${\rm Re}\, u=\pi$
cancel due to the periodicity, and the integrals over segments
at infinity vanish because the function tends to zero
as $u\to \pm i\infty$ (recall that $P^{+}-P^- =Q+2$, so 
the denominator contains two
extra $\sin$-functions).
\square

\begin{remark}
It is instructive to follow how the 
trigonometric parametrization  
obtained in the previous subsection works in the one-component
case $N=1$ considered in Section \ref{section:dmKP}. 
This (not so simple) question is discussed in detail in Appendix D.
\end{remark}

\section{Multi-component Toda lattice
hierarchy and its dispersionless version}

In the $N$-component Toda lattice (TL) 
hierarchy\footnote{Compared to the hierarchy dealt with 
in \cite{TZ25}, we consider here its somewhat more general
version with two sets of discrete variables rather than one.}
the independent variables are
$2N$ infinite sets of ``times'',
\beq\label{T1}
\begin{array}{l}
{\bf t}=\{{\bf t}_1, {\bf t}_2, \ldots , {\bf t}_N\}, \qquad
{\bf t}_{\alpha}=\{t_{\alpha , 1}, t_{\alpha , 2}, t_{\alpha , 3}, 
\ldots \, \},
\\ \\
\bar {\bf t}=\{\bar {\bf t}_1, \bar {\bf t}_2, \ldots , 
\bar {\bf t}_N\}, \qquad
\bar {\bf t}_{\alpha}=\{\bar t_{\alpha , 1}, \bar t_{\alpha , 2}, 
\bar t_{\alpha , 3}, \ldots \, \},
\end{array}
\qquad \alpha = 1, \ldots , N
\eeq
and two finite sets of discrete variables
$$
{\bf n}=\{n_1, \ldots , n_N\}, \quad
\bar {\bf n}=\{\bar n_1, \ldots , \bar n_N\}, \quad n_{\alpha}, \bar n_{\alpha}
\in \ZZ 
$$
such that
\beq\label{T2}
|{\bf n}|=-|\bar {\bf n}|.
\eeq
The universal
dependent variable is the tau-function
$\tau ({\bf n}, \bar {\bf n}, {\bf t},
\bar {\bf t})$. In the fermionic approach it is defined
as the following expectation value:
\beq\label{T3}
\tau ({\bf n}, \bar {\bf n}, {\bf t}, \bar {\bf t})=
\lbr {\bf n}\bigl | e^{J({\bf t})} g e^{-\bar J(\bar {\bf t})} 
\bigr |-\bar {\bf n}\rbr ,
\eeq
where $g$ is a neutral Clifford group element of the form (\ref{g0}).

\subsection{Multi-component TL hierarchy}
\label{section:multi-Toda}

As is shown in \cite{TZ25},
the tau-function satisfies an infinite
number of bilinear equations which can be encoded in a single 
integral bilinear functional relation of the form
\beq\label{T4}
\begin{array}{l}
\displaystyle{
\sum_{\gamma =1}^N \epsilon_{\gamma}({\bf n})
\epsilon_{\gamma}({\bf n'})
\oint_{C_{\infty}} \frac{dz}{z^2}\, z^{n_{\gamma}-n_{\gamma}'}\,
e^{\xi ({\bf t}_{\gamma}-{\bf t'}_{\gamma}, z)}}
\\ \\
\phantom{aaaaaaaaaaaaaaa}\displaystyle{
\times \, \tau \Bigl ({\bf n}\! -\! {\bf e}_{\gamma}, 
\bar {\bf n}, {\bf t}-[z^{-1}]_{\gamma} , \bar {\bf t}\Bigr )
\tau \Bigl (
{\bf n'}\! +\! {\bf e}_{\gamma}, \bar {\bf n'}, 
{\bf t'}+[z^{-1}]_{\gamma} , \bar {\bf t}'\Bigr )}
\\ \\
\displaystyle{
=\, \sum_{\gamma =1}^N \epsilon_{\gamma}(\bar {\bf n})
\epsilon_{\gamma}(\bar {\bf n'})
\oint_{C_{\infty}} \frac{dz}{z^2}\, z^{\bar n_{\gamma}-\bar n_{\gamma}'}\,
e^{\xi (\bar {\bf t}_{\gamma}-\bar {\bf t'}_{\gamma}, z)}}
\\ \\
\phantom{aaaaaaaaaaaaaaa}\displaystyle{
\times \, \tau \Bigl ({\bf n} , 
\bar {\bf n}-\! {\bf e}_{\gamma}, {\bf t} , \bar {\bf t}
-[z^{-1}]_{\gamma}\Bigr )
\tau \Bigl (
{\bf n'}, \bar {\bf n'}\! +\! {\bf e}_{\gamma}, 
{\bf t'} , \bar {\bf t}' +[z^{-1}]_{\gamma}\Bigr )}.
\end{array}
\eeq
Here it is implied that $| {\bf n}|+|\bar {\bf n}|=1$,
$| {\bf n}'|+|\bar {\bf n}'|=-1$
For any $N\geq 1$, after setting 
$\bar {\bf n}'=\bar {\bf n}$, $\bar {\bf t}'=\bar {\bf t}$
in (\ref{T4}),
the bar-variables do not
participate in the equation entering as parameters. 
Then the right-hand side of (\ref{T4}) vanishes identically
and the rest becomes the integral bilinear
equation for the tau-function of the $N$-component 
KP hierarchy (equation (\ref{mmkp2}) with 
${\bf m}={\bf m}'$). In this sense the latter hierarchy
can be regarded as a subhierarchy of the multi-component TL. 
On the other hand, the following theorem fully proven in \cite{TZ25}
states 
that the $2N$-component KP
is in fact equivalent to the $N$-component TL. 

\begin{theorem}\cite{UT84,TZ25}
The $N$-component Toda lattice hierarchy is equivalent 
to the $2N$-component KP hierarchy. The equivalence is established
by the following relation for their tau-functions:
\beq\label{tautau}
\tau^{\rm TL} ({\bf n}, \bar {\bf n}, {\bf t}, \bar {\bf t})=
(-1)^{\frac{1}{2}|\bar {\bf n}| (|\bar {\bf n}|-1)}
\tau^{\rm KP}(\tilde {\bf n}, \tilde {\bf t}),
\eeq
where the sets of the variables $\tilde {\bf n}, \tilde {\bf t}$ 
are 
$$\tilde {\bf n}=\{n_1, \ldots , n_N, \bar n_1, \ldots , \bar n_N\},
\quad
 \tilde {\bf t}=\{{\bf t}_1, \ldots , {\bf t}_N, 
 \bar {\bf t}_1, \ldots , \bar {\bf t}_N\}.
 $$
\end{theorem}

\noindent
Let us present the main points of this identification
(for more details see \cite{TZ25}). 
For the case of $M$-component KP hierarchy the integral
bilinear equation reads:
\beq\label{T5}
\begin{array}{l}
\displaystyle{
\sum_{\gamma =1}^M \epsilon_{\gamma}({\bf n})
\epsilon_{\gamma}({\bf n'})
\oint_{C_{\infty}} \frac{dz}{z^2}\, z^{n_{\gamma}-n_{\gamma}'}\,
e^{\xi ({\bf t}_{\gamma}-{\bf t'}_{\gamma}, z)}}
\\ \\
\phantom{aaaaaaaaaaaaaaa}\displaystyle{
\times \, \tau \Bigl ({\bf n}\! -\! {\bf e}_{\gamma}, 
{\bf t}-[z^{-1}]_{\gamma} \Bigr )
\tau \Bigl (
{\bf n'}\! +\! {\bf e}_{\gamma}, 
{\bf t'}+[z^{-1}]_{\gamma} \Bigr )}=0.
\end{array}
\eeq
(Here $\tau ({\bf n}, {\bf t})$ is the tau-function
of the $M$-component 
KP hierarchy.) 
It is valid for all ${\bf t}$, ${\bf t'}$ and ${\bf n}$, ${\bf n'}$
such that $|{\bf n}|=1$ and $|\bar {\bf n}'|=-1$.
Let us re-denote the variables in (\ref{T5}) with $M=2N$ in 
the following special way. Let the index $\mu$ run from $1$ to $N$ and set
\beq\label{re}
n_{N+\mu}=\bar n_{\mu}, \qquad 
{\bf t}_{N+\mu}=\bar {\bf t}_{\mu}.
\eeq
Divide the sum over $\gamma$ in (\ref{T5}) in two: 
one from $1$ to $N$ and
the other from $N+1$ to $2N$. Then, after the obvious redefinition of the
tau-function, the first sum in (\ref{T5}) becomes
almost equal to the left-hand side of (\ref{T4}) while the second one in
(\ref{T5}) is almost the right-hand side of (\ref{T4}). 
(``Almost'' because
it still remains to identify the sign factors.) 
The comparison of the sign factors shows that
what comes from (\ref{T5}) as the left-hand side of (\ref{T4}) contains
an extra sign factor $(-1)^{|\bar {\bf n}|-|\bar {\bf n}'|+1}$. As is easy
to see, it can be eliminated after multiplying the tau-function by 
the sign factor $(-1)^{\frac{1}{2}|\bar {\bf n}| (|\bar {\bf n}|-1)}$,
i.e., we identify the tau-functions as in (\ref{tautau}).
Note that on the KP side there is the restriction
$|{\bf n}|+|\bar {\bf n}|=0$,
while on the TL side it is
$|{\bf n}|=-|\bar {\bf n}|$ which is the same.

Summarizing, we see that 
the relation between multi-component TL and KP
hierarchies is two-fold. On the one hand, the latter is a subhierarchy 
of the former (``a half'' of it). On the other hand, the $N$-component
TL can be regarded, after a renaming of the variables,
as the $2N$-component KP. 
This fact makes it possible not to consider 
the Toda lattice separately but just 
to translate any statement about it to the language of the latter
hierarchy.

\subsection{The dispersionless limit}
\label{section:disp-Toda}

According to the multi-point TL-KP equivalence 
explained in section \ref{section:multi-Toda}
(see, in particular, (\ref{tautau}) and (\ref{re})),
it remains to rewrite the first equation in (\ref{t15}),
i.e.,
\beq\label{dt1}
\epsilon_{\beta \alpha}(a^{-1}-b^{-1})^{\delta_{\alpha \beta}}
e^{\nabla_{\alpha}(a)\nabla_{\beta}(b)F}=
\sin \Bigl (u_{\alpha}(a)-u_{\beta}(b)\Bigr )
\eeq
in terms of the variables of the TL hierarchy.
So, we introduce the bar-counterpart of the vector field (\ref{mmkp11}):
\beq\label{n1a}
\bar \nabla_{\alpha}(z)=\bar \p_{\alpha}+\sum_{k\geq 1}\frac{z^{-k}}{k}\,
\p_{\bar t_{\alpha , 1}},
\eeq
where
$
\bar \p_{\alpha}\equiv \p_{\bar t_{\alpha , 0}}.
$
Also, we need to take into account that the $F$-functions of the two
hierarchies slightly differ because of the sign factor in (\ref{tautau}).
Writing it as
$$
(-1)^{\frac{1}{2}|\bar {\bf n}|(|\bar {\bf n}|-1)}=
e^{\frac{i\pi}{2}|\bar {\bf n}|(|\bar {\bf n}|-1)},
$$
we have:
$$
\frac{i\pi}{2}|\bar {\bf n}|(|\bar {\bf n}|-1)=
\frac{i\pi}{2}\left (\sum_{\mu , \nu}\bar n_{\mu}\bar n_{\nu}-
\sum_{\mu}\bar n_{\mu}\right )=
\frac{i\pi}{2\hbar^2}\left (\sum_{\mu , \nu}\bar t_{\mu ,0}\bar t_{\nu ,0}-
\hbar \sum_{\mu}\bar t_{\mu ,0}\right ),
$$
from which one can see that
the relation between the two $F$-functions
is as follows:
\beq\label{FF}
F^{\rm TL}({\bf t}, \bar {\bf t})
= \frac{i\pi}{2}\sum_{\mu , \nu}\bar t_{\mu ,0}\bar t_{\nu ,0} +
F^{\rm KP}(\tilde {\bf t}),
\eeq
where the sets of times are
$$
{\bf t}=\{{\bf t}_1, \ldots , {\bf t}_{N}\}, 
\qquad
\bar {\bf t}=\{ \bar {\bf t}_1, \ldots , 
\bar {\bf t}_{N}\}, \qquad
\tilde {\bf t}=\{{\bf t}_1, \ldots , {\bf t}_{2N}\},
$$
with ${\bf t}_{\mu}=\tilde {\bf t}_{\mu}$,
$\bar {\bf t}_{\mu}=\tilde {\bf t}_{\mu +N}$ for $\mu =1, \ldots , N$,
and
$$
{\bf t}_{\gamma}=\{ t_{\gamma , 0}, t_{\gamma , 1}, \ldots , \}, 
\qquad
\bar {\bf t}_{\gamma}=\{ \bar t_{\gamma , 0}, \bar t_{\gamma , 1}, \ldots , \}.
$$
Therefore, translating equation (\ref{dt1}) into the Toda language, 
we should take into account that 
$$
e^{\nabla_{\alpha +N}(a)\nabla_{\beta +N}(b)F^{\rm KP}} \to -\, 
e^{\bar \nabla_{\alpha}(a)\bar \nabla_{\beta}(b)F^{\rm TL}},
\quad \alpha , \beta =1, \ldots , N.
$$
This can be done by introducing, along with the $u_{\alpha}$'s,
the functions
$\bar u_{\alpha}(z)=-u_{\alpha +N}(z)$, with the expansions 
near $\infty$ being of the form
\beq\label{baru}
\bar u_{\alpha}(z)=\bar \eta_{\alpha}({\bf t}, \bar {\bf t})
+\sum_{k\geq 1}\bar c_k^{(\alpha )}({\bf t}, \bar {\bf t}) z^{-k}.
\eeq
The final result is the following system of equations:
\beq\label{final}
\left \{
\begin{array}{rcl}
E_{\beta \alpha}(a,b) e^{\nabla_{\alpha}(a)\nabla_{\beta}(b)F}
&=& \displaystyle{ \sin \Bigl (u_{\alpha}(a)- 
u_{\beta}(b)\Bigr ),}
\\ && \\
e^{\nabla_{\alpha}(a)\bar \nabla_{\beta}(b)F}& =& \displaystyle{
\sin \Bigl (u_{\alpha}(a)+ \bar u_{\beta}(b))\Bigr ),}
\\ && \\
E_{\beta \alpha}(a,b) e^{\bar \nabla_{\alpha}(a)\bar \nabla_{\beta}(b)F}
& =& \displaystyle{\sin \Bigl (\bar u_{\alpha}(a)- 
\bar u_{\beta}(b))\Bigr )}.
\end{array} \right.
\eeq
Here $\alpha , \beta =1, \ldots , N$.
Our definition of the $\bar u_{\alpha}$'s is such that the bar-counterpart
of equations (\ref{dt1}) has exactly the same form, with 
$\bar u_{\alpha}$-functions
instead of $u_{\alpha}$'s. However, because of this
the equation
that mixes the variables with and without bar (the second line in (\ref{final}))
contains $u_{\alpha}+\bar u_{\beta}$ rather than
$u_{\alpha}-\bar u_{\beta}$. 
In particular, at $N=1$ the system (\ref{final}) is
\beq\label{final2}
\left \{
\begin{array}{rcl}
(a^{-1}-b^{-1})e^{\nabla (a)\nabla(b)F}
&=& \displaystyle{ \sin \Bigl (u(a)- 
u(b)\Bigr ),}
\\ && \\
e^{\nabla(a)\bar \nabla (b)F}& =& \displaystyle{
\sin \Bigl (u(a)+ \bar u(b))\Bigr ),}
\\ && \\
(a^{-1}-b^{-1})e^{\bar \nabla(a)\bar \nabla (b)F}
& =& \displaystyle{\sin \Bigl (\bar u(a)- 
\bar u(b))\Bigr )}.
\end{array} \right.
\eeq

\section{One-component hierarchies of Pfaff type}

\subsection{The simplest case: 
BKP and its dispersionless limit}
\label{section:smallBKP}

The BKP hierarchy \cite{JM83}, \cite{DJKM82}-\cite{Z21} 
is also known as the ``small BKP hierarchy'' to distinguish it from
the more general ``large BKP hierarchy'' considered in the
next subsection. The independent variables are
$$
\t =\{t_1, t_3, t_5, \ldots \}.
$$
(We use the letter $\t $ instead of ${\bf t}$ 
for this set that contains only times 
with odd indices.)

\subsubsection{The small BKP hierarchy}

The tau-function of the small BKP hierarchy,
$\tau (\t )$, satisfies the integral bilinear equation of the form:
\beq\label{BB1}
\frac{1}{2\pi i}\oint_{C_{\infty}}\, \frac{dz}{z}\,
e^{\xi_{o} (\t -\t ',z)}
\tau \Bigl (\t -2[z^{-1}]_{o}\Bigr )
\tau \Bigl (\t ' +2[z^{-1}]_{o}\Bigr )=\tau (\t )\tau (\t '),
\eeq
valid for all $\t , \, \t '$. We use the 
following notation:
\beq\label{xio}
\begin{array}{l}
\displaystyle{
\xi _{o} (\t ,z)=\sum_{k \geq 1, \, {\rm odd}}
t_k z^k,}
\\ \\ 
\t \pm 2[z^{-1}]_{o} =\Bigl \{t_1 \pm 2z^{-1}, \,
t_3 \pm \frac{2}{3} z^{-3}, \, t_5 \pm \frac{2}{5} z^{-5},
\ldots \, \Bigr \}.
\end{array}
\eeq
The simplest solution to (\ref{BB1}) is $\tau (\t )=1$.

The Miwa substitution is 
$\displaystyle{
\t -\t '=2\sum_{i=1}^P [a_i^{-1}]_{o},}
$
where the points $a_i\in \CC$ belonging to a neighborhood
of infinity are assumed to be distinct\footnote{Note that 
there is no need to include terms like $-[b_i^{-1}]_{o}$
in the right-hand side because 
$-[b_i^{-1}]_{o}=[-b_i^{-1}]_{o}$.}.
We have:
\beq\label{BB2}
e^{\xi_{o} (\t -\t ',z)}=\prod_{i=1}^P \frac{a_i +z}{a_i-z},
\eeq
and the integral in (\ref{BB1}) can be evaluated by the residue
calculus taking into account 
that there is a non-zero residue at $\infty$.
As a result, equation (\ref{BB1}) converts into
\beq\label{BB3}
\sum_{s=1}^P \Bigl (\prod_{i=1, \neq s}A^{-1}(a_s, a_i)\Bigr )
\tau \Bigl (\t +2\sum_{j=1, \neq s} \!\!\! [a^{-1}_j ]_{o}\Bigr )
\tau \Bigl (\t +2\, [a^{-1}_s ]_{o}\Bigr )=
s_P \tau \bigl (\t \bigr )
\tau \Bigl (\t +2\sum_{j=1}[a^{-1}_j ]_{o}\Bigr ),
\eeq
where
\beq\label{A}
A(a,b)=\frac{a^{-1}-b^{-1}}{a^{-1}+b^{-1}} \, =\, 
\frac{b-a}{b+a}\, =\, \frac{E(a,\, b)}{E(a, -b)}
\eeq
and
$$
s_P =\frac{1}{2}\Bigl (1-(-1)^P\Bigr )=\left \{
\begin{array}{l}
0 \quad \mbox{for even $P$}
\\ 
1 \quad \mbox{for odd $\, P$}.
\end{array}\right.
$$
So, for $P$ even the right-hand side of (\ref{BB3}) is $0$.
At $P=1$ and $P=2$ (\ref{BB3}) 
is a trivial identity. The simplest meaningful
choice is $P=3$, in which case we obtain the well known 4-term
bilinear equation first found by Miwa in \cite{Miwa82}
(equation (\ref{BB6}) below).

\begin{remark}
The equation 
(\ref{BB3}) with odd $P$ can be obtained as a particular case
of the one with even $P$ simply by letting $a_P\to \infty$.
That is why we will not consider the case of odd $P$ separately
and assume that $P$ is even. 
\end{remark}

For even $P$, equation
(\ref{BB3}) acquires the simpler form
\beq\label{BB3a}
\sum_{s=1}^P \Bigl(\prod_{i=1, \neq s}A^{-1}(a_s, a_i)\Bigr )
\tau \Bigl (\t +2\sum_{j=1, \neq s} \!\!\! [a^{-1}_j ]_{o}\Bigr )
\tau \Bigl (\t +2\, [a^{-1}_j ]_{o}\Bigr )=0
\eeq
with $0$ in the right-hand side. 
For even $P\geq 4$ it contains
$P$ bilinear terms, with the coefficients being rational
functions of the $a_i$'s.
The simplest case is $P=4$:
\beq\label{BB5}
\begin{array}{l}
\phantom{+}
\Bigl (A(a_1,a_2)A(a_1,a_3)A(a_1,a_4)\Bigr )^{-1}
\tau \Bigl (\t +2[a_2^{-1}]_{o} +2[a_3^{-1}]_{o} +2[a_4^{-1}]_{o}\Bigr )
\tau \Bigl (\t +2[a_1^{-1}]_{o}\Bigr )
\\ \\
+\, 
\Bigl (A(a_2,a_1)A(a_2,a_3)A(a_2,a_4)\Bigr )^{-1}
\tau \Bigl (\t +2[a_1^{-1}]_{o} +2[a_3^{-1}]_{o} +2[a_4^{-1}]_{o}\Bigr )
\tau \Bigl (\t +2[a_2^{-1}]_{o}\Bigr )
\\ \\
+\, 
\Bigl (A(a_3,a_2)A(a_3,a_1)A(a_3,a_4)\Bigr )^{-1}
\tau \Bigl (\t +2[a_1^{-1}]_{o} +2[a_2^{-1}]_{o} +2[a_4^{-1}]_{o}\Bigr )
\tau \Bigl (\t +2[a_3^{-1}]_{o}\Bigr )
\\ \\
+\, 
\Bigl (A(a_4,a_1)A(a_4,a_2)A(a_4,a_3)\Bigr )^{-1}
\tau \Bigl (\t +2[a_1^{-1}]_{o} +2[a_2^{-1}]_{o} +2[a_3^{-1}]_{o}\Bigr )
\tau \Bigl (\t +2[a_4^{-1}]_{o}\Bigr )\, =\, 0.
\end{array}
\eeq
Letting $a_4\to\infty$, we get the 3-point equation
\beq\label{BB6}
\begin{array}{l}
\phantom{+}
\Bigl (A(a_1,a_2)A(a_1,a_3)\Bigr )^{-1}
\tau \Bigl (\t +2[a_2^{-1}]_{o} +2[a_3^{-1}]_{o}\Bigr )
\, \tau \Bigl (\t +2[a_1^{-1}]_{o}\Bigr )
\\ \\
+\, 
\Bigl (A(a_2,a_2)A(a_2,a_3)\Bigr )^{-1}
\tau \Bigl (\t +2[a_1^{-1}]_{o} +2[a_3^{-1}]_{o} \Bigr )
\, \tau \Bigl (\t +2[a_2^{-1}]_{o}\Bigr )
\\ \\
+\, 
\Bigl (A(a_3,a_2)A(a_3,a_3)\Bigr )^{-1}
\tau \Bigl (\t +2[a_1^{-1}]_{o} +2[a_2^{-1}]_{o} \Bigr )
\, \tau \Bigl (\t +2[a_3^{-1}]_{o}\Bigr )
\\ \\
-\, 
\tau \Bigl (\t +2[a_1^{-1}]_{o} +2[a_2^{-1}]_{o} +2[a_3^{-1}]_{o}\Bigr )
\, \tau \bigl (\t \bigr )\, =\, 0
\end{array}
\eeq
first obtained by Miwa in \cite{Miwa82}.
Similarly to Theorem \ref{theorem:kp}, the following statement
holds true:

\begin{theorem}\cite{Shigyo13}
\label{theorem:bkp}
Equation (\ref{BB6}) is equivalent to the whole BKP hierarchy
defined by (\ref{BB1}).
\end{theorem}

\begin{remark}
At $a_3=\infty$ the left-hand side of (\ref{BB6}) 
vanishes identically.
To obtain something non-trivial from this, one should expand 
the equation in 
powers of $a_3^{-1}\to 0$ and keep the first nonvanishing 
term (of order $a_3^{-1}$). As a result, one 
obtains an equation containing derivatives with respect to the
times.
\end{remark}

\subsubsection{The dispersionless limit: the dBKP hierarchy}

To perform the dispersionless limit, we introduce the BKP
version of the differential operator (\ref{kp7}):
\beq\label{db1}
D^{o}(z)=\sum_{k\geq 1, \, {\rm odd}}
\frac{z^{-k}}{k} \, \p_{t_k}.
\eeq
Details of the limit are basically the same as in the KP case,
and, omitting them, we present only the result:
the limiting form of equation (\ref{BB3}) is
\beq\label{db2}
\sum_{s=1}^P \prod_{i=1, \neq s}^P\Bigl (  A^{-1}(a_s, a_i)
e^{-4D^{o}(a_s)D^{o}(a_i)F}\Bigr )=0.
\eeq
In particular,
the limit of the 3-point equation (\ref{BB6}) reads:
\beq\label{BB6a}
\begin{array}{l}
\phantom{a}(a_1-a_2)(a_3+a_1)(a_3+a_2)e^{4D^{o}(a_1)D^{o}(a_2) F} 
+
(a_2-a_3)(a_1+a_2)(a_1+a_3)e^{4D^{o}(a_2)D^{o}(a_3) F} 
\\ \\
\phantom{aaa}+\,
(a_3-a_1)(a_2+a_3)(a_2+a_1)e^{4D^{o}(a_3)D^{o}(a_1) F}
-\, 
\displaystyle{\prod_{i<j}^3 (a_i-a_j)
e^{4D^{o}(a_i)D^{o}(a_j)F}=0.}
\end{array}
\eeq

The next step is to obtain from it a 2-point relation letting
$a_3\to \infty$. As it was already mentioned, 
in this limit the left-hand side vanishes identically.
So, we should expand it in 
powers of $a_3^{-1}\to 0$ and keep the first 
nonvanishing term. We present the result in the form which is
analogous to the corresponding equation in the KP case 
(see (\ref{kp12})).
Set
\beq\label{db3}
p(z)=z-2D^{o}(z)\p_{t_1}F,
\eeq
then the 2-point equation obtained as a corollary of (\ref{BB6a}) is
\beq\label{db4}
\frac{a-b}{a+b}\, \, e^{4D^{o}(a)D^{o}(b)F}=
\frac{p(a)-p(b)}{p(a)+p(b)}
\eeq
(we have put $a_1=a$, $a_2=b$).

\begin{proposition}
The 2-point equation (\ref{db4}) 
is equivalent to the whole
hierarchy (\ref{db2}).
\end{proposition} 

\noindent
{\it Proof.}
Plugging (\ref{db4}) into
(\ref{db2}), we arrive at the relation
\beq\label{db5}
\sum_{s=1}^P\prod_{i=1, \neq s}^P 
\frac{p_s+p_i}{p_s-p_i}=0, \qquad p_i\equiv p(a_i),
\eeq
which is satisfied identically (for all $p_i\neq 0$) 
since the left-hand side is
proportional (for even $P$ only!) to sum 
of residues of the function
$$
f(p)=\frac{1}{p}\, \prod_{i=1}^P \frac{p+p_i}{p-p_i},
$$
including the residues at zero and at infinity (which for even $P$
cancel each other).
\square

So, we see that, like in the dKP case, the BKP dynamical 
curve is rather meaningless:
it is a rational curve (just the Riemann sphere) 
defined by the trivial linear equation
of the form $x-y=0$.

\begin{remark}\label{remark:novikov-veselov}
The lack of space does not allow us to consider the 
2-component small BKP hierarchy, which contains 
the well-known Novikov-Veselov equation \cite{NV84a,NV84b}.
Its dispersionless limit was investigated by Takasaki \cite{Takasaki06}.
In this case the dynamical curve is less trivial, it is a rational
curve defined by equation 
of the form $xy=u$, where $u$ is a dynamical variable. 
\end{remark}

To conclude this section, 
we comment on how equation (\ref{db4}) for the small dBKP is
related to equation (\ref{kp12}) for the dKP hierarchy. 
There is a well known relation between KP and BKP tau-functions:
the latter is square root of the former, in which one should 
put $t_{2k}=0$ for all $k\geq 1$ and restrict oneself by 
a class of KP-solutions that satisfy certain conditions 
(for details see, e.g,
\cite{Z21} and \cite{KZ21}). In the dispersionless limit 
these conditions for the dKP $F$-function
$F^{\rm KP}({\bf t})$ are as follows:
\beq\label{db6}
\p_{t_{2k}}F^{\rm KP}(t_1, t_2, t_3, t_4, \ldots )
\Bigr |_{\, {\bf t}_{\rm even}=0}=0.
\quad \mbox{for all $k\geq 1$ and all times ${\bf t}_{\rm odd}=
{\sf t}$ }.
\eeq
Here ${\bf t}_{\rm even}$ (${\bf t}_{\rm odd}$) 
is the set of all ``even'' times
$t_2, t_4, t_6, \ldots $ (respectively, ``odd'' times
$t_1, t_3, t_5, \ldots $). To avoid a misunderstanding,
we emphasize that these conditions in no way 
mean that the function does not depend on the ``even'' variables:
for example, second order derivatives $\p_{t_{2k}}^2 F^{\rm KP}$ 
and higher ones may be
non-zero. If these conditions are satisfied, the relation
between the $F$-functions is as follows:
\beq\label{db7}
F^{\rm KP}(t_1, 0, t_3, 0, t_5, 0, \ldots )=2
F^{\rm BKP}(t_1, t_3, t_5, \ldots ).
\eeq
Therefore,
$
D(z)F^{\rm KP}=2D^{o}(z)F^{\rm BKP}
$, so
$$
p(z)=z-D(z)\p_{t_1}F^{\rm KP}=z-2D^{o}(z)\p_{t_1}F^{\rm BKP}.
$$
Note that for this class of solutions the function $p(z)$ 
is odd: $p(-z)=-p(z)$.
Moreover, we have
\beq\label{kp12b}
D(a)\Bigl (D(b)-D(-b)\Bigr )F^{\rm KP}=4D^{o}(a)D^{o}(b)F^{\rm BKP}.
\eeq
Now, let us write down 
two copies of equation (\ref{kp12}) with $a_1=a$:
one for 
$a_2=b$, another for $a_2=-b$:
\beq\label{kp12a}
\begin{array}{l}
(a-b)e^{D(a)D(b)F^{\rm KP}} =p(a)-p(b),
\\ \\
(a+b)e^{D(a)D(-b)F^{\rm KP}} =p(a)+p(b)
\end{array}
\eeq
and divide one by another. Using (\ref{kp12b}), 
we get 
equation (\ref{db4}).

\subsection{One-component DKP and dDKP}

\subsubsection{One-component DKP}

The next case of interest is the DKP hierarchy, also known 
as the coupled KP hierarchy or Pfaff lattice 
(see \cite{HO}-\cite{Kodama}). 
Similarly to the mKP case,
the corresponding tau-function $\tau(n, \textbf{t})$ 
is a function of the continuous times 
${\bf t}=\{t_1, t_2, t_3, \ldots \}$
and a discrete variable $n \in \ZZ$. The 
tau-function satisfies the following bilinear equation:
\beq\label{dkp-bil-rel}
\begin{array}{l}
\displaystyle{
\oint_{C_{\infty}} \frac{dz}{z^2} \, z^{n - n'}
e^{\xi ({\bf t}-{\bf t'}, z)} 
\tau \Bigl( n - 1, {\bf t} - [z^{-1}] \Bigr)
\tau \Bigl( n' + 1, {\bf t'} + [z^{-1}] \Bigr) 
}
\\ \\
\phantom{aaaaa}
\displaystyle{+ \,
\oint_{C_{\infty}} \frac{dz}{z^2} \, z^{-(n-n')}
e^{-\xi ({\bf t}-{\bf t'}, z)} 
\tau \Bigl( n + 1, {\bf t} + [z^{-1}] \Bigr)
\tau \Bigl( n' - 1, {\bf t'} - [z^{-1}] \Bigr) \, =0.
}
\end{array}
\eeq
Here, the contour $C_{\infty}$ and $\xi(\textbf{t}, z)$ are the 
same as in the KP and mKP cases. This equation holds 
for all ${\bf t}, \, {\bf t'}$ and all $n,\, n'$ such that
\beq\label{dkp-parity}
n - n' \in 2\ZZ,
\eeq
i.e., $n, n'$ should be either both even or both 
odd\footnote{If $n - n' \in 2\z + 1$, 
a non-zero right-hand side arises in equation (\ref{dkp-bil-rel}),
see equation (\ref{L1}) below.
For arbitrary $n, n'$ it defines the large BKP hierarchy 
considered in Section 
\ref{section:largeBKP}.}.}
Note that the first line in (\ref{dkp-bil-rel})
is the same as the left-hand side of the corresponding equation 
(\ref{mkp1}) for the mKP hierarchy. The second line of
(\ref{dkp-bil-rel}) is obtained from the first  one 
by the interchange $(n, {\bf t}) \leftrightarrow (n',{\bf t'}).$

In terms of free fermions, the tau-function 
satisfying (\ref{dkp-bil-rel}) is represented as the expectation value
\beq\label{tau-dkp}
\tau (n, {\bf t})=
\lbr n \bigl | \, e^{J({\bf t})} g \bigr | \, 0 \rbr,
\eeq
where the Clifford group element $g$ is now of the following general form:
\beq\label{dkp-group-element}
g = 
\exp \Bigl( 
\sum_{i,j \in \z} A_{ij} \psi_{i}\psi_{j}^{*}
+
\sum_{i,j \in \z} B_{ij} \psi_{i}\psi_{j}
+
\sum_{i,j \in \z} C_{ij} \psi_{i}^{*}\psi_{j}^{*}
 \Bigr).
\eeq
Comparing with the mKP case, where $g$ carries the 
definite charge 0 (see (\ref{g01})),
the more general Clifford group element (\ref{dkp-group-element}) 
does not have a definite charge,
only the even parity of the charge is fixed.
The representation (\ref{tau-dkp}) makes it evident that
$\tau (n, {\bf t})=0$ for odd $n$. This fact allows one to
extend equation (\ref{dkp-bil-rel}) to all $n, n'$ simply setting
$\tau (n, {\bf t})=0$ for odd $n$, then this equation holds
for all $n, n'$ but is non-trivial only for odd $n, n'$, 
otherwise it becomes the identity $0=0$.

Contrary to the case of the KP, mKP and small BKP hierarchies,
where the simplest tau-function is just a constant, the simplest
solution to (\ref{dkp-bil-rel}) is not so obvious. It is given
in the following proposition.

\begin{proposition}\label{proposition:simplest}
The tau-function
\beq\label{simplest}
\tau(n, {\bf t}) = 
\exp \Bigl(\frac{1}{2}\sum_{k \geq 1} k t_{k}^{2}\Bigl)
\eeq
solves equation (\ref{dkp-bil-rel}).
\end{proposition}

\noindent
We postpone the proof of this proposition 
till Section \ref{section:multiDKP}, where a similar statement
is proved for the more general multi-component 
case.

Similarly to the mKP case, the most general Miwa substitution is
\beq\label{mkp2a}
\left \{\begin{array}{l}
\displaystyle{
n - n' = P^{+} - P^{-},}
\\ \\
\displaystyle{
{\bf t} - {\bf t}' =
\sum_{i = 1}^{P^{+}}[a_i^{-1}]
-\sum_{k = 1}^{P^{-}}[b_k^{-1}],}
\end{array}\right.
\eeq
where $P^{+} - P^{-}$ is even.
The points $a_{i}, b_{j}$ are again assumed 
to be distinct. With this substitution,
the residue calculus in (\ref{dkp-bil-rel}) leads to the following 
general multi-point Hirota-Miwa relation:
\beq\label{dkp-gen-H-M}
\begin{array}{l}
\displaystyle{
\sum_{s = 1}^{P^{+}} \prod_{i = 1, \neq s}^{P^{+}} E^{-1}(a_{s}, a_{i}) \,
\prod_{k = 1}^{P^{-}} E(a_{s}, b_{k})\,
\tau \Bigl( 
n + P^{+} \! -\! 1, {\bf t} + \sum_{i\neq s}^{P^{+}}
[a_{i}^{-1}] \Bigr)}
\\ \\ \phantom{aaaaaaaaaaaaaaaaaaaaaaaaa}
\displaystyle{\times \, \,
\tau \Bigl(
n + P^{-} \! +\! 1, {\bf t} + [a_{s}^{-1}]+ \sum_{k=1}^{P^{-}}
[b_{k}^{-1}] \Bigr) }
\\ \\ + \,
\displaystyle{
\sum_{s = 1}^{P^{-}} \prod_{i = 1, \neq s}^{P^{-}} E^{-1}(b_{s}, b_{i}) \,
\prod_{k = 1}^{P^{+}} E(b_{s}, a_{k})\,
\tau \Bigl( 
n + P^{-} \! -\! 1, {\bf t} + \sum_{i\neq s}^{P^{-}}
[b_{i}^{-1}] \Bigr)}
\\ \\ \phantom{aaaaaaaaaaaaaaaaaaaaaaaaa}
\displaystyle{\times \, 
\tau \Bigl(
n + P^{+} \! +\! 1, {\bf t} + [b_{s}^{-1}] + \sum_{k = 1}^{P^{+}}
[a_{k}^{-1}] \Bigr)}
=
0,
\end{array}
\eeq
where the function $E(a,b)$ is the same as in (\ref{mkp4}).
Clearly, the most significant difference compared to 
(\ref{mkp3}) is that not only each $a_{i}$ but also 
each $b_{j}$ produces now a term in (\ref{dkp-gen-H-M}):
the $a_i$'s (respectively, $b_j$'s) give rise to the first 
(respectively, second) sum in (\ref{dkp-gen-H-M}).
Since possible values of $n$ and $n'$ are restricted
by the parity condition (\ref{dkp-parity}), 
the total number of points (and thus the total number
of terms in (\ref{dkp-gen-H-M})), which is 
$P^+ +P^-$, is even. 

\begin{remark}
If $P^-=0$, then the second sum in (\ref{dkp-gen-H-M}) is absent
and the equation formally coincides with the 
general Hirota-Miwa relation (\ref{mkp3}) for the mKP hierarchy
(with $P^-=0$). However, the important difference is that
in the DKP case $P^+$ must be even due to the 
parity condition whereas for mKP there is no such restriction.
\end{remark}

Following the terminology from \cite{SZ25b}, 
we call (\ref{dkp-gen-H-M}) the non-degenerate 
$(P^+ +P^-)$-point relation.
The symmetry $(n, {\bf t})\leftrightarrow (n', {\bf t}')$
allows us to assume that $P^+ \geq P^-$ without loss of generality.
The simplest nontrivial cases are 4-point non-degenerate 
relations corresponding to the choices 
$(P^{+}, P^{-})=(4, 0)$, $(P^{+}, P^{-}) = (3, 1)$ and
$(P^{+}, P^{-}) = (2, 2)$. The first two possibilities 
lead to the following equations:

\noindent
\underline{$(P^{+}, P^{-})=(4, 0)$}:
\beq\label{dkp-H-M(4,0)}
\begin{array}{l}
\phantom{-}E(b, c)E(b, d)E(c, d)
\tau \Bigl( n + 2, {\bf t} + [b^{-1}] + [c^{-1}] + [d^{-1}]\Bigr)
\tau \Bigl( n, {\bf t} + [a^{-1}]\Bigr)
\\ \\
-\, E(a, c)E(a, d)E(c, d)
\tau \Bigl( n + 2, {\bf t} + [a^{-1}] + [c^{-1}] + [d^{-1}]\Bigr)
\tau \Bigl( n, {\bf t} + [b^{-1}]\Bigr)
\\ \\
-\,E(a, b)E(a, d)E(b, d)
\tau \Bigl( n + 2, {\bf t} + [a^{-1}] + [b^{-1}] + [d^{-1}]\Bigr)
\tau \Bigl( n, {\bf t} + [c^{-1}]\Bigr)
\\ \\
+\, E(a, b)E(a, c)E(b, c)
\tau \Bigl( n + 2, {\bf t} + [a^{-1}] + [b^{-1}] + [c^{-1}]\Bigr)
\tau \Bigl( n, {\bf t} + [d^{-1}]\Bigr)
\, =0,
\end{array}
\eeq

\noindent
\underline{$(P^{+}, P^{-})=(3, 1)$}:
\beq\label{dkp-H-M(3,1)}
\begin{array}{l}
\phantom{-}E(b, c)E(a, d)
\tau \Bigl( n, {\bf t} + [b^{-1}] + [c^{-1}]\Bigr)
\tau \Bigl( n, {\bf t} + [a^{-1}] + [d^{-1}]\Bigr)
\\ \\
-\, E(b, d)E(a, c)
\tau \Bigl( n, {\bf t} + [a^{-1}] + [c^{-1}]\Bigr)
\tau \Bigl( n, {\bf t} + [b^{-1}] + [d^{-1}]\Bigr)
\\ \\
-\,E(a, b)E(c, d)
\tau \Bigl( n, {\bf t} + [a^{-1}] + [b^{-1}] \Bigr)
\tau \Bigl( n, {\bf t} + [c^{-1}] + [d^{-1}]\Bigr)
\\ \\
+\, E(a, b)E(a, c)E(b, c)E(d, a)E(d, b)E(d, c) 
\\ \\ \phantom{aaaaaaaaaaa} \times
\tau \Bigl( n + 2, {\bf t} + [a^{-1}] + [b^{-1}] + 
[c^{-1}] + [d^{-1}]\Bigr)
\tau \Bigl( n-2, {\bf t}\Bigr)
\, =0.
\end{array}
\eeq 

\noindent
The choice 
$(P^{+}, P^{-}) = (2, 2)$ leads to the relation 
that coincides with (\ref{dkp-H-M(4,0)}) after rearranging the terms,
so there are only two essentially different 
non-degenerate 4-point relations.
Equation (\ref{dkp-H-M(4,0)}) 
coincides with the Hirota-Miwa relation for the mKP hierarchy
obtained from (\ref{mkp3}) at $P^+=4$, $P^-=0$.

\begin{remark}
Equation (\ref{dkp-H-M(3,1)}) is similar 
to the equation (\ref{mkp5}) for the KP tau-function
but differs from it by having the fourth term in the 
left-hand side which is absent
in the KP case.
\end{remark}

\subsubsection{Dispersionless limit of DKP:
uni\-for\-mi\-za\-ti\-on via elliptic functions}
\label{section:dDKP}

The dispersionless limit 
can be performed similarly to the mKP case
(Section \ref{section:ddmKP}). The tau-function is replaced
by the $F$-function defined in the standard way:
\beq\label{F00}
F = F(t_{0}, {\bf t}) =
\lim\limits_{\hbar \to 0}
\hbar^2 \log \tau \Bigl (\hbar^{-1} t_0, \hbar^{-1}{\bf t}  \Bigr ).
\eeq
The $\nabla$-operator is introduced by the formula 
\beq\label{nabla-dkp-def}
\nabla(z)
=
\p_{t_0} + \sum_{k \geq 1}\frac{z^{-k}}{k}\, \p_{t_{k}}
\eeq
(cf. (\ref{mkp7})).
The result of the limit for equation ({\ref{dkp-gen-H-M}}) is
\beq\label{dkp-gen-disp-H-M}
\begin{array}{l}
\displaystyle{
\sum_{s = 1}^{P^{+}}
\left(
\prod_{i = 1, \neq s}^{P^{+}} E(a_{s}, a_{i})
e^{ \nabla (a_{i}) \nabla (a_{s}) F}
\right )^{-1}
\left(
\prod_{k=1}^{P^{-}} E(a_{s}, b_{k}) 
e^{ \nabla (a_{s})  \nabla(b_{k})F}
\right)}
\\ \\\!\!\!  +\,\,
\displaystyle{
\sum_{s = 1}^{P^{-}}
\left(
\prod_{i = 1, \neq s}^{P^{-}} E(b_{s}, b_{i})
e^{ \nabla (b_{i}) \nabla (b_{s}) F}
\right )^{-1}
\left(
\prod_{k = 1}^{P^{+}} E(b_{s}, a_{k}) 
e^{ \nabla (b_{s})  \nabla(a_{k})F}
\right) = 0.}
\end{array}
\eeq
An important difference compared to the mKP case 
(see (\ref{mkp10})) is presence of the second sum written in
the second line and having the same structure as the first one
but with the exchange $\{a_i\} \leftrightarrow \{b_k\}$.
The other difference is the parity condition:
$P^{+} + P^{-} \in 2\ZZ$, i.e., the total number of terms 
has to be even.

As in the mKP case, in order to recover the hidden algebraic
curve, we need to consider the dispersionless limits
of the two 4-point non-degenerate 
relations (\ref{dkp-H-M(4,0)}), (\ref{dkp-H-M(3,1)}). 
They are:
\beq\label{40c}
\begin{array}{l}
\phantom{-}
E(b,c) E(b,d) E(d,c)
e^{\bigl( \nabla(b) \nabla(c) + \nabla(b)
\nabla(d) + \nabla(c)\nabla(d) \bigr) F}
\\ \\
-\, 
E(a, c) E(a, d) E(d, c)
e^{\bigl( \nabla(a) \nabla(c) + \nabla(a)
\nabla(d) + \nabla(d)\nabla(c) \bigr) F}
\\ \\
-\, 
E(a, b) E(a, d) E(b, d)
e^{\bigl( \nabla(a) \nabla(b) + \nabla(a)
\nabla(d) + \nabla(b)\nabla(d) \bigr) F}
\\ \\
+\, 
E(a, b)E(a, c)E(b, c)
e^{\bigl( \nabla(a) \nabla(b) + \nabla(a)
\nabla(c) + \nabla(b)\nabla(c) \bigr) F}
\, =0
\end{array}
\eeq
and
\beq\label{31c}
\begin{array}{l}
\phantom{-}
E(b,c) E(a,d)
e^{\bigl ( \nabla(b)\nabla(c) + \nabla(a)\nabla(d)\bigr )F}
-\, 
E(b,d) E(a,c)
e^{\bigl ( \nabla(b)\nabla(d) + \nabla(a)\nabla(c)\bigr )F}
\\ \\
+\, 
E(a,b) E(c,d)
e^{\bigl ( \nabla(a)\nabla(b) + \nabla(c)\nabla(d)\bigr )F}
+\, E(a,b) E(a,c) E(b,c) E(d,a) E(d,b) E(d,c)
\\ \\
\times 
\,\, e^{\bigl ( 
\nabla(a) \nabla(b) +
\nabla(a) \nabla(c) +
\nabla(a) \nabla(d) +
\nabla(b) \nabla(c) + 
\nabla(b) \nabla(d) +
\nabla(c) \nabla(d)
\bigr )F}
\, =0.
\end{array}
\eeq

To proceed, it is convenient to introduce the
$g$-function:
\beq\label{g-def-dkp}
g(a,b) = E(a,b) e^{\nabla(a) \nabla(b) F},
\eeq
in terms of which equations (\ref{40c}), (\ref{31c})
can be written in the following more compact form:
\beq\label{eq1}
\begin{array}{l}
\displaystyle{
g(c,d)\Bigl (g(a,c)g(a,d)-g(b,c)g(b,d)\Bigr )
=\, g(a,b)\Bigl (g(a,d)g(b,d)-g(a,c)g(b,c)\Bigr ),
}
\end{array}
\eeq
and
\beq\label{eq2}
\begin{array}{l}
\displaystyle{
g(a,b)g(c,d)\Bigl (g(a,c)g(b,c)g(a,d)g(b,d)-1\Bigr )
=\, g(b,c)g(a,d)-g(b,d)g(a,c).
}
\end{array}
\eeq
Also, the functions
\beq\label{wp}
\begin{array}{l}
w(a)=g(a,c)\Bigr |_{c^{-1}\to 0}=a^{-1}e^{\nabla (a)\p_{t_0}F},
\\ \\
\displaystyle{
p(a)=-\p_{c^{-1}}\log g(a,c)\Bigr |_{c^{-1}\to 0}=
a-\nabla (a)\p_{t_1}F,
}
\end{array}
\eeq
will play an important role: they satisfy a polynomial equation
that defines the dynamical curve.

\begin{theorem}
The functions $w(z)$ and $p(z)$ defined by (\ref{wp}) satisfy
the equation
\beq\label{ddkp-curve-ep}
R^{2} \Bigl( w^{2}(z) + w^{-2}(z) \Bigr) 
= p^{2}(z) + V,
\eeq
where 
\beq\label{R111}
R=e^{\p_{t_0}^2F},
\quad
V = (\p_{t_0}\p_{t_{1}}F)^{2} +2 \p_{t_{1}}^{2} F -
\p_{t_0} \p_{t_{2}}F.
\eeq
\end{theorem}

\noindent
{\it Proof.}
Multiplying equations (\ref{eq1}) and (\ref{eq2}) (the left-hand side
of the former by the left-hand side of the latter, and the same 
for their right-hand sides) and reorganizing the terms,
we get the relation
\beq\label{ddkp-sep-vars}
\begin{array}{l}
\displaystyle{
\frac
{\Bigl( 1 + g^{2}(a, c) g^{2}(a, d)\Bigl)g^{2}(d, c)
- \Bigl( g^{2}(a, c) + g^{2}(a, d) \Bigr)}
{g(a, c) g(a, d)}}
\\ \\ =
\displaystyle{
\frac
{\Bigl(1 +g^{2}(b, c) g^{2}(b, d)\Bigl)g^{2}(d, c)
-\Bigl( g^{2}(b, c) + g^{2}(b, d) \Bigl)}
{g(b, c) g(b, d)}},
\end{array}
\eeq
which means separation of the $a$ and $b$ variables. Indeed,
the left-hand side of (\ref{ddkp-sep-vars}) is a function of $a$ while the 
right-hand side is a function of $b$. To represent this separation
in a more explicit form, we consider the limit
$c^{-1}, d^{-1} \to 0$. To perform the limit, we rewrite
equation
(\ref{ddkp-sep-vars}) in the
following equivalent form:
\beq\label{ddkp-sep-vars1}
\begin{array}{l}
\displaystyle{
g(a,c)g(a,d) +(g(a,c)g(a,d))^{-1} -g^{-2}(c,d)
\Bigl ( \frac{g(a,c)}{g(a,d)} +\frac{g(a,d)}{g(a,c)}\Bigr )
}
\\ \\ =\,
\displaystyle{
g(b,c)g(b,d) +(g(b,c)g(b,d))^{-1} -g^{-2}(c,d)
\Bigl ( \frac{g(b,c)}{g(b,d)} +\frac{g(b,d)}{g(b,c)}\Bigr )
}.
\end{array}
\eeq
The limit $c^{-1}, d^{-1}$ in (\ref{ddkp-sep-vars1}) is singular.
To resolve the singularity, we need some preparations.
Using the definition (\ref{wp}), we have:
\beq\label{wp1}
\begin{array}{l}
\displaystyle{
\left.
\p_{c^{-1}}\Bigl (\frac{g(c,a)}{g(d,a)}\Bigr )\right |_{c^{-1}, d^{-1}
\to 0} \! =-p(a),
\quad
\left.
\p_{c^{-1}}\Bigl (\frac{g(d,a)}{g(c,a)}\Bigr )\right |_{c^{-1}, d^{-1}
\to 0} \! =p(a),}
\end{array}
\eeq
and, moreover,
\beq\label{wp2}
\begin{array}{ll}
&\displaystyle{ \left. \p^2_{c^{-1}}\Bigl (\frac{g(c,a)}{g(d,a)}+
\frac{g(d,a)}{g(c,a)}\Bigr )\right |_{c^{-1}, d^{-1}
\to 0}}
\\ & \\
=& \displaystyle{\left.
\p_{c^{-1}}\Bigl (-\frac{\p_{c^{-1}}g(c,a)}{w(a)}+
\frac{w(a)\p_{c^{-1}}g(c,a)}{g^2 (c,a)}\Bigr )
\right |_{c^{-1}\to 0}}
\\ & \\
=& \displaystyle{
\left.
\left (\Bigl (\frac{g(a,c)}{w(a)}+\frac{w(a)}{g(a,c)}\Bigr )
\Bigl (\p_{c^{-1}}\log g(c,a)\Bigr )^2 
+
\Bigl (\frac{w(a)}{g(c,a)}-\frac{g(c,a)}{w(a)}\Bigr )
\p^2_{c^{-1}}\log g(c,a)
\right )
\right |_{c^{-1}\to 0}
}
\\ & \\
=& 2\Bigl (\p_{c^{-1}}\log g(c,a)\Bigr )^2 
\\ & \\
=& 2p^2(a).
\end{array}
\eeq
Now we can resolve the singularity by means of equation (\ref{wp2}):
$$
\begin{array}{ll}
& \displaystyle{\left.
g^{-2}(c,d)\Bigl (\frac{g(a,c)}{g(a,d)} + \frac{g(a,d)}{g(a,c)}-2
\Bigr )\right |_{c^{-1}, d^{-1} \to 0}}
\\ & \\
=& \displaystyle{\left. \frac{e^{-2\p_{t_0}^2F}}{(c^{-1}-d^{-1})^2}
\Bigl (\frac{g(a,c)}{g(a,d)} + \frac{g(a,d)}{g(a,c)}-2
\Bigr )\right |_{c^{-1}, d^{-1} \to 0}}
\\ & \\
=& \displaystyle{\left. \frac{1}{2}e^{-2\p_{t_0}^2F}\p_{c^{-1}}^2
\Bigl (\frac{g(a,c)}{g(a,d)} + \frac{g(a,d)}{g(a,c)}
\Bigr )\right |_{c^{-1}, d^{-1} \to 0}}
\\ & \\
=&e^{-2\p_{t_0}^2F}p^2(a).
\end{array}
$$
Therefore, equation (\ref{ddkp-sep-vars1}) yields:
$$
e^{2\p_{t_0}^2F}\Bigl (w^2(a)+w^{-2}(a)\Bigr )-p^2(a)=
e^{2\p_{t_0}^2F}\Bigl (w^2(b)+w^{-2}(b)\Bigr )-p^2(b).
$$
We conclude from it that the expression
$
e^{2\p_{t_0}^2F}\Bigl (w^2(z)+w^{-2}(z)\Bigr )-p^2(z)
$
does not depend on $z$ and is equal to 
a constant\footnote{It is a constant as a function of $z$ but 
may depend on times.} which we denote 
as $-V$. It can be found by considering the limit 
$z\to \infty$.
\square

We see that the functions $w(z)$ and $p(z)$ are constrained 
by equation (\ref{ddkp-curve-ep}) which defines an elliptic curve.
It is the dynamical curve in the case under consideration.
The real benefit of this curve can be obtained as a result of its uniformization.

\begin{proposition}\cite{AZ14}
The elliptic curve (\ref{ddkp-curve-ep}) can be uniformized
by means of elliptic functions:
\beq\label{mkp14a}
w(z) = \frac{\theta_1 (u(z) | \tau)}{\theta_4(u(z) | \tau)},
\quad
p(z) = \gamma \theta_{4}^{2}(0 | \tau)
\frac
{\theta_2(u(z) | \tau)\theta_3(u(z) | \tau)}
{\theta_1(u(z) | \tau)\theta_4(u(z) | \tau)},
\eeq
where $\theta_{i}(u | \tau)$ are Jacobi 
theta-functions (see Appendix B) 
depending on the modular parameter $\tau$, $u(z)$ is some
function of $z$ and $\gamma$ is any $z$-independent constant.
The constants $R$ and $V$ in (\ref{ddkp-curve-ep})
are given by the following formulas:
\beq \label{RV}
R = \gamma \theta_{2}(0| \tau)\theta_{3}(0| \tau),
\qquad
V = \gamma^{2}(\theta_{2}^{4}(0 | \tau) + \theta_{3}^{4}(0 | \tau)) \, .
\eeq
\end{proposition}

\noindent
For the proof see Appendix C.
In this parametrization, the equation of the curve is satisfied
identically due to the identity
\beq\label{identity}
\theta_4^4(0)\frac{\theta_2^2(u)\theta_3^2(u)}{\theta_1^2(u)
\theta_4^2(u)}=\theta_2^2(0) \theta_3^2(0) \left (
\frac{\theta_4^2(u)}{\theta_1^2(u)}+
\frac{\theta_1^2(u)}{\theta_4^2(u)}\right )-
\Bigl (\theta_2^4(0)+\theta_3^4(0)\Bigr ).
\eeq
Hereafter, we do not indicate the dependence on 
the modular parameter $\tau$ explicitly, if this does not lead
to a misunderstanding.

At this stage $\gamma$ is an arbitrary parameter but we will see 
that it can not be put equal to a fixed number because it is a
dynamical variable, as well as the modular parameter $\tau$:
$\gamma =\gamma ({\bf t})$, $\tau =\tau ({\bf t})$.
The function $u(z)$ 
is the generating function of dynamical variables. 
Let us normalize it
by the condition $u(\infty) = 0$ and assume that its
expansion near $\infty$ is of the form
\beq\label{u(z)}
u(z, {\bf t}) = 
\frac{c_{1}({\bf t})}{z} + \frac{c_{2}({\bf t})}{z^{2}} + \dots \, ,
\eeq
where $c_{i}({\bf t})$ are dynamical variables. 

\begin{proposition}
The coefficient 
$c_{1}({\bf t})$ is connected with the variable 
$\gamma = \gamma({\bf t})$ from (\ref{mkp14a}) as follows:
\beq
c_{1}({\bf t}) = \frac{\gamma({\bf t})}{\pi} \, .
\eeq
\end{proposition}

\noindent
{\it Proof.}
To see this, one should expand 
$p(z)$ given by (\ref{mkp14a}), 
as $z^{-1} \to 0$ and find coefficient at the leading term, which
should be 1 according to (\ref{wp}). In the calculation,
identity (\ref{theta1prime}) from Appendix B is used.
\square

\begin{theorem}\cite{AZ14}
In the elliptic parametrization equations (\ref{eq1}) and
(\ref{eq2}) are equivalent to the single equation
\beq\label{g(a,b)1}
(a^{-1}-b^{-1})e^{\nabla (a)\nabla (b)F}=
\frac{\theta_1(u(a)-u(b))}{\theta_4(u(a)-u(b))}.
\eeq
\end{theorem}

\noindent
{\it Proof.}
After the uniformization, equations (\ref{eq1}) and
(\ref{eq2}) become equivalent and we can use any of them.
From (\ref{eq1}) we have:
$$
g(a,b)=g(c,d) \, \frac{g(a,c)g(a,d)-g(b,c)g(b,d)}{g(a,d)g(b,d)
-g(a,c)g(b,c)}.
$$
Letting $c, d\to \infty$, as before, we obtain:
$$
\left.
\vphantom{\frac{a}{b}}
\frac{g(a,d)g(b,d) -g(a,c)g(b,c)}{g(c,d)}\right |_{c^{-1}, d^{-1}\to 0}
=e^{-\p_{t_0}^2F}w(a)w(b)\Bigl (p(a)+p(b)\Bigr ),
$$
hence
\beq\label{g(a,b)}
g(a,b)=e^{\p_{t_0}^2F}w^{-1}(a)w^{-1}(b)\,
\frac{w^2(a)-w^2(b)}{p(a)+p(b)}.
\eeq
Plugging here the elliptic parametrization (\ref{mkp14a}),
we obtain, taking into account (\ref{RV}):
\beq\label{g(a,b)1a}
g(a,b)=(a^{-1}-b^{-1})e^{\nabla (a)\nabla (b)F}=
\frac{\theta_1(u(a)-u(b))}{\theta_4(u(a)-u(b))}.
\eeq
\square

\begin{remark}
The elliptic parametrization of the function $w(a)$ (\ref{mkp14a}) 
is included in (\ref{g(a,b)1}) as the limiting case
$b\to \infty$.
\end{remark}

Remarkably, the substitution (\ref{g(a,b)1}) allows one to solve
the whole (dispersionless) hierarchy. 

\begin{theorem}
The equation (\ref{g(a,b)1}) is equivalent to the
whole dDKP hierarchy defined by (\ref{dkp-gen-disp-H-M})
for all $P^+$, $P^-$ satisfying the parity condition.
\end{theorem}

\noindent
{\it Proof.}
In terms
of the function $g(a,b)$ the general equation (\ref{dkp-gen-disp-H-M})
has the form
\beq\label{gab}
\sum_{s = 1}^{P^{+}}
\Bigl (
\prod_{i = 1, \neq s}^{P^{+}} g(a_{s}, a_{i})
\Bigr )^{-1}
\Bigl (
\prod_{k=1}^{P^{-}} g(a_{s}, b_{k}) 
\Bigr ) +
\sum_{s = 1}^{P^{-}}
\Bigl (
\prod_{i = 1, \neq s}^{P^{-}} g(b_{s}, b_{i})
\Bigr )^{-1}`
\Bigl (
\prod_{k = 1}^{P^{+}} g(b_{s}, a_{k}) 
\Bigr ) = 0.
\eeq
The elliptic parametrization (\ref{g(a,b)1}) turns it
into an identity for elliptic functions:
\beq\label{id1(ddkp)}
\sum_{s = 1}^{P^{+}}
\prod_{{\scriptsize \begin{array}{l}i = 1\\ i \neq s \end{array}}}^{P^{+}}
\frac{\theta_4 (u_{i} - u_{s})}{\theta_1 (u_{i}-u_{s})}\,
\prod_{k = 1}^{P^{-}}
\frac{\theta_1 (u_{s} - v_{k})}{\theta_4 (u_{s} - v_{k})}
+ \sum_{s = 1}^{P^{-}}
\prod_{{\scriptsize \begin{array}{l}m = 1\\ m \neq s \end{array}}}^{P^{-}}
\frac{\theta_4 (v_{m} - v_{s})}{\theta_1 (v_{m} - v_{s})}\,
\prod_{l=1}^{P^{-}}
\frac{\theta_1 (v_{s} - u_{l})}{\theta_4 (v_{s} - u_{l})} = 0.
\eeq
Here $u_i =u(a_i)$, $v_k=u(b_k)$ can be regarded as arbitrary
(distinct) complex numbers.
Shifting the $v$-variables as $v_{i} \to v_{i} + \frac{\tau}{2}$,
we rewrite this in the following equivalent form:
\beq\label{id2(ddkp)}
\sum_{s = 1}^{P^{+}}
\prod_{{\scriptsize \begin{array}{l}i = 1\\ i \neq s \end{array}}}^{P^{+}}
\frac{\theta_4 (u_{i} - u_{s})}{\theta_1 (u_{i} - u_{s})}\,
\prod_{k = 1}^{P^{-}}
\frac{\theta_4 (u_{s} - v_{k})}{\theta_1 (u_{s} - v_{k})}
 + \sum_{s = 1}^{P^{-}}
\prod_{{\scriptsize \begin{array}{l}m = 1\\ m \neq s \end{array}}}^{P^{-}}
\frac
{\theta_4 (v_{m} -v_{s})}{\theta_1 (v_{m} - v_{s})}\,
\prod_{l = 1}^{P^{-}}
\frac{\theta_4 (v_{s} - u_{l})}{\theta_1 (v_{s} -u_{l})} = 0.
\eeq
To see that this holds identically for all 
$\{u_{i}\}$, $\{v_{k}\}$, it is enough to notice that the left-hand 
side is proportional to the sum of residues of the following
elliptic function with periods $1, \tau$:
\beq\label{id4(ddkp)}
f(u) = 
\prod_{i = 1}^{P^{+}} \frac{\theta_4(u - u_{i})}{\theta_1(u - u_{i})}
\prod_{j = 1}^{P^{-}} \frac{\theta_4(u - v_{j})}{\theta_1(u - v_{j})},
\eeq
which is zero. (Note that the condition that $P^{+}+P^{-}\in 2\ZZ_{+}$ is 
important for this.)
\square

We have shown that equation (\ref{g(a,b)1}) represents
the whole hierarchy in the elliptic form.
Some words on its meaning are in order.
For brevity, denote the function 
$\theta_1(u)/\theta_4(u)$ by ${\sf sn}(u)$ and the inverse function
by ${\sf arcsn}(u)$\footnote{Strictly speaking, 
the function ${\sf sn}(u)$ defined in this way is not quite what is called
the elliptic sinus function ${\rm sn}(u)$ (see Appendic C).
However, it is very similar to it: it differs 
from it by
a common $u$-independent factor and a re-scaling of the variable.
So, the function ${\sf arcsn} (u)$ is, up to some details,
the elliptic integral
of the first kind.}. 
Then we can write: 
\beq\label{u1}
u(a)=
{\sf arcsn}  \Bigl (a^{-1}e^{\nabla (a)
\p_{t_0}F}\Bigr ).
\eeq
So, equation (\ref{g(a,b)1}) can be written as
\beq\label{t161}
(a^{-1}-b^{-1})e^{\nabla (a)\nabla (b)F}
={\sf sn}  \Bigl ( {\sf arcsn} \, (a^{-1}
e^{\nabla (a)\p_{t_0}F})-
{\sf arcsn} \, (b^{-1}
e^{\nabla (b)\p_{t_0}F})\Bigr ).
\eeq
Therefore, as equation (\ref{t161}) shows, the general 
second order derivatives $\p_{t_m}\p_{t_n}F$ with all
$m,n\geq 0$ in the left-hand side of
(\ref{t161}) are expressed through the
particular derivatives $\p_{t_m}\p_{t_0} F$ with $m\geq 0$
(which are in the right-hand side). Such structure of the
equations is common for dispersionless hierarchies. A similar
pattern we observe for the dispersionless KP and mKP hierarchies
(see equations (\ref{kp12}) and (\ref{mkp13d})). The only difference
is that the explicit expressions for general derivatives
through the particular ones in the
present case are essentially more
complicated. 

At last, we should determine the elliptic modular
parameter $\tau$. 

\begin{proposition}\label{proposition:tau}
The elliptic modular
parameter $\tau =\tau ({\bf t})$ is a function of times
implicitly determined from the equation
\beq \label{RV1}
 \frac{\theta_2^2(0|\tau )}{\theta_3^2(0|\tau )}+
\frac{\theta_3^2(0|\tau )}{\theta_2^2(0|\tau )}
=e^{-2\p_{t_0}^2F}
\Bigl ((\p_{t_0}\p_{t_{1}}F)^{2} +2 \p_{t_{1}}^{2} F -
\p_{t_0} \p_{t_{2}}F\Bigr ).
\eeq
\end{proposition}

\noindent
{\it Proof.}
From (\ref{RV}) we see that the ratio $V/R^2$ depends only on 
$\tau$ and equals
the left-hand side of (\ref{RV1}). From (\ref{R111}) it follows
that the same ratio equals also its right-hand side.
\square

\subsubsection{The $F_1$-function}
\label{section:F1-DKP}

In the next section we will need the next-to-leading term $F_1$
of the $\hbar$-expansion
\beq\label{F001}
F(t_{0}, {\bf t}; \hbar ) =
\hbar^2 \log \tau \Bigl (\hbar^{-1} t_0, \hbar^{-1}{\bf t}  \Bigr )
=F_0 +\hbar F_1 + O(\hbar^2), \quad \hbar \to 0.
\eeq

\begin{proposition}
The function $F_1$ in the expansion (\ref{F001}) satisfies 
the following homogeneous linear equation:
\beq\label{exp4}
\begin{array}{l}
\displaystyle{
\sum_{s=1}^{P^+}
\prod_{{\scriptsize \begin{array}{l}i=1\\ i\neq s \end{array}}}^{P^+}
E^{-1}(a_s, a_i)
\prod_{k=1}^{P^-} E(a_s, b_k)
e^{\nabla (a_s)(S^- -S^+ +\nabla (a_s))F_0}
\nabla (a_s)(S^-\!  -\! S^+\!  +\! \nabla (a_s))F_1
}
\\ \\
\displaystyle{
+\,\, \sum_{s=1}^{P^-}
\prod_{{\scriptsize \begin{array}{l}k=1\\ k\neq s \end{array}}}^{P^-}
E^{-1}(b_s, b_k)
\prod_{i=1}^{P^+} E(b_s, a_i)
e^{\nabla (b_s)(S^+ -S^- +\nabla (b_s))F_0}
\nabla (b_s)(S^+ \! -\! S^- \! +\! \nabla (b_s))F_1  =  0,
}
\end{array}
\eeq
where the operators $S^{\pm}$ are 
$\displaystyle{
S^+ = \sum_{i=1}^{P^+} \nabla (a_i), \;
S^- = \sum_{i=1}^{P^-} \nabla (b_i)}$
(cf. (\ref{Spm})) and 
the function $F_0$ satisfies equation (\ref{dkp-gen-disp-H-M}).
\end{proposition}

\noindent
{\it Proof.}
Equation (\ref{exp4}) can be obtained by expanding the $\hbar$-version
of equation (\ref{dkp-gen-H-M}) up to the first order in $\hbar$
in the same way as this was done in the 
simpler example of the $\hbar$-KP hierarchy in Section \ref{section:F1}
(see equation (\ref{exp3})). 
\square

In the full analogy with the dKP case considered 
in Section \ref{section:F1}, the function 
\beq\label{F1v}
F_1=\p_{v}F_0(t_0, {\bf t}; v ), 
\eeq
where
$v$ is any continuous parameter of solutions to (\ref{dkp-gen-disp-H-M}),
satisfies equation (\ref{exp4}). Indeed, the
left-hand side (\ref{exp4}) becomes the $v$-derivative of 
(\ref{dkp-gen-disp-H-M}), and hence it is equal to zero.

\subsection{Large BKP and its dispersionless version}

This hierarchy was introduced in \cite{KL98} under the name
``charged BKP hierarchy''. Later, the authors of \cite{OST12,OST16}
suggested to call it ``large BKP hierarchy'' as opposed to 
the ``neutral'' (or ``small'') BKP considered in 
Section \ref{section:smallBKP} (see also \cite{vLO15,book}).
Recently it was rediscovered in 
\cite{KZ22,PZ23} as a subhierarchy of
the 2D Toda lattice and called there the Toda chain of type B (B-Toda).
As such, it can be regarded as a natural integrable discretization of the
BKP hierarchy. For realization in terms of free fermions see
\cite{vLO15}.

\subsubsection{Large BKP}
\label{section:largeBKP}

The set of independent variables is the same as for the 
DKP hierarchy: $n\in \ZZ$ and ${\bf t}=\{t_1, t_2, t_3, \ldots \}$,
and the tau-function is $\tau (n, {\bf t})$.
The bilinear equation for the tau-function has the form
\beq\label{L1}
\begin{array}{ll}
&\displaystyle{
\frac{1}{2\pi i} \oint_{C_{\infty}} \frac{dz}{z^2}\, 
z^{n-n'} e^{\xi ({\bf t}-{\bf t}', z)}
\tau \Bigl ( n-1, {\bf t}-[z^{-1}]\Bigr )
\tau \Bigl ( n'+1, {\bf t}'+[z^{-1}]\Bigr )}
\\ &\\
+& \displaystyle{
\frac{1}{2\pi i} \oint_{C_{\infty}} \frac{dz}{z^2}\, 
z^{n'-n} e^{-\xi ({\bf t}-{\bf t}', z)}
\tau \Bigl ( n+1, {\bf t}+[z^{-1}]\Bigr )
\tau \Bigl ( n'-1, {\bf t}'-[z^{-1}]\Bigr )}
\\ &\\
=& \frac{1}{2} \Bigl (1-(-1)^{n-n'}\Bigr )
\tau (n, {\bf t})\tau (n', {\bf t}'),
\end{array}
\eeq
which is valid for all $n, n'$, ${\bf t}, {\bf t}'$.
Equation (\ref{L1}) has the following obvious symmetry:
\beq\label{sym1}
(n, {\bf t}) \longleftrightarrow (n', {\bf t}').
\eeq 

\begin{remark}
The simplest solution (\ref{simplest}) to the DKP hierarchy
is simultaneously a solution to the large BKP hierarchy.
\end{remark}

If $n$ and $n'$ are of the same parity (both even or both odd),
the right-hand side vanishes and (\ref{L1}) becomes the 
integral bilinear equation (\ref{dkp-bil-rel}) 
for the DKP hierarchy. More
precisely, the full set of equations of the large BKP can
be divided into three groups: 
the ``even'' sector consisting of equations that connect
tau-functions with even $n$'s (the DKP), 
the ``odd'' sector consisting of equations that connect
tau-functions with odd $n$'s (another copy of DKP) and
equations that ``intertwine'' the even and odd sectors
(they connect tau-functions $\tau (n, {\bf t})$,
$\tau (m, {\bf t})$ with $n-m\in 2\ZZ +1$). The latter 
set of equations can be written in another form which is
more suitable for the most general dispersionless limit
discussed below in the next subsection. Namely, let the
tau-function be 
$$
\tau (n, {\bf t}) \quad \mbox{for even $n$ and
$\sigma (n, {\bf t})$ for odd $n$} .
$$ 
In this notation, 
equation (\ref{L1}) for $n$ even, $n'$ odd acquires the form
\beq\label{L1c}
\begin{array}{ll}
&\displaystyle{
\frac{1}{2\pi i} \oint_{C_{\infty}} \frac{dz}{z^2}\, 
z^{n-n'} e^{\xi ({\bf t}-{\bf t}', z)}
\sigma \Bigl ( n-1, {\bf t}-[z^{-1}]\Bigr )
\tau \Bigl ( n'+1, {\bf t}'+[z^{-1}]\Bigr )}
\\ &\\
+& \displaystyle{
\frac{1}{2\pi i} \oint_{C_{\infty}} \frac{dz}{z^2}\, 
z^{n'-n} e^{-\xi ({\bf t}-{\bf t}', z)}
\sigma \Bigl ( n+1, {\bf t}+[z^{-1}]\Bigr )
\tau \Bigl ( n'-1, {\bf t}'-[z^{-1}]\Bigr )}
\\ &\\
=& \frac{1}{2} \Bigl (1-(-1)^{n-n'}\Bigr )
\tau (n, {\bf t})\sigma (n', {\bf t}').
\end{array}
\eeq
The equation 
in the case when $n$ is odd and $n'$ even is equivalent to 
it due to the symmetry (\ref{sym1}).

As before, the substitution 
\beq\label{L2}
\left \{\begin{array}{l}
\displaystyle{
n-n'=P^+ -P^-,}
\\ \\
\displaystyle{
{\bf t}-{\bf t}' =\sum_{i=1}^{P^+}[a_i^{-1}]-
\sum_{k=1}^{P^-}[b_k^{-1}]}
\end{array}\right.
\eeq
allows one to apply residue calculus.
The case $P^+-P^-\in 2\ZZ$ corresponds to DKP and was already
considered. Here we are interested in the case
$P^+-P^-\in 2\ZZ +1$, which leads to 
equations that are specific for the
large BKP. The residue calculus yields the following general
Hirota-Miwa equation:
\beq\label{L3}
\begin{array}{l}
\displaystyle{
\sum_{s=1}^{P^+}
\prod_{{\scriptsize \begin{array}{l}i=1\\ i\neq s \end{array}}}^{P^+}
E^{-1}(a_s, a_i)
\prod_{k=1}^{P^-} E(a_s, b_k)
\tau \Bigl (n+P^+ -1,  {\bf t}+\sum_{i\neq s}^{P^+}[a_i^{-1}]\Bigr )}
\\ 
\displaystyle{\phantom{aaaaaaaaaaaaaaaaaaaaaaaaaaaaaaaaaaaaa}
\times \, \sigma \Bigl (n+P^- +1, {\bf t}+[a_s^{-1}]+\sum_{k=1}^{P^-}
[b_k^{-1}]\Bigr )}
\\ \\
\displaystyle{
+\, \sum_{s=1}^{P^-}
\prod_{{\scriptsize \begin{array}{l}k=1\\ k\neq s \end{array}}}^{P^-}
E^{-1}(b_s, b_k)
\prod_{i=1}^{P^+}E(b_s, a_i)
\tau \Bigl (n+P^+ +1, {\bf t}+[b_s^{-1}] +\sum_{i=1}^{P^+}
[a_i^{-1}], \Bigr )}
\\ 
\displaystyle{\phantom{aaaaaaaaaaaaaaaaaaaaaaaaaaaaaaaaaaaaa}
\times \, 
\sigma \Bigl (n+P^- -1,  {\bf t} +\sum_{k\neq s}^{P^-}[b_k^{-1}]
\Bigr )}
\\ \\
\phantom{}
\displaystyle{
=\sigma \Bigl (n+P^+ , {\bf t}+\sum_{i=1}^{P^+}[a_i^{-1}] \Bigr )
\, \tau \Bigl (n+P^- , {\bf t} +\sum_{k=1}^{P^-}[b_k^{-1}])\Bigr ).}
\end{array}
\eeq
This equation contains $P^+ + P^- +1$ bilinear terms.
Recall that $P=P^++P^-\geq 3$ here is odd. 

The simplest
nontrivial case of (\ref{L3}) is $P=3$ that leads to 4-term relations. 
Taking into account the symmetry (\ref{sym1}), we should consider
two cases: $(P^+ , P^-)=(3,0)$ and $(P^+ , P^-)=(2,1)$, i.e., 
\beq\label{L4}
\left \{\begin{array}{l}
\displaystyle{
n-n'=3,}
\\ \\
\displaystyle{
{\bf t}-{\bf t}' =[a_1^{-1}]+[a_2^{-1}]+[a_3^{-1}]}
\end{array}\right.  \quad  \mbox{and} \quad
\left \{\begin{array}{l}
\displaystyle{n-n'=1,}
\\ \\
{\bf t}-{\bf t}' =[a_1^{-1}]+[a_2^{-1}]-[b_1^{-1}].
\end{array}\right.
\eeq
The corresponding 3-point 4-term equations are:

\beq\label{(3,0)}
\begin{array}{l}
\phantom{a}\,
E^{-1}(a_1, a_2)E^{-1}(a_1, a_3)\tau \Bigl (n+2, {\bf t}+[a_2^{-1}]+
[a_3^{-1}]\Bigr )\tau \Bigl (n+1, {\bf t}+[a_1^{-1}]\Bigr )
\\ \\
+\, E^{-1}(a_2, a_1)E^{-1}(a_2, a_3)\tau \Bigl (n+2, {\bf t}+[a_1^{-1}]+
[a_3^{-1}]\Bigr )\tau \Bigl (n+1, {\bf t}+[a_2^{-1}]\Bigr )
\\ \\
+\, E^{-1}(a_3, a_1)E^{-1}(a_3, a_2)\tau \Bigl (n+2, {\bf t}+[a_1^{-1}]+
[a_2^{-1}]\Bigr )\tau \Bigl (n+1, {\bf t}+[a_3^{-1}]\Bigr )
\\ \\
=\, \tau \Bigl (n+3, {\bf t}+[a_1^{-1}]+
[a_2^{-1}]+ [a_3^{-1}]\Bigr )\tau (n ,{\bf t}),
\end{array}
\eeq

\noindent
and

\beq\label{(2,1)}
\begin{array}{l}
\phantom{a}\,
E^{-1}(a_1, a_2)E(a_1, b_1)\tau \Bigl (n+1, {\bf t}+[a_2^{-1}]\Bigr )
\tau \Bigl (n+2, {\bf t}+[a_1^{-1}]+[b_1^{-1}]\Bigr )
\\ \\
+\, E^{-1}(a_2, a_1)E(a_2, b_1)\tau \Bigl (n+1, {\bf t}+[a_1^{-1}]\Bigr )
\tau \Bigl (n+2, {\bf t}+[a_2^{-1}]+[b_1^{-1}]\Bigr )
\\ \\
+\, 
E(b_1, a_1)E(b_1, a_2)\tau \Bigl (n+3, {\bf t}+[a_1^{-1}]+[a_2^{-1}]
+[b_1^{-1}]\Bigr )
\tau  (n, {\bf t})
\\ \\
=\, \tau \Bigl (n+2, {\bf t}+[a_1^{-1}]+
[a_2^{-1}]\Bigr )\tau (n+1 ,{\bf t}+[b_1^{-1}]),
\end{array}
\eeq

\noindent
where we have returned to the previous notation 
$\sigma (n, {\bf t})=\tau (n, {\bf t})$ for odd $n$.
Remarkably, these two equations are actually the same: this can be
easily seen by the substitution $b_1=a_3$ and reorganizing the terms.

\subsubsection{Dispersionless limit of the large BKP hierarchy}
\label{section:dLargeBKP}

The essential difference between the large BKP hierarchy and the previous cases is that it admits more than one dispersionless versions.
One of them was 
discussed in \cite{Z24a} (where the equivalent B-Toda hierarchy
was dealt with).

As before, the first step is to rescale the times as 
$
n \to t_0 /\hbar$, ${\bf t}\to {\bf t}/\hbar$. However, 
in the present case we should
take into account the possibility that in general
the functions $\tau$ and $\sigma$ may have different limits
as $\hbar \to 0$.
To take care of this, we set
\beq\label{FG}
\tau \Bigl (\hbar^{-1}t_0, \hbar^{-1}{\bf t}  \Bigr )=
\exp \left ( \frac{1}{\hbar^2}\,
F(t_0, {\bf t}; \hbar )\right ), \quad
\sigma \Bigl (\hbar^{-1}t_0, \hbar^{-1}{\bf t}  \Bigr )=
\exp \left ( \frac{1}{\hbar^2}\,
G(t_0, {\bf t}; \hbar )\right ),
\eeq
where $F$ and $G$ a priori are different functions,
and assume that they have $\hbar$-expansions 
of the form 
$$
\begin{array}{l}
\displaystyle{
F(t_0, {\bf t}; \hbar )=F_0(t_0, {\bf t})+
\sum_{k\geq 1}F_k(t_0, {\bf t})\hbar^k,}
\\ \\
\displaystyle{
G(t_0, {\bf t}; \hbar )=G_0(t_0, {\bf t})+
\sum_{k\geq 1}G_k(t_0, {\bf t})\hbar^k},
\end{array}
$$
where all $F_k$ and $G_k$ are smooth functions of their arguments. 

Introducing the operator
$\displaystyle{
\nabla (z)=\p_{t_0} +\sum_{k\geq 1} \frac{z^{-k}}{k}\, \p_{t_k},}
$
we can write the $\hbar$-version of equation (\ref{L3}) in the form
\beq\label{L3b}
\begin{array}{l}
\displaystyle{
\sum_{s=1}^{P^+}
\prod_{{\scriptsize \begin{array}{l}i=1\\ i\neq s \end{array}}}^{P^+}
E^{-1}(a_s, a_i)
\prod_{k=1}^{P^-} E(a_s, b_k)
e^{\hbar \sum_{i\neq s}^{P^+}\nabla (a_i)}e^{\hbar^{-2}F(t_0, {\bf t})}\,
e^{\hbar \sum_{k\neq s}^{P^-}\nabla '(b_k)+\hbar
\nabla ' (a_s)}e^{\hbar^{-2}G(t_0', {\bf t'})}}
\\ \\
\displaystyle{
+\, \sum_{s=1}^{P^-}
\prod_{{\scriptsize \begin{array}{l}k=1\\ k\neq s \end{array}}}^{P^-}
E^{-1}(b_s, b_k)
\prod_{i=1}^{P^+}E(b_s, a_i)
e^{\hbar \sum_{i=1}^{P^+}\nabla (a_i)+\hbar \nabla (b_s)}
e^{\hbar^{-2}F(t_0, {\bf t})} \, 
e^{\hbar \sum_{k\neq s}^{P^-}\nabla '(b_k)}
e^{\hbar^{-2}G(t_0', {\bf t'})}}
\\ \\
\displaystyle{
=e^{\hbar \sum_{i=1}^{P^+}\nabla (a_i)}e^{\hbar^{-2}G(t_0, {\bf t})}\,
e^{\hbar \sum_{k=1}^{P^-}\nabla '(b_k)}e^{\hbar^{-2}F(t_0', {\bf t'})}
},
\end{array}
\eeq
where the operator $\nabla$ acts to the variables $t_0, {\bf t}$ 
and the operator $\nabla '$ acts to the variables $t_0', {\bf t'}$.
Expanding this equation in powers of $\hbar$ as $\hbar \to 0$, 
one can see that the limit exists only if the leading terms coincide,
i.e., $G_0=F_0$. However, in general the functions $F$ and $G$ may
differ in the next order. Taking this into account, we set
\beq\label{FGf}
G_1(t_0, {\bf t})=F_1(t_0, {\bf t})-f(t_0, {\bf t}).
\eeq
Then the $\hbar \to 0$ limit of (\ref{L3b}) reads
\beq\label{dL1}
\begin{array}{l}
\displaystyle{
\sum_{s=1}^{P^+}
\prod_{{\scriptsize \begin{array}{l}i=1\\ i\neq s \end{array}}}^{P^+}
E^{-1}(a_s, a_i)
\prod_{k=1}^{P^-} E(a_s, b_k)
\exp \!
\left (\nabla (a_s) \Bigl (\sum_{k=1}^{P^-}\nabla (b_k)
-  \sum_{i\neq s}^{P^+}\nabla (a_i)\Bigr )F_0 -\nabla (a_s)f\Bigr )\right )
}
\\ \\
\displaystyle{
+\,\, \sum_{s=1}^{P^-}
\prod_{{\scriptsize \begin{array}{l}k=1\\ k\neq s \end{array}}}^{P^-}
E^{-1}(b_s, b_k)
\prod_{i=1}^{P^+} E(b_s, a_i)
\exp \! \left (\nabla (b_s) \Bigl 
(\sum_{i=1}^{P^+}\nabla (a_i)
-  \sum_{k\neq s}^{P^-}\nabla (b_k)\Bigr )F_0 +\nabla (b_s)f\Bigr )\right )
}
\\ \\
\phantom{aaaaa}\displaystyle{
=\exp \! \left (\sum_{k=1}^{P^-}\nabla (b_k)f-
\sum_{i=1}^{P^+}\nabla (a_i)f \right ).
}
\end{array}
\eeq

Let us first consider the case $f=0$ (or $f=\mbox{const}$).
(The case of nonzero $f$ is more complicated 
and will be discussed in the next section in the more general 
context of
multi-component Pfaff hierarchies.) For $f=0$ equation
(\ref{dL1}) simplifies:
\beq\label{dL1a}
\begin{array}{l}
\displaystyle{
\sum_{s=1}^{P^+}
\prod_{{\scriptsize \begin{array}{l}i=1\\ i\neq s \end{array}}}^{P^+}
E^{-1}(a_s, a_i)
\prod_{k=1}^{P^-} E(a_s, b_k)
\exp \!
\left (\nabla (a_s) \Bigl (\sum_{k=1}^{P^-}\nabla (b_k)
-  \sum_{i\neq s}^{P^+}\nabla (a_i)\Bigr )F_0 \Bigr )\right )
}
\\ \\
\displaystyle{
+\,\, \sum_{s=1}^{P^-}
\prod_{{\scriptsize \begin{array}{l}k=1\\ k\neq s \end{array}}}^{P^-}
E^{-1}(b_s, b_k)
\prod_{i=1}^{P^+} E(b_s, a_i)
\exp \! \left (\nabla (b_s) \Bigl 
(\sum_{i=1}^{P^+}\nabla (a_i)
-  \sum_{k\neq s}^{P^-}\nabla (b_k)\Bigr )F_0\Bigr )\right )
\, = \, 1.
}
\end{array}
\eeq
The simplest
nontrivial case is $P=3$ with two 
possibilities: $(P^+, P^-)=(3,0)$ and $(P^+, P^-)=(2,1)$.
Both of them lead to the following two 3-point 4-term relation:
\beq\label{(3,0)d}
\begin{array}{l}
(a^{-1}-b^{-1})e^{\nabla (a)\nabla (b)F} +
(b^{-1}-c^{-1})e^{\nabla (b)\nabla (c)F} 
+ (c^{-1}-a^{-1})e^{\nabla (c)\nabla (a)F} 
\\ \\
\phantom{aaa} + \,\,
(a^{-1}-b^{-1})(b^{-1}-c^{-1})(c^{-1}-a^{-1})
e^{\nabla (a)\nabla (b)F+\nabla (b)\nabla (c)F+
\nabla (c)\nabla (a)F}=0,
\end{array}
\eeq
where we have put $a_1=a$, $b_1=b$, $c_1=c$. It is the
dispersionless version of equation
(\ref{(3,0)}) (and of (\ref{(2,1)}) since they are the same).
The further limit
$c\to \infty$ gives the following 2-point equation:
\beq\label{L6}
(a^{-1}-b^{-1})e^{\nabla (a)\nabla (b)F}\Bigl (
1-(ab)^{-1}e^{\nabla (a)\p_{0}F +\nabla (b)\p_0 F}\Bigr )=
a^{-1}e^{\nabla (a)\p_{0}F}-b^{-1}e^{\nabla (b)\p_{0}F},
\eeq
where $\p_0 \equiv \p_{t_0}$. This equation, being expanded in negative
powers of $a,b$, generates an infinite number of nonlinear partial
differential equations for the function $F$. The simplest 
equation contained in (\ref{L6}) is
\beq\label{L7}
F_{02}-2F_{11}-F_{01}^2 +2e^{2F_{00}}=0,
\eeq
where we denote $F_{mn}\equiv \p_{t_m}\p_{t_n}F$. 

\begin{remark}
Comparing this with the corresponding equation (\ref{L7a})
of the dmKP hierarchy, we see that (\ref{L7}) has an extra term
(the last term in (\ref{L7})). The first three terms are seemingly
the same but 
one should have in mind that they can not be identified because
$F_{\tilde 0n}$ in (\ref{L7a}) and $F_{0n}$ in (\ref{L7}) are 
different functions.
\end{remark}

In terms of the function
\beq\label{L8}
w(z)=z^{-1}e^{\nabla (z)\p_0F}
\eeq
equation (\ref{L6}) can be written in the form
\beq\label{L6a}
(a^{-1}-b^{-1})e^{\nabla (a)\nabla (b)F}=
\frac{w(a)-w(b)}{1-w(a)w(b)}.
\eeq

\begin{proposition}
Equation (\ref{L6a}) is equivalent to the general
Hirota-Miwa equations (\ref{dL1a}) and (\ref{dkp-gen-disp-H-M}). 
\end{proposition}

\noindent
{\it Proof.}
We should show that the substitution (\ref{L6a}) solves 
the whole set of equations (\ref{dL1a}) and
(\ref{dkp-gen-disp-H-M}). Substituting (\ref{L6a})
into (\ref{dL1a}), we have:
\beq\label{L10}
\begin{array}{l}
\displaystyle{
\sum_{s=1}^{P^+}
\prod_{{\scriptsize \begin{array}{l}i=1\\ i\neq s \end{array}}}^{P^+}
\frac{1-w(a_s)w(a_i)}{w(a_s)-w(a_i)}
\prod_{k=1}^{P^-} \frac{w(a_s)-w(b_k)}{1-w(a_s)w(b_k)}}
\\ \\
\phantom{aaaaaaaa}\displaystyle{
+\sum_{s=1}^{P^-}
\prod_{{\scriptsize \begin{array}{l}k=1\\ k\neq s \end{array}}}^{P^-}
\frac{1-w(b_s)w(b_k)}{w(b_s)-w(b_k)}
\prod_{i=1}^{P^+} 
\frac{w(b_s)-w(a_i)}{1-w(b_s)w(a_i)}}
=1.
\end{array}
\eeq
One can see that for $P^+-P^- \in 2\ZZ +1$ this equality is satisfied
identically for all       $a_i$ and $b_k$. Indeed, denoting for brevity
$w(a_i)=w_i$, $w(b_k)=v_k$, we rewrite it as
\beq\label{L10a}
\begin{array}{l}
\displaystyle{
\sum_{s=1}^{P^+}
\prod_{{\scriptsize \begin{array}{l}i=1\\ i\neq s \end{array}}}^{P^+}
\frac{1-w_s w_i}{w_s-w_i}
\prod_{k=1}^{P^-} \frac{w_s-v_k}{1-w_sv_k}
+\sum_{s=1}^{P^-}
\prod_{{\scriptsize \begin{array}{l}k=1\\ k\neq s \end{array}}}^{P^-}
\frac{1-v_sv_k}{v_s-v_k}
\prod_{i=1}^{P^+} 
\frac{v_s-w_i}{1-v_sw_i}}
=1.
\end{array}
\eeq
To prove this identity, consider the function
$$
g(w)=\frac{1}{1-w^2}\, \prod_{i=1}^{P^+} 
\frac{1-ww_i}{w-w_i} \, \prod_{k=1}^{P^-}
\frac{w-v_k}{1-wv_k}
$$
with simple poles at $w=w_i$, $w=v_k^{-1}$ and $w=\pm 1$
(we assume that $w_i, v_k \neq \pm 1$).
Denote the left-hand side of (\ref{L10a}) by $S$. 
The sum of residues of $g(w)$ is 
$$
0=S-\frac{1}{2}\Bigl (1-(-1)^{P^+ +P^-}\Bigr )
$$
(the last terms come from the poles at $w=\pm 1$).
Therefore, for odd $P^{+} +P^-$ we have $S=1$.

The same argument shows that the substitution 
(\ref{L6a}) solves not only equations (\ref{dL1}) (that intertwine the
even and odd sectors of the hierarchy) but also equations 
of the form (\ref{dkp-gen-disp-H-M})
inside each sector. Indeed, they have the same form as
(\ref{dL1}) but with $0$ in the right-hand side instead of $1$, i.e.,
$S=0$ which is just what actually 
holds since in this case $P^+ +P^-\in 2\ZZ$. 
\square

\subsubsection{Large dBKP as a trigonometric 
degeneration of dDKP}
\label{section:modular}

Letting  $b\to \infty$ in (\ref{L6a}), we obtain, in the order $b^{-1}$:
$
a-\nabla (a)\p_{t_1}F =e^{F_{00}} \Bigl (w^{-1}(a)-w(a)\Bigr ),
$
or
\beq\label{L9}
p(z)=R_0 \Bigl (w^{-1}(z)-w(z)\Bigr ), \qquad R_0\equiv e^{F_{00}} .
\eeq
Equation (\ref{L9}) defines a rational curve 
which is the dynamical curve for the dispersionless large BKP
hierarchy.

\begin{proposition}\label{proposition:largeBKP}
The curve defined by the equation (\ref{L9}) is uniformized by
hyperbolic functions as follows:
\beq\label{td1}
w(z)=\tanh u(z), \qquad p(z)=\frac{\gamma}{\sinh u(z) \, \cosh u(z)}, 
\quad R_0=\gamma ,
\eeq
where $u(z)$ is a function of $z$ having 
the expansion near $\infty$ of the form (\ref{u(z)}) and $\gamma$
is a $z$-independent constant.
In this parametrization, equation (\ref{L6a}) acquires the form
\beq\label{L6b}
(a^{-1}-b^{-1}) e^{\nabla (a)\nabla (b)F}=
\tanh \Bigl (u(a)-u(b)\Bigr ).
\eeq
\end{proposition}

\noindent
This proposition provides the trigonometric parametrization
of the dispersionless large BKP hierarchy.
The proof consists in a direct verification that the substitution
(\ref{td1}) turns equation (\ref{L9}) into identity.

Equations (\ref{td1}), (\ref{L6b}) 
can be regarded as the $\tau \to +i0$ degeneration of
the ones from Section \ref{section:dDKP} (in particular, (\ref{L6b})
is the degeneration of (\ref{g(a,b)1})). To see this, we need
formulas (\ref{mod2d}) from Appendix B 
that connect the theta-functions with
modular parameters $\tau$ and $-1/\tau$. In particular, we have:
$$
\frac{\theta_2(0|\tau )}{\theta_3(0|\tau )}=
\frac{\theta_4(0|-1/\tau )}{\theta_3(0|-1/\tau )} \to 1 \;\;
\mbox{as $\tau \to +i0$},
$$
so the right-hand side of (\ref{RV1}) is equal to 2. This is
indeed the case since $F_{01}^2 +2F_{11}-F_{02}=2e^{2F_{00}}$
by virtue of (\ref{L7}).
For simplicity, assume that $\tau =it$, where $t\in \RR_+$, then
$-1/\tau =i/t \in i \RR_+$. Also, we put $u=tv$.
Then formulas (\ref{mod2d}) imply:
$$
\lim_{\tau \to +i0}\frac{\theta_1(u|\tau )}{\theta_4(u|\tau )}=
i\lim_{t\to +0}\frac{\theta_1(-iv|i/t)}{\theta_2(-iv|i/t)}
=i\frac{\sin (-iv)}{\cos (-iv)}=\tanh v,
$$
so the right-hand side of (\ref{L6b}) is indeed the 
degeneration of the one of (\ref{g(a,b)1}).

\begin{remark}
It is natural to ask what happens with the elliptic curve
(\ref{ddkp-curve-ep}) 
$
R^2(w^2 +w^{-2}) =p^2 +V
$
in this limit. According to (\ref{R111}) and (\ref{L7}),
$$
V=F_{01}^2 +2F_{11}-F_{02}=2e^{2F_{00}}=2R^2,
$$
so the equation of the curve acquires the form
$
R^2 (w-w^{-1})^2 =p^2.
$
This means that the elliptic curve given by equation of degree 4
degenerates and
splits into two rational components defined by equations of degree 2:
\beq\label{split}
R(w -w^{-1})=\pm p,
\eeq
which can be uniformized using hyperbolic functions.
\end{remark}

\section{Multi-component hierarchies of Pfaff type}

\subsection{Multi-component DKP and its dispersionless version}

\subsubsection{Multi-component DKP}
\label{section:multiDKP}

In the $N$-component DKP hierarchy the independent variables are: 
\beq\label{times2a}
\begin{array}{l}
{\bf t}=\{{\bf t}_1, {\bf t}_2, \ldots , {\bf t}_N\}, \qquad
{\bf t}_{\alpha}=\{t_{\alpha , 1}, t_{\alpha , 2}, t_{\alpha , 3}, 
\ldots \, \},
\\ \\
{\bf n}=\{n_1, \ldots , n_N\}, \quad
\quad n_{\alpha} \in \ZZ , \quad \alpha = 1, \ldots , N.
\end{array}
\eeq
The tau-function $\tau({\bf n}, {\bf t})$ 
is defined as the expectation value:
\beq\label{taumdkp}
\tau ({\bf n}, {\bf t}) =
\lbr
{\bf n} \bigl | \, e^{J({\bf t})} g \bigr | \, 0 \rbr,
\eeq
where the Clifford group element $g$ has the following general
form:
\beq\label{g-mdkp}
g=\exp 
\biggl (\sum_{i,j\in \z} \sum_{\alpha , \beta} 
\Bigl(
A_{ij}^{(\alpha \beta)}\psi_{i}^{(\alpha )}\psi_j^{*(\beta )}
+
B_{ij}^{(\alpha \beta)}\psi_{i}^{*(\alpha )}\psi_j^{*(\beta )}
+
C_{ij}^{(\alpha \beta)}\psi_{i}^{(\alpha )}\psi_j^{(\beta )}
\Bigr)
\biggr ).
\eeq
Tau-functions defined in this way are nonzero only if 
the parity condition 
$|{\bf n}|\in 2\ZZ $ holds.
As it was shown in \cite{SZ25a},
the tau-function satisfies the following bilinear equation:

\beq\label{mDKP-bil-rel}
\begin{array}{l}
\displaystyle{
\sum_{\gamma =1}^N\epsilon_{\gamma}({\bf n}\! -\! {\bf n}')
\oint_{C_{\infty}}\! \frac{dz}{z^2} 
z^{n_{\gamma}-n_{\gamma}'}
e^{\xi ({\bf t}_{\gamma}-{\bf t}_{\gamma}', z)}}
\\ \\
\phantom{aaaaaaaaaaaaa}
\displaystyle{
\times \tau ({\bf n}\! -\! {\bf e}_{\gamma}, 
{\bf t}\! -\! [z^{-1}]_{\gamma})\tau ({\bf n}'\! +\! {\bf e}_{\gamma}, 
{\bf t}'\! +\! [z^{-1}]_{\gamma})}
\\ \\
+\, 
\displaystyle{
\sum_{\gamma =1}^N \epsilon_{\gamma}({\bf n}\! -\! {\bf n}')
\oint_{C_{\infty}}\! \frac{dz}{z^2} 
z^{n_{\gamma}'-n_{\gamma}}
e^{-\xi ({\bf t}_{\gamma}-{\bf t}_{\gamma}', z)}}
\\ \\
\phantom{aaaaaaaaaaaaa}
\displaystyle{
\times \tau ({\bf n}\! +\! {\bf e}_{\gamma}, 
{\bf t}\! +\! [z^{-1}]_{\gamma})\tau ({\bf n}'\! -\! {\bf e}_{\gamma}, 
{\bf t}'\! -\! [z^{-1}]_{\gamma})} = 0,
\end{array}
\eeq
with the same definition (\ref{epsilon}) 
of $\epsilon_{\gamma}$. This equation holds for an arbitrary ${\bf n}$ and ${\bf n}'$ such that
$
|{\bf n}|, |{\bf n}'| \in 2\ZZ + 1 .
$
Note the obvious symmetry of equation (\ref{mDKP-bil-rel}):
\beq\label{sym3}
({\bf n}, {\bf t}) \leftrightarrow ({\bf n'}, {\bf t'}).
\eeq
The simplest solution to (\ref{mDKP-bil-rel}) is given by the
following proposition which generalizes Proposition 
\ref{proposition:simplest} (see (\ref{simplest})).

\begin{proposition}\cite{SZ25b}
\label{proposition:simplest11}
The function 
\beq\label{simplest1}
\tau ({\bf n}, {\bf t})=
\exp \left ( \frac{1}{2}\sum_{\gamma =1}^N \sum_{k\geq 1} kt_{\gamma , k}^2
\right )
\eeq
is a solution to equation (\ref{mDKP-bil-rel}).
\end{proposition}

\noindent
{\it Proof.}
Plugging (\ref{simplest1}) into (\ref{mDKP-bil-rel}), we see, extracting 
a common multiplier, that the left-hand side
is proportional to
$$
\begin{array}{lll}
L&=&\displaystyle{\sum_{\gamma =1}^N\epsilon_{\gamma}({\bf n}-{\bf n'})
\oint_{C_{\infty}}\! \frac{dz}{z^2-1} \,
z^{n_{\gamma}-n_{\gamma}'}} e^{\xi ({\bf t}_{\gamma}-{\bf t'}_{\gamma}, z)-
\xi ({\bf t}_{\gamma} -{\bf t'}_{\gamma}, z^{-1})}
\\ && \\
&& \displaystyle{
+\, \sum_{\gamma =1}^N\epsilon_{\gamma}({\bf n}-{\bf n'})
\oint_{C_{\infty}}\! \frac{dz}{z^2-1} \, 
z^{n_{\gamma}'-n_{\gamma}}}e^{-\xi ({\bf t}_{\gamma}-{\bf t'}_{\gamma}, z)+
\xi ({\bf t}_{\gamma} -{\bf t'}_{\gamma}, z^{-1})}.
\end{array}
$$
Changing the integration variable $z\to z^{-1}$ in the second line,
we have:
$$
\begin{array}{lll}
L&=&\displaystyle{\sum_{\gamma =1}^N\epsilon_{\gamma}({\bf n})
\epsilon_{\gamma}({\bf n'})
\Bigl (\oint_{C_{\infty}}-\oint_{C_{0}}\Bigr ) \frac{dz}{z^2-1} \, 
z^{n_{\gamma}-n_{\gamma}'}} e^{\xi ({\bf t}_{\gamma}-{\bf t'}_{\gamma}, z)-
\xi ({\bf t}_{\gamma}-{\bf t'}_{\gamma}, z^{-1})},
\end{array}
$$
where $C_0$ is a small circle around $0$. Therefore, the integral is 
given by residues at the simple poles at the points $z=\pm 1$ lying inside
the annulus bordered by the circles $C_0$ and $C_{\infty}$:
$$
\begin{array}{lll}
L&=&\displaystyle{
\pi i\sum_{\gamma =1}^N\epsilon_{\gamma}({\bf n})
\epsilon_{\gamma}({\bf n'})\Bigl (1-(-1)^{n_{\gamma}-n'_{\gamma}}\Bigr )}
\\ && \\
&&=\, \displaystyle{\pi i\sum_{\gamma =1}^N \Bigl (\epsilon_{\gamma}({\bf n})
\epsilon_{\gamma}({\bf n'})-\epsilon_{\gamma -1}({\bf n})
\epsilon_{\gamma -1}({\bf n'})\Bigr )}
\\ && \\
&& =\, \pi i \, \Bigl (1-(-1)^{|{\bf n}|-|{\bf n'}|}\Bigr )
\end{array}
$$
which is zero if
$|{\bf n}|$ and $|{\bf n'}|$ are of the same parity.
\square

The Miwa substitution 
\beq\label{m1a}
\left \{ \begin{array}{l}
\displaystyle{
{\bf n}-{\bf n}' =-{\bf k}=\sum_{i=1}^{P^+}{\bf e}_{\alpha_i}-
\sum_{k=1}^{P^-}{\bf e}_{\beta_k},}
\\  \\
\displaystyle{
{\bf t}-{\bf t}' =-{\bf T}=\sum_{i=1}^{P^+}[a_i^{-1}]_{\alpha_i}-
\sum_{k=1}^{P^-}[b_k^{-1}]_{\beta_k}}
\end{array} \right.
\eeq
allows one to calculate 
the integral by residue calculus.
The result is the following general Hirota-Miwa equation:
\beq\label{gen-H-M(multi-DKP)}
\begin{array}{l}
\displaystyle{
\sum_{s=1}^{P^+}
\prod_{{\scriptsize \begin{array}{l}i=1\\ i\neq s \end{array}}}^{P^+}
E^{-1}_{\alpha_i \alpha_s}(a_s, a_i)
\prod_{k=1}^{P^-}
E_{\beta_k \alpha_s}(a_s, b_k)}
\\ \\
\phantom{aaaaaaa}\displaystyle{
\times \, 
\tau \Bigl ({\bf n}-{\bf e}_{\alpha_s},  {\bf t}-[a_s^{-1}]_{\alpha_s},\Bigr )
\tau \Bigl ({\bf n}+{\bf k}+{\bf e}_{\alpha_s}, 
{\bf t}+{\bf T}+[a_s^{-1}]_{\alpha_s}
\Bigr )}
\\ \\
\displaystyle{
+\sum_{s=1}^{P^-}
\prod_{{\scriptsize \begin{array}{l}k=1\\ k\neq s \end{array}}}^{P^-}
E^{-1}_{\beta_k \beta_s}(b_s, b_k)
\prod_{i=1}^{P^+}
E_{\alpha_i \beta_s}(b_s, a_i)}
\\ \\
\phantom{aaaaaaa}\displaystyle{
\times \, 
\tau \Bigl ({\bf n}+{\bf e}_{\beta_s}, 
{\bf t}+[b_s^{-1}]_{\beta_s}, \Bigr )
\tau \Bigl ({\bf n}+{\bf k}-{\bf e}_{\beta_s},  
{\bf t}+{\bf T}-[b_s^{-1}]_{\beta_s}
\Bigr )=0,}
\end{array}
\eeq
with the parity condition:
\beq\label{m2}
P^{+} - P^{-} \in 2\ZZ .
\eeq

The case of our main interest is $P^{+} + P^{-} = 4$. There are five possibilities but taking into account the symmetry 
(\ref{sym3}) it is enough to consider the following three:
\beq\label{nd1a}
(P^+, P^-)=(4,0), \; (3,1), \; (2,2).
\eeq
To write the formulas below in compact form the following short-hand notation is useful:
\beq\label{nd6}
{\bf n}^{\alpha}={\bf n} +{\bf e}_{\alpha}, 
\quad
{\bf n}^{\alpha \beta}={\bf n} +{\bf e}_{\alpha}+{\bf e}_{\beta}
\eeq
and so on. Similar notation will be used for shifts of continuous times:
\beq\label{nd6a}
{\bf t}^{[a_{\alpha}]}={\bf t} + [a^{-1}]_{\alpha}, 
\quad
{\bf t}^{[a_{\alpha}b_{\beta}]}={\bf t} + [a^{-1}]_{\alpha}+ [b^{-1}]_{\beta},
\eeq
For the first possibility, 
$(P^ +, P^-)=(4,0)$, we have:
\beq\label{40}
\left \{ \begin{array}{l}
{\bf n}-{\bf n}'={\bf e}_{\alpha}+{\bf e}_{\beta}+
{\bf e}_{\nu}+{\bf e}_{\mu}, 
\\  \\
{\bf t}-{\bf t}'=[a^{-1}]_{\alpha} +[b^{-1}]_{\beta}+ [c^{-1}]_{\nu}
+ [d^{-1}]_{\mu}.
\end{array} \right.
\eeq
In this particular case the general formula (\ref{gen-H-M(multi-DKP)}) gives:
\beq\label{h-m-40(multi-DKP)}
\begin{array}{l}
\phantom{-}E_{\nu \beta}(b,c)E_{\mu \beta}(b,d)E_{\nu \mu}(d,c)
\tau \Bigl ({\bf n}^{\beta \nu \mu}, {\bf t}^{[b_{\beta}c_{\nu}d_{\mu}]}\Bigr )
\tau \Bigl ({\bf n}^{\alpha}, {\bf t}^{[a_{\alpha}]}\Bigr )
\\ \\
-\, E_{\nu \alpha}(a,c)E_{\mu \alpha}(a,d)E_{\nu \mu}(d,c)
\tau \Bigl ({\bf n}^{\alpha \nu \mu}, {\bf t}^{[a_{\alpha}c_{\nu}d_{\mu}]}\Bigr )
\tau \Bigl ({\bf n}^{\beta}, {\bf t}^{[b_{\beta}]}\Bigr )
\\ \\
-\, E_{\beta \alpha}(a,b)E_{\mu \alpha}(a,d)E_{\mu \beta}(b,d)
\tau \Bigl ({\bf n}^{\alpha \beta \mu}, {\bf t}^{[a_{\alpha}b_{\beta}
d_{\mu}]}\Bigr )
\tau \Bigl ({\bf n}^{\nu}, {\bf t}^{[c_{\nu}]}\Bigr )
\\ \\
+\, E_{\beta \alpha}(a,b)E_{\nu \alpha}(a,c)E_{\nu \beta}(b,c)
\tau \Bigl ({\bf n}^{\alpha \beta \nu}, {\bf t}^{[a_{\alpha}b_{\beta}
c_{\nu}]}\Bigr )
\tau \Bigl ({\bf n}^{\mu}, {\bf t}^{[d_{\mu}]}\Bigr )\, =0.
\end{array}
\eeq
For the second possibility, 
$(P^ +, P^-)=(3,1)$, we have:
\beq\label{31}
\left \{ \begin{array}{l}
{\bf n}-{\bf n}'={\bf e}_{\alpha}+{\bf e}_{\beta}+{\bf e}_{\nu}-{\bf e}_{\mu}, 
\\  \\
{\bf t}-{\bf t}'=[a^{-1}]_{\alpha} +[b^{-1}]_{\beta}+ [c^{-1}]_{\nu}-
[d^{-1}]_{\mu}.
\end{array}\right.
\eeq
The corresponding Hirota-Miwa equation is
\beq\label{h-m-31(multi-DKP)}
\begin{array}{l}
\phantom{-}E_{\nu \beta}(b,c)E_{\mu \alpha}(a,d)
\tau \Bigl ({\bf n}^{\beta \nu}, {\bf t}^{[b_{\beta}c_{\nu}]}\Bigr )
\tau \Bigl ({\bf n}^{\alpha \mu}, {\bf t}^{[a_{\alpha}d_{\mu}]}\Bigr )
\\ \\
-\, E_{\mu \beta}(b,d)E_{\nu \alpha}(a,c)
\tau \Bigl ({\bf n}^{\alpha \nu}, {\bf t}^{[a_{\alpha}c_{\nu}]}\Bigr )
\tau \Bigl ({\bf n}^{\beta \mu}, {\bf t}^{[b_{\beta}d_{\mu}]}\Bigr )
\\ \\
+\, E_{\beta \alpha}(a,b)E_{\mu \nu}(c,d)
\tau \Bigl ({\bf n}^{\alpha \beta}, {\bf t}^{[a_{\alpha}b_{\beta}]}\Bigr )
\tau \Bigl ({\bf n}^{\nu \mu}, {\bf t}^{[c_{\nu}d_{\mu}]}\Bigr )
\\ \\
+\, E_{\beta \alpha}(a,b)E_{\nu \alpha}(a,c)E_{\nu \beta}(b,c)
E_{\alpha \mu}(d,a)E_{\beta \mu}(d,b)E_{\nu\mu}(d,c)
\\ \\
\phantom{aaaaaaaaaaaaaaaaaaaaaa}
\times \tau \Bigl ({\bf n}^{\alpha \beta \nu \mu}, 
{\bf t}^{[a_{\alpha}b_{\beta}c_{\nu}d_{\mu}]}\Bigr )
\tau \Bigl ({\bf n}, {\bf t}\Bigr )\, =0.
\end{array}
\eeq
The choice $(P^{+}, P^{-}) = (2,2)$ leads to an 
equation which is equivalent to (\ref{h-m-40(multi-DKP)}).

\subsubsection{Multi-component dDKP:
uni\-for\-mi\-za\-ti\-on via elliptic functions}

Using the notation introduced for the multi-component dmKP case,
we can write the 
dispersionless version of equation (\ref{gen-H-M(multi-DKP)}) 
in the form 
\beq\label{gen-H-M}
\begin{array}{l}
\displaystyle{
\sum_{s = 1}^{P^{+}}
\left(
\prod_{i = 1, \neq s}^{P^{+}} 
E_{\alpha_i \alpha_s}(a_{s}, a_{i})
e^{ \nabla_{\alpha_{i}} (a_{i}) \nabla_{\alpha_{s}} (a_{s}) F}
\right )^{-1}
\left(
\prod_{k=1}^{P^{-}} E_{\beta_k \alpha_s}(a_{s}, b_{k}) 
e^{ \nabla_{\alpha_{s}} (a_{s})  \nabla_{\beta_{k}}(b_{k})F}
\right)}
\\ \\ +\,
\displaystyle{
\sum_{s = 1}^{P^{-}}
\left(
\prod_{i = 1, \neq s}^{P^{-}} E_{\beta_k \beta_s}(b_{s}, b_{i})
e^{ \nabla_{\beta_{i}} (b_{i}) \nabla_{\beta_{s}} (b_{s}) F}
\right )^{-1}
\left(
\prod_{k = 1}^{P^{+}} E_{\alpha_i \beta_s}(b_{s}, a_{k}) 
e^{ \nabla_{\beta_{s}} (b_{s}) \nabla_{\alpha_{k}}(a_{k})F}
\right) = 0.}
\end{array}
\eeq
In particular, the limiting forms of equations (\ref{h-m-40(multi-DKP)}), (\ref{h-m-31(multi-DKP)}) are:

\beq\label{40c1}
\begin{array}{l}
\phantom{-}E_{\nu \beta}(b,c)E_{\mu \beta}(b,d)E_{\nu \mu}(d,c)
e^{\bigl ( \nabla_{\beta}(b) \nabla_{\nu}(c)+ \nabla_{\beta}(b)
\nabla_{\mu}(d)+ \nabla_{\nu}(c)\nabla_{\mu}(d)\bigr )F}
\\ \\
-\, E_{\nu \alpha}(a,c)E_{\mu \alpha}(a,d)E_{\nu \mu}(d,c)
e^{\bigl ( \nabla_{\alpha}(a) \nabla_{\nu}(c)+ \nabla_{\alpha}(a)
\nabla_{\mu}(d)+ \nabla_{\nu}(c)\nabla_{\mu}(d)\bigr )F}
\\ \\
-\, E_{\beta \alpha}(a,b)E_{\mu \alpha}(a,d)E_{\mu \beta}(b,d)
e^{\bigl ( \nabla_{\alpha}(a) \nabla_{\beta}(b)+ \nabla_{\alpha}(a)
\nabla_{\mu}(d)+ \nabla_{\beta}(b)\nabla_{\mu}(d)\bigr )F}
\\ \\
+\, E_{\beta \alpha}(a,b)E_{\nu \alpha}(a,c)E_{\nu \beta}(b,c)
e^{\bigl ( \nabla_{\alpha}(a) \nabla_{\beta}(b )+ \nabla_{\alpha}(a)
\nabla_{\nu}(c)+ \nabla_{\beta}(b)\nabla_{\nu}(c)\bigr )F}
\, =0,
\end{array}
\eeq

\beq\label{31c1}
\begin{array}{l}
\phantom{-}E_{\nu \beta}(b,c)E_{\mu \alpha}(a,d)
e^{\bigl ( \nabla_{\beta}(b) \nabla_{\nu}(c)+ \nabla_{\alpha}(a)
\nabla_{\mu}(d)\bigr )F}
\\ \\
-\, E_{\mu \beta}(b,d)E_{\nu \alpha}(a,c)
e^{\bigl ( \nabla_{\alpha}(a) \nabla_{\nu}(c)+ \nabla_{\beta}(b)
\nabla_{\mu}(d\bigr )F}
\\ \\
+\, E_{\beta \alpha}(a,b)E_{\mu \nu}(c,d)
e^{\bigl ( \nabla_{\alpha}(a) \nabla_{\beta}(b)+ \nabla_{\nu}(c )
\nabla_{\mu}(d)\bigr )F}
\\ \\
+\, E_{\beta \alpha}(a,b)E_{\nu \alpha}(a,c)E_{\nu \beta}(b,c)
E_{\alpha \mu}(d,a)E_{\beta \mu}(d,b)E_{\nu\mu}(d,c)
\\ \\
\times 
e^{\bigl ( \nabla_{\alpha}(a) \nabla_{\beta}(b)+ \nabla_{\alpha}(a )
\nabla_{\nu}(c)+ \nabla_{\alpha}(a) \nabla_{\mu}(d)+
\nabla_{\beta}(b) \nabla_{\nu}(c)+ \nabla_{\beta}(b) \nabla_{\mu}(d)+
\nabla_{\nu}(c) \nabla_{\mu}(d)\bigr )F}
\, =0,
\end{array}
\eeq
where the indices $\{\alpha, \beta, \nu, \mu\}$ correspond 
to $\{a, b, c, d\}$. 

\noindent
Similarly to (\ref{g-def-dkp}), it is convenient to introduce
the $g$-function, but now it is supplied with indices:
\beq\label{g-def}
g_{\alpha \beta}(a,b) = 
\epsilon_{\beta \alpha}
(a^{-1} - b^{-1})^{\delta_{\alpha \beta}} e^{\nabla_{\alpha}(a) 
\nabla_{\beta}(b) F},
\eeq
with the same sign factor $\epsilon_{\beta \alpha}$ as in 
(\ref{E1}).
Equation (\ref{gen-H-M}), (\ref{40c1}), (\ref{31c1}) then read:

\beq\label{gen-H-M1}
\begin{array}{l}
\displaystyle{
\sum_{s = 1}^{P^{+}}
\left(
\prod_{i = 1, \neq s}^{P^{+}} 
g_{\alpha_s \alpha_i}(a_{s}, a_{i})
\right )^{-1}
\left(
\prod_{k=1}^{P^{-}} g_{\alpha_k \alpha_s}(a_{s}, b_{k}) 
\right)}
\\ \\
\displaystyle{
\phantom{aaaaaaaaaaaaaaa} +\,
\sum_{s = 1}^{P^{-}}
\left(
\prod_{i = 1, \neq s}^{P^{-}} g_{\beta_s \beta_k}(b_{s}, b_{i})
\right )^{-1}
\left(
\prod_{k = 1}^{P^{+}} g_{\beta_s \alpha_i}(b_{s}, a_{k}) 
\right) = 0,}
\end{array}
\eeq

\beq \label{g1}
\begin{array}{l}
\displaystyle{
g_{\beta \nu}(b, c)
g_{\beta \mu}(b, d) 
g_{\mu \nu}(d, c) 
-
g_{\alpha \nu}(a, c)
g_{\alpha \mu}(a, d)
g_{\mu \nu}(d, c)
}
\\ \\ 
\phantom{aaaaaaaaaaaaa}\displaystyle{-\,
g_{\alpha \beta}(a, b)
g_{\alpha \mu}(a, d)
g_{\beta \mu}(b, d)
+
g_{\alpha \beta}(a, b) 
g_{\alpha \nu}(a, c)
g_{\beta \nu}(b, c) = 0,
}
\end{array}
\eeq

\beq \label{g2}
\begin{array}{l}
\displaystyle{
g_{\beta \nu}(b, c) g_{\alpha \mu}(a, d)
-
g_{\beta \mu}(b, d) g_{\alpha \nu}(a, c)
+
g_{\alpha \beta}(a, b) g_{\nu \mu}(c, d)
}
\\ \\ 
\phantom{aaaaaaaaaaaaa}\displaystyle{ +\, 
g_{\alpha \beta}(a, b) g_{\alpha \nu}(a, c)
g_{\beta \nu}(b, c) g_{\mu \alpha}(d, a)
g_{\mu \beta}(d, b) g_{\mu \nu}(d, c) = 0.
}
\end{array}
\eeq

\begin{theorem}\cite{SZ24,SZ25b}
The single equation
\beq\label{E11}
g_{\alpha \beta}(a,b)=
E_{\beta \alpha}(a,b) e^{\nabla_{\alpha}(a)\nabla_{\beta}(b)F}
= \frac{\theta_1(u_{\alpha}(a)- u_{\beta}(b))}{\theta_4(u_{\alpha}(a)-
u_{\beta}(b))},
\eeq
where $\theta_1 (u), \, \theta_4(u)$ are Jacobi's theta-functions 
with an elliptic modular parameter $\tau$ (not shown explicitly 
in (\ref{E11})) and
$u_{\alpha}(z)=u_{\alpha}(z; {\bf t})$
are generating functions of dynamical variables 
with the following expansions as 
$z\to \infty$:
\beq\label{u}
u_{\alpha}(z)=\eta_{\alpha}({\bf t}) + 
\sum_{k\geq 1}c^{(\alpha )}_k({\bf t}) z^{-k},
\eeq
is equivalent to the general equation (\ref{gen-H-M}),
i.e., to the
whole multi-component dDKP hierarchy. 
\end{theorem}

\noindent
The idea of the proof is that 
the substitution (\ref{E11}) converts into
identities {\it all} 
dispersionless Hirota-Miwa equations of the general form (\ref{gen-H-M1})
(or (\ref{gen-H-M})).
Indeed, denote for brevity
$$
u_i=u_{\alpha_i}(a_i), \qquad v_k=u_{\beta_k}(b_k), 
\quad i,k=1, \ldots , N,
$$
then (\ref{gen-H-M1}) converts into
the same already proven identity (\ref{id1(ddkp)}), as in the 
one-component case.
For the detailed proof see \cite{SZ25b}.

Similarly to the one-component case,
the meaning of equation (\ref{E11}) 
is that general second order derivatives
of the function $F$ with respect to the independent variables 
are expressed through
some special second order derivatives.
Let us explain this
in more details. 
Denote the function 
$\theta_1(u)/\theta_4(u)$ by ${\sf sn}(u)$, 
and the inverse function
by ${\sf arcsn}(u)$, as in Section \ref{section:dDKP}
(see (\ref{u1}) and the footnote there).
Putting $\beta =\alpha$ and $b=\infty$ in (\ref{E11}), we 
conclude that
\beq\label{u11}
u_{\alpha}(a)-\eta_{\alpha}=
{\sf arcsn} \Bigl (a^{-1}e^{\nabla_{\alpha}(a)
\p_{\alpha}F}\Bigr ),
\eeq
where $\p_{\alpha}\equiv \p_{t_{\alpha , 0}}$.
Thus, equation (\ref{E11}) for $\alpha
\neq \beta$ can be written as
\beq\label{t1611}
\epsilon_{\beta \alpha}e^{\nabla_{\alpha}(a)\nabla_{\beta}(b)F}
={\sf sn} \Bigl ( \eta_{\alpha}-\eta_{\beta}+{\sf arcsn}\, (a^{-1}
e^{\nabla_{\alpha}(a)\p_{\alpha}F})-
{\sf arcsn}\, (b^{-1}
e^{\nabla_{\beta}(b)\p_{\beta}F})\Bigr ).
\eeq
At the same time, since
$$
\eta_{\alpha}-\eta_{\beta} = \left \{
\begin{array}{l}
\displaystyle{\phantom{-}\sum_{\gamma =\alpha}^{\beta -1} 
(\eta_{\gamma} -\eta_{\gamma +1}),} \quad \alpha <\beta ,
\\  \\
\displaystyle{
-\sum_{\gamma =\beta}^{\alpha -1} (\eta_{\gamma}-\eta_{\gamma +1}),}  
\quad \alpha >\beta , 
\end{array} \right.
$$
we can write:
\beq\label{u21}
\eta_{\alpha}-\eta_{\beta}= \left \{
\begin{array}{l}
\displaystyle{
-\sum_{\gamma =\alpha}^{\beta -1} 
{\sf arcsn} (e^{\p_{\gamma}\p_{\gamma +1}F}),}
\quad \alpha <\beta ,
\\  \\
\displaystyle{ \phantom{-}
\sum_{\gamma =\beta}^{\alpha -1} 
{\sf arcsn} (e^{\p_{\gamma}\p_{\gamma +1}F}),}
\quad \alpha >\beta .
\end{array} \right.
\eeq
Therefore, as equation (\ref{t1611}) shows, the general 
second order derivatives of the function $F$ are expressed through the
particular derivatives $\p_{\alpha}^2 F$, $\p_{\alpha}\p_{\alpha +1}F$,
$\p_{t_{\alpha ,k}}\p_{\alpha}F$ for $\alpha =1, \ldots , M$. 
Note that only differences of the $\eta_{\alpha}$'s enter equations
of the hierarchy\footnote{The reason is that the uniformization 
variable $u$ living in the 
fundamental parallelogram is defined only up to an additive 
constant.}.

Putting $b=\infty$ in (\ref{E11}) and denoting $a=z$, we get
\beq\label{E1a1}
\epsilon_{\beta \alpha}z^{-\delta_{\alpha \beta} }
e^{\nabla_{\alpha}(z)\p_{\beta}F}=
\frac{\theta_1(u_{\alpha}(z)-\eta_{\beta})}{\theta_4(u_{\alpha}(z)-\eta_{\beta})}.
\eeq
Let us denote
\beq\label{E51}
R_{\alpha}=e^{\p_{\alpha}^2F}, \qquad
R_{\alpha \beta}=e^{\p_{\alpha} \p_{\beta}F}  =R_{\beta \alpha}\;\; (\mbox{for}\; 
\alpha \neq \beta ).
\eeq
The further limit $z\to \infty$ in (\ref{E1a1}) then yields, 
for $\beta \neq \alpha$:
\beq\label{E1b}
R_{\alpha \beta}=\epsilon_{\beta \alpha}
\frac{\theta_1(\eta_{\alpha \beta})}{\theta_4(\eta_{\alpha \beta})},
\qquad \eta_{\alpha \beta}\equiv \eta_{\alpha}-\eta_{\beta},
\eeq
which is explicitly symmetric in $\alpha$ and $\beta$, as it should be
according to the definition (\ref{E51}).
In the limit $z\to \infty$ in (\ref{E1a1}) for $\beta =\alpha$ both sides
tend to zero. Comparing the leading terms as $z^{-1}\to 0$, we get
the relation
\beq\label{E1c}
R_{\alpha}=\pi c_1^{(\alpha )}\theta_2(0|\tau ) \theta_3(0|\tau ),
\eeq
where $c_1^{(\alpha )}$ is the coefficient at $z^{-1}$ in the expansion
(\ref{u}) (to obtain the right-hand side in this form, one should
use identity (\ref{theta1prime}) from Appendix B).
This relation generalizes (\ref{RV}) to the multi-component case.

It is convenient to introduce the function
\beq\label{S}
S(u)=\log \frac{\theta_1(u)}{\theta_4(u)}.
\eeq
It has the (quasi)periodicity properties $S(u+1)=S(u)+i\pi$,
$S(u+\tau )=S(u)$. Its $u$-derivative, $S'(u)$, 
is already a double-periodic
function
with periods $1$ and $\tau$. Explicitly
it is given by the formula 
\beq\label{S'}
S'(u)=\pi \theta_4^2(0)\frac{\theta_2(u)\, \theta_3 (u)}{\theta_1(u)\, 
\theta_4(u)},
\eeq
which can be proved comparing
the analytic properties of the two sides and using (\ref{theta1prime}).

\begin{remark}
The identity (\ref{id}) from Appendix B can be viewed as a nonlinear
differential equation for the function $S(u)$:
\beq\label{S1}
\left (\frac{S'(u)}{\pi \theta_2(0)\, \theta_3 (0)}\right )^2
=2\cosh (2S(u)) -\frac{\theta_2^2(0)}{\theta_3^2(0)}-
\frac{\theta_3^2(0)}{\theta_2^2(0)}.
\eeq
\end{remark}

\noindent
Let us compare the next-to-leading terms, as $z^{-1}\to 0$, 
in the both
sides of (\ref{E1a1}) for $\alpha \neq \beta$. This gives:
$$
R_{\alpha \beta}\Bigl (1+\frac{c_1^{(\alpha )}}{z} S'(\eta_{\alpha \beta})
\Bigr ) =R_{\alpha \beta} \Bigl (1+\frac{1}{z}\, 
\p_{\beta}\p_{t_{\alpha , 1}}F
\Bigr ) +O(z^{-2}),
$$ 
i.e.,
$
c_1^{(\alpha )}S'(\eta_{\alpha \beta})= \p_{\beta}\p_{t_{\alpha , 1}}F.
$
Using relations (\ref{E1b}), (\ref{E1c}), (\ref{S'}), we arrive at
the equation
\beq\label{E1e}
e^{(\p_{\beta}-\p_{\alpha})\p_{\alpha}F}
\p_{\beta}\p_{t_{\alpha , 1}}F=
\epsilon_{\beta \alpha} \frac{\theta_4^2(0)\, \theta_2(\eta_{\alpha \beta})
\, \theta_3(\eta_{\alpha \beta})}{\theta_2(0)\, 
\theta_3(0)\,\theta_4^2(\eta_{\alpha \beta})}
\eeq
which we will need later.

At the first glance the elliptic ansatz (\ref{E11}) 
seems to work equally well
for any elliptic modular parameter $\tau$,
including the degenerate cases $\tau \to +i0$ or $\tau \to +i\infty$, 
in which
elliptic functions become trigonometric or hyperbolic.
However, a more thorough analysis performed below shows that the internal 
consistency of the elliptic uniformization of the whole hierarchy
imposes strict restrictions on possible values of $\tau$. Moreover,
the elliptic modular parameter is required to be
a dynamical variable depending on the times ${\bf t}$ in 
a prescribed way. Namely, it can be
expressed in terms of mixed second order derivatives of $F$,
like in the one-component case (see
equation (\ref{E11a}) below).

To determine the modular parameter,
we should make explicit the elliptic curve hidden in the 
hierarchy. This can be done by considering degenerate cases of the
4-point relations. Namely, we put $\alpha =\beta =\nu \neq \mu$ and
tend $c,d\to \infty$ in (\ref{40c1}), (\ref{31c1}). Renaming 
$\mu \leftrightarrow \beta$ after this, we see that equation (\ref{40c1})
yields
\beq\label{E21}
\begin{array}{l}
-\, b^{-1}e^{(\nabla_{\alpha}(b)\p_{\alpha}+\nabla_{\alpha}(b)
\p_{\beta}+\p_{\alpha}\p_{\beta})F}
+\, a^{-1}e^{(\nabla_{\alpha}(a)\p_{\alpha}+\nabla_{\alpha}(a) \p_{\beta}
+\p_{\alpha}\p_{\beta})F}
\\ \\
\phantom{aaaaaaaaa}
-\, (a^{-1}-b^{-1})e^{(\nabla_{\alpha}(a)\nabla_{\alpha}(b)
+\nabla_{\alpha}(a) \p_{\beta}+\nabla_{\alpha}(b) \p_{\beta})F}
\\ \\
\phantom{aaaaaaaaaaaaaaaa}
+\, (a^{-1}-b^{-1})(ab)^{-1}e^{(\nabla_{\alpha}(a)\nabla_{\alpha}(b)
+\nabla_{\alpha}(a) \p_{\alpha}+\nabla_{\alpha}(b) \p_{\alpha})F}=0.
\end{array}
\eeq
In a similar way, equation (\ref{31c1}) yields:
\beq\label{E31}
\begin{array}{l}
b^{-1}e^{(\nabla_{\alpha}(b)\p_{\alpha}+\nabla_{\alpha}(a)
\p_{\beta})F}
- a^{-1}e^{(\nabla_{\alpha}(a)\p_{\alpha}+\nabla_{\alpha}(b) \p_{\beta})F}
+(a^{-1}-b^{-1})e^{(\nabla_{\alpha}(a)\nabla_{\alpha}(b)+\p_{\alpha}
\p_{\beta})F}
\\ \\
\phantom{aaaa}
-\, (a^{-1}-b^{-1})(ab)^{-1}e^{(\nabla_{\alpha}(a)\nabla_{\alpha}(b)
+\nabla_{\alpha}(a)\p_{\alpha}+\nabla_{\alpha}(a) \p_{\beta}
+\nabla_{\alpha}(b)\p_{\alpha} +\nabla_{\alpha}(b)\p_{\beta}
+\p_{\alpha}\p_{\beta})F}=0.
\end{array}
\eeq
These two equations have to be satisfied 
simultaneously, and this requirement 
allows one to recover the dynamical curve and make it explicit, i.e., 
to represent its points as solutions of a polynomial equation in two complex
variables.

As before, we
introduce the following functions:
\beq\label{E41}
\begin{array}{l}
w_{\alpha}(a) = \displaystyle{g_{\alpha \alpha}(a,b)\Bigr |_{b^{-1}
\to 0} =a^{-1}e^{\nabla_{\alpha}(a)\p_{\alpha}F},}
\\  \\
w_{\alpha \beta}(a) = 
\displaystyle{g_{\alpha \beta}(a,b)\Bigr |_{b^{-1}
\to 0} =e^{\nabla_{\alpha}(a)\p_{\alpha}F} 
\quad (\mbox{for $\alpha \neq \beta$} )}
\end{array}
\eeq
(recall that the $g$-function is given by (\ref{g-def})).

\begin{theorem}
For all $\alpha \neq \beta$ 
the functions $w_{\alpha}(z)$, $w_{\alpha \beta}(z)$ are constrained
by the equation
\beq\label{E9}
R_{\alpha \beta}^2 \Bigl (w_{\alpha}^2 w_{\alpha \beta}^2 +1\Bigr )-
\Bigl (w_{\alpha}^2 + w_{\alpha \beta}^2\Bigr )+V_{\alpha \beta}
w_{\alpha}w_{\alpha \beta}=0,
\eeq
where $R_{\alpha \beta}=e^{\p_{\alpha}\p_{\beta}F}$ 
(as in (\ref{E51})) and
\beq\label{V}
V_{\alpha \beta}=2e^{(\p_{\beta}-\p_{\alpha})\p_{\alpha}F}
\p_{\beta}\p_{t_{\alpha , 1}}F.
\eeq
\end{theorem}

\noindent
{\it Proof.}
In the notation (\ref{E41}) equations (\ref{E21}), 
(\ref{E31}) acquire the form
\beq\label{E6}
\begin{array}{l}
R_{\alpha \beta}\Bigl (w_{\alpha}(a)w_{\alpha \beta}(a)-
w_{\alpha}(b)w_{\alpha \beta}(b)\Bigr )
\\ \\
\phantom{aaaaaaaaaaa}
=(a^{-1}-b^{-1})e^{\nabla_{\alpha}(a)\nabla_{\alpha}(b)F}
\Bigl (w_{\alpha \beta}(a)w_{\alpha \beta}(b)-
w_{\alpha}(a)w_{\alpha}(b)\Bigr ),
\end{array}
\eeq

\beq\label{E7}
\begin{array}{l}
\phantom{aa}w_{\alpha}(a)w_{\alpha \beta}(b)-
w_{\alpha}(b)w_{\alpha \beta}(a)
\\ \\
\phantom{aaaaaaaaaaa}
=R_{\alpha \beta}(a^{-1}-b^{-1})e^{\nabla_{\alpha}(a)\nabla_{\alpha}(b)F}
\Bigl (1- w_{\alpha}(a)w_{\alpha}(b)w_{\alpha \beta}(a)w_{\alpha \beta}(b)
\Bigr ).
\end{array}
\eeq

\noindent
Excluding $e^{\nabla_{\alpha}(a)\nabla_{\alpha}(b)F}$, 
we obtain a relation for the functions $w_{\alpha}$ and
$w_{\alpha \beta}$. After some simple
transformations it can be represented in the form
$$
\begin{array}{l}
\displaystyle{
\phantom{aa}R_{\alpha \beta}^2 \left (w_{\alpha}(a)w_{\alpha \beta}(a)+
\frac{1}{w_{\alpha}(a)w_{\alpha \beta}(a)}\right ) -
\frac{w_{\alpha}(a)}{w_{\alpha \beta}(a)}-
\frac{w_{\alpha \beta}(a)}{w_{\alpha}(a)}}
\\ \\
=\, \displaystyle{
R_{\alpha \beta}^2 \left (w_{\alpha}(b)w_{\alpha \beta}(b)+
\frac{1}{w_{\alpha}(b)w_{\alpha \beta}(b)}\right ) -
\frac{w_{\alpha}(b)}{w_{\alpha \beta}(b)}-
\frac{w_{\alpha \beta}(b)}{w_{\alpha}(b)}}.
\end{array}
$$
The left-hand side depends only on $a$ while the right-hand side depends
only on $b$, hence both equal to some constant (meaning that
it does not depend on $a$ or $b$ but can depend on 
${\bf t}$) which we denote
as $-V_{\alpha \beta}$:
\beq\label{E8}
R_{\alpha \beta}^2 \left (w_{\alpha}(z)w_{\alpha \beta}(z)+
\frac{1}{w_{\alpha}(z)w_{\alpha \beta}(z)}\right ) -
\frac{w_{\alpha}(z)}{w_{\alpha \beta}(z)}-
\frac{w_{\alpha \beta}(z)}{w_{\alpha}(z)}=-V_{\alpha \beta}.
\eeq
The constant can be found from the $z\to \infty$ limit of 
(\ref{E8}). The result is given by (\ref{V}). 
Equation (\ref{E8}) is the same as
(\ref{E9}).
\square

The left-hand side of (\ref{E9}) is
a polynomial $P(w_{\alpha}, w_{\alpha \beta})$ 
quadratic in each of
the two complex variables, hence the equation 
$P(w_{\alpha}, w_{\alpha \beta})=0$ 
defines an elliptic curve. It is the dynamical curve for the
multi-component hierarchy.
The functions (\ref{E41}) are meromorphic 
functions on this curve, and $z^{-1}$ plays the role of a local parameter
in a neighborhood of $\infty$. The both functions are regular at $\infty$.

It is well known that any algebraic curve defined by the equation 
$P(x,y)=0$
with a bi-quadratic polynomial $P(x,y)$ can be uniformized by
elliptic functions.

\begin{proposition}\cite{SZ24,SZ25b}
The dynamical curve defined by equation (\ref{E9}) is uniformized
as follows:
\beq\label{uni}
w_{\alpha}(z)=
\frac{\theta_1(u_{\alpha}(z)-\eta_{\alpha})}{\theta_4(u_{\alpha}(z)-
\eta_{\alpha})}, \qquad
w_{\alpha \beta}(z)=\epsilon_{\beta \alpha}
\frac{\theta_1(u_{\alpha}(z)-\eta_{\beta})}{\theta_4(u_{\alpha}(z)-
\eta_{\beta})} \quad (\mbox{for $\alpha \neq \beta$}),
\eeq 
with
\beq\label{V1}
R_{\alpha \beta}=\epsilon_{\beta \alpha}
\frac{\theta_1(\eta_{\alpha \beta})}{\theta_4(\eta_{\alpha \beta})},
 \quad
V_{\alpha \beta}=
2\epsilon_{\beta \alpha} \frac{\theta_4^2(0)\, \theta_2(\eta_{\alpha \beta})
\, \theta_3(\eta_{\alpha \beta})}{\theta_2(0)\, 
\theta_3(0)\,\theta_4^2(\eta_{\alpha \beta})}.
\eeq
\end{proposition}

\noindent
This means that the equation of the curve is satisfied identically if one
substitutes into (\ref{E9}) 
$w_{\alpha}$, $w_{\alpha \beta}$, 
$R_{\alpha \beta}$, $V_{\alpha \beta}$ as they are 
expressed in (\ref{uni}),
(\ref{E1b}), (\ref{V1}) respectively. More details on uniformization
of elliptic curves by theta-functions are contained in Appendix C. 

\begin{remark}
The right-hand side of formula (\ref{V1}) for $V_{\alpha \beta}$ 
implies that 
$V_{\beta \alpha}=-V_{\alpha \beta}$, although this is not seen from
the original definition (\ref{V}). 
\end{remark}

After these preparations, 
we are ready to answer the question how the
modular parameter $\tau$ should be chosen. 

\begin{proposition}\label{proposition:tau1}
The elliptic modular
parameter $\tau =\tau ({\bf t})$ is a function of times
implicitly determined from the equation
\beq\label{E11a}
\frac{\theta_2^2 (0|\tau )}{\theta_3^2 (0|\tau )}+
\frac{\theta_3^2 (0|\tau )}{\theta_2^2 (0|\tau )}=
e^{2\p_{\alpha}\p_{\beta}F}+e^{-2\p_{\alpha}\p_{\beta}F}-
e^{-2\p_{\alpha}^2F}(\p_{\beta}\p_{t_{\alpha ,1}}F)^2 .
\eeq
\end{proposition}

\noindent
{\it Proof.}
Note first of all that
\beq\label{E10}
\frac{V_{\alpha \beta}}{R_{\alpha \beta}}=
\frac{2\, S'(\eta_{\alpha \beta})}{\pi \theta_2(0)\, \theta_3(0)}
\eeq
(see (\ref{E1b}), (\ref{S'}) and (\ref{V1})).
Then identity (\ref{S1}) can be written as
$$
R_{\alpha \beta}^2 +R_{\alpha \beta}^{-2} -
\Bigl ( \frac{V_{\alpha \beta}}{2R_{\alpha \beta}}\Bigr )^2 =
\frac{\theta_2^2 (0)}{\theta_3^2 (0)}+\frac{\theta_3^2 (0)}{\theta_2^2 (0)},
$$
which is (\ref{E11a}).
\square

It remains to explain how the curve (\ref{E9}) 
is related to the curve (\ref{ddkp-curve-ep}) 
from Section \ref{section:dDKP} which is defined by another
polynomial equation. In fact, it is one and the same curve,
although defined in different equivalent ways.
To see this, we introduce, in addition to (\ref{E41}), the 
functions
\beq\label{E41a}
\begin{array}{l}
\displaystyle{ \left.
p_{\alpha}(a) = -\frac{\p_{b^{-1}}
g_{\alpha \alpha}(a,b)}{g(a,b)} \right |_{b^{-1}\to 0}
=a-\nabla_{\alpha}(a)\p_{t_{\alpha , 1}}F},
\\  \\
\displaystyle{\left.
p_{\alpha \beta}(a)=-\frac{ \p_{b^{-1}}
g_{\alpha  \beta}(a,b)}{g(a,b)} \right |_{b^{-1}\to 0}
=-\nabla_{\alpha}(a)\p_{t_{\beta , 1}}F
\quad (\mbox{for $\alpha \neq \beta$} ),}
\end{array}
\eeq
generalizing the definition (\ref{wp}) to the multi-component case.

\begin{proposition}
For all $\alpha \neq \beta$ the functions $p_{\alpha}(z)$, 
$p_{\alpha \beta}(z)$ 
satisfy the polynomial equations
\beq\label{p1}
\begin{array}{l}
R_{\alpha}^2 \Bigl (w_{\alpha}^2(z)+w_{\alpha}^{-2}(z)\Bigr )=p_{\alpha}^2(z)
+V_{\alpha},
\\ \\
R_{\alpha}^2 \Bigl (w_{\beta \alpha}^2(z)+w_{\beta 
\alpha}^{-2}(z)\Bigr )=p_{\beta \alpha}^2(z)
+V_{\alpha} \quad \mbox{for $\alpha \neq \beta$},
\end{array}
\eeq
where
\beq\label{p2}
R_{\alpha}=e^{\p^2_{\alpha}F}, \quad
V_{\alpha}=(\p_{\alpha}\p_{t_{\alpha , 1}}F)^2 +2 \p_{t_{\alpha ,1}}^2 F-
\p_{\alpha}\p_{t_{\alpha , 2}}F.
\eeq
\end{proposition}

\noindent
{\it Sketch of proof.}
Equations (\ref{p1}) can be obtained from
(\ref{g1}), (\ref{g2}) in the same way as
(\ref{ddkp-curve-ep}) 
was obtained from
(\ref{eq1}), (\ref{eq2}). Namely, from (\ref{g1}), (\ref{g2}) 
it follows that 
\beq\label{ddkp-sep-vars2}
\begin{array}{l}
\displaystyle{
\frac
{\Bigl( 1 + g_{\alpha \nu}^{2}(a, c) g_{\alpha \mu}^{2}(a, d)
\Bigl)g_{\mu \nu}^{2}(d, c)
- \Bigl( g_{\alpha \nu}^{2}(a, c) + g_{\alpha \mu}^{2}(a, d) \Bigr)}
{g_{\alpha \nu}(a, c) g_{\alpha \mu}(a, d)}}
\\ \\ =
\displaystyle{
\frac
{\Bigl(1 +g_{\beta \nu}^{2}(b, c) g_{\beta \mu}^{2}(b, d)
\Bigl)g_{\mu \nu}^{2}(d, c)
-\Bigl( g_{\beta \nu}^{2}(b, c) + g_{\beta \mu}^{2}(b, d) \Bigl)}
{g_{\beta \nu}(b, c) g_{\beta \mu}(b, d)}}.
\end{array}
\eeq
Performing the limit $c^{-1}, d^{-1}\to 0$, we obtain from this
two equations (one for $\beta =\alpha$, 
another for $\beta \neq \alpha$) in which the variables $a$ and $b$
are separated:
$$
\begin{array}{l}
R^2_{\alpha}\Bigl (w^2_{\alpha}(a)+w^{-2}_{\alpha}(a)\Bigr )-
p_{\alpha}^2(a)=
R^2_{\alpha}\Bigl (w^2_{\alpha}(b)+w^{-2}_{\alpha}(b)\Bigr )-
p_{\alpha}^2(b),
\\ \\
R^2_{\alpha}\Bigl (w^2_{\beta \alpha}(a)+w^{-2}_{\beta \alpha}(a)\Bigr )-
p_{\beta \alpha}^2(a)=
R^2_{\alpha}\Bigl (w^2_{\beta \alpha}(b)+w^{-2}_{\beta \alpha}(b)\Bigr )-
p_{\beta \alpha}^2(b).
\end{array}
$$
These equations mean that the left-hand sides do not depend on $a$.
Evaluating their values as $a\to \infty$, we obtain (\ref{p1}), (\ref{p2}).
\square

\begin{remark}
For each $\alpha$,
the first equation in (\ref{p1}) coincides with
(\ref{ddkp-curve-ep}). This could be expected from the very
beginning
because the restriction of the multi-component hierarchy to
each component is equivalent to the one-component hierarchy
considered in Section \ref{section:dDKP}. The second equation
defines the same curve which, moreover, is isomorphic to the 
curve (\ref{E9}).
The functions $p_{\alpha}, p_{\beta \alpha}$
are expressed through theta-functions as follows:
\beq\label{p3}
\begin{array}{lll}
p_{\alpha}(z)& =& \displaystyle{
\frac{\gamma_{\alpha}}{\pi}S'(u_{\alpha}(z)-\eta_{\alpha}),}
\\ && \\
&=&\displaystyle{
\gamma_{\alpha}\theta_4^2(0)\, \frac{\theta_2(u_{\alpha}(z)-\eta_{\alpha})
\theta_3(u_{\alpha}(z)-
\eta_{\alpha})}{\theta_1(u_{\alpha}(z)-\eta_{\alpha})
\theta_4(u_{\alpha}(z)-\eta_{\alpha})}},
\end{array}
\eeq

\beq\label{p3a}
\begin{array}{lll}
p_{\beta \alpha}(z)& =& \displaystyle{
\frac{\gamma_{\alpha}}{\pi}S'(u_{\beta}(z)-\eta_{\alpha}),}
\\ && \\
&=&\displaystyle{
\gamma_{\alpha}\theta_4^2(0)\, \frac{\theta_2(u_{\beta}(z)-\eta_{\alpha})
\theta_3(u_{\beta}(z)-
\eta_{\alpha})}{\theta_1(u_{\beta}(z)-\eta_{\alpha})
\theta_4(u_{\beta}(z)-\eta_{\alpha})}}.
\end{array}
\eeq
where $\gamma_{\alpha}=\gamma_{\alpha}({\bf t})$ is a dynamical variable and
the elliptic function $S'(u)$ is given by (\ref{S'}). These formulas 
should be supplemented by expressions for $R_{\alpha}$ and $V_{\alpha}$:
\beq\label{p4}
R_{\alpha}=\gamma_{\alpha}\theta_2(0)\theta_3(0), \qquad
V_{\alpha}=\gamma_{\alpha}^2 \Bigl (\theta_2^4(0)+\theta_3^4(0)\Bigr ).
\eeq
Substituting all this into equations (\ref{p1}), one can see that they 
are satisfied identically.
\end{remark}

\begin{remark}
Equations (\ref{p4}) give an alternative 
formula to determine the modular
parameter, which is similar to (\ref{RV1}).
From equations (\ref{p4}) we conclude that
\beq\label{E11a1}
\frac{V_{\alpha}}{R^2_{\alpha}}
 =e^{-2\p_{\alpha}^2F}\Bigl ((\p_{\alpha}
 \p_{t_{\alpha , 1}}F)^2 +2 \p_{t_{\alpha ,1}}^2 F-
\p_{\alpha}\p_{t_{\alpha , 2}}F\Bigr )=
\frac{\theta_2^2 (0|\tau )}{\theta_3^2 (0|\tau )}+
\frac{\theta_3^2 (0|\tau )}{\theta_2^2 (0|\tau )},
\eeq
where the right-hand side (the same as in (\ref{E11a})) depends only
on the modular parameter $\tau$ while the left-hand side in general 
depends on all the times according to equations of the hierarchy.
\end{remark}

\subsubsection{The $F_1$-function}
\label{section:F1-DKP2}

Similarly to how it was done in Section \ref{section:F1-DKP} for
the one-component dDKP, one can obtain a linear equation for the
$F_1$-function in the expansion
$F=F_0 +\hbar F_1 +O(\hbar^2)$
in the muti-component case. To this end,
one should expand the $\hbar$-version
of equation
(\ref{gen-H-M(multi-DKP)}) up to the first order in $\hbar$.
Omitting the details,
we present the result.

\begin{proposition}
The function $F_1$ for the $N$-component dDKP hierarchy
satisfies the homogeneous linear equation
\beq\label{exp5}
\begin{array}{l}
\displaystyle{
\sum_{s=1}^{P^+}
\prod_{{\scriptsize \begin{array}{l}i=1\\ i\neq s \end{array}}}^{P^+}
E_{\alpha_i , \alpha_s}^{-1}(a_s, a_i)
\prod_{k=1}^{P^-} E_{\beta_k , \alpha_s}(a_s, b_k)}
\\ \\
\phantom{aaaaaaaaaaa}
\times \,
e^{\nabla_{\alpha_s}(a_s)(S^- -S^+ +\nabla_{\alpha_s}(a_s))F_0}
\nabla_{\alpha_s}(a_s)(S^- -S^+ +\nabla_{\alpha_s}(a_s))F_1
\\ \\
\displaystyle{
+\,\, \sum_{s=1}^{P^-}
\prod_{{\scriptsize \begin{array}{l}k=1\\ k\neq s \end{array}}}^{P^-}
E_{\beta_k , \beta_s}^{-1}(b_s, b_k)
\prod_{i=1}^{P^+} E_{\alpha_i , \beta_s}(b_s, a_i)}
\\ \\
\phantom{aaaaaaaaaaa}
\times \, e^{\nabla_{\beta_s}(b_s)(S^+ -S^- +\nabla_{\beta_s}(b_s))F_0}
\nabla_{\alpha_s}(a_s)(S^- -S^+ +\nabla_{\beta_s}(b_s))F_1
=  0,
\end{array}
\eeq
where the operators $S^{\pm}$ are 
$\displaystyle{
S^+ =\sum_{i=1}^{P^+} \nabla_{\alpha_i}(a_i), \;
S^- =\sum_{k=1}^{P^-} \nabla_{\beta_k}(b_k).}$
\end{proposition}

\noindent
Recall that $P^+ + P^-$ is even.
The function $F_0$ in (\ref{exp5}) satisfies equation
(\ref{gen-H-M}).
Like in the previous cases, the function 
\beq\label{F1v1}
F_1=\p_{v}F_0,
\eeq
where
$v$ is any continuous parameter of solutions to (\ref{gen-H-M}),
satisfies equation (\ref{exp5}).

\subsection{Multi-component Pfaff-Toda hierarchy and its 
dispersionless limit}

\subsubsection{Multi-component Pfaff-Toda hierarchy}

In the $N$-component Pfaff-Toda lattice  
hierarchy the independent variables are
$2N$ infinite sets of ``times'',
\beq\label{PT1}
\begin{array}{l}
{\bf t}=\{{\bf t}_1, {\bf t}_2, \ldots , {\bf t}_N\}, \qquad
{\bf t}_{\alpha}=\{t_{\alpha , 1}, t_{\alpha , 2}, t_{\alpha , 3}, 
\ldots \, \},
\\ \\
\bar {\bf t}=\{\bar {\bf t}_1, \bar {\bf t}_2, \ldots , 
\bar {\bf t}_N\}, \qquad
\bar {\bf t}_{\alpha}=\{\bar t_{\alpha , 1}, \bar t_{\alpha , 2}, 
\bar t_{\alpha , 3}, \ldots \, \},
\end{array}
\qquad \alpha = 1, \ldots , N
\eeq
and two finite sets of discrete variables
$$
{\bf n}=\{n_1, \ldots , n_N\}, \quad
\bar {\bf n}=\{\bar n_1, \ldots , \bar n_N\}, \quad n_{\alpha}, \bar n_{\alpha}
\in \ZZ 
$$
such that
\beq\label{PT2}
|{\bf n}|-|\bar {\bf n}|\in 2 \ZZ .
\eeq
The universal
dependent variable is the tau-function
$\tau ({\bf n}, \bar {\bf n}, {\bf t},
\bar {\bf t})$. In the fermionic approach it is defined
as the following expectation value:
\beq\label{PT3}
\tau ({\bf n}, \bar {\bf n}, {\bf t}, \bar {\bf t})=
\lbr {\bf n}\bigl | e^{J({\bf t})} g e^{-\bar J(\bar {\bf t})} 
\bigr |-\bar {\bf n}\rbr ,
\eeq
where $g$ is a Clifford group element of the form (\ref{g-mdkp}).
The representation (\ref{PT3}) has the same form as the one for
the multi-component Toda lattice (see (\ref{T3})) but the restriction
on possible values of ${\bf n}$, $\bar {\bf n}$ is much weaker:
instead of $|{\bf n}|=-|\bar {\bf n}|$ now
$|{\bf n}|$ and $|\bar {\bf n}|$ are only required to be of the 
same parity (otherwise the expectation value (\ref{PT3}) is
identically zero).

The general bilibear equation for the tau-function
obtained in \cite{SZ25a} has the form
\beq\label{main}
\begin{array}{l}
\displaystyle{
\sum_{\gamma =1}^N\epsilon_{\gamma}({\bf n})\epsilon_{\gamma}({\bf n}')
\oint_{C_{\infty}}\frac{dz}{z^2} 
z^{n_{\gamma}-n_{\gamma}'}
e^{\xi ({\bf t}_{\gamma}-{\bf t}_{\gamma}', z)}}
\\ \\
\phantom{aaaaaaaaaaa}\displaystyle{\times \, 
\tau ({\bf n}\! -\! {\bf e}_{\gamma}, \bar {\bf n}, 
{\bf t}\! -\! [z^{-1}]_{\gamma},
\bar {\bf t})\tau ({\bf n}'\! +\! {\bf e}_{\gamma}, \bar {\bf n}', 
{\bf t}'\! +\! [z^{-1}]_{\gamma},
\bar {\bf t}')}
\\ \\
+\, 
\displaystyle{
\sum_{\gamma =1}^N \epsilon_{\gamma}({\bf n})\epsilon_{\gamma}({\bf n}')
\oint_{C_{\infty}}\frac{dz}{z^2} 
z^{n_{\gamma}'-n_{\gamma}}
e^{-\xi ({\bf t}_{\gamma}-{\bf t}_{\gamma}', z)}}
\\ \\
\phantom{aaaaaaaaaaa}\displaystyle{\times \, 
\tau ({\bf n}\! +\! {\bf e}_{\gamma}, \bar {\bf n}, 
{\bf t}\! +\! [z^{-1}]_{\gamma},
\bar {\bf t})\tau ({\bf n}'\! -\! {\bf e}_{\gamma}, \bar {\bf n}', 
{\bf t}'\! -\! [z^{-1}]_{\gamma},
\bar {\bf t}')}
\end{array}
\eeq
$$
\begin{array}{l}
=\, 
\displaystyle{
\sum_{\gamma =1}^N \epsilon_{\gamma}(\bar {\bf n})
\epsilon_{\gamma}(\bar {\bf n}')
\oint_{C_{\infty}}\frac{dz}{z^2} 
z^{\bar n_{\gamma}-\bar n_{\gamma}'}
e^{\xi (\bar {\bf t}_{\gamma}-\bar {\bf t}_{\gamma}', z)}}
\\ \\
\phantom{aaaaaaaaaaa}\displaystyle{\times \, 
\tau ({\bf n}, \bar {\bf n}\! -\! {\bf e}_{\gamma}, 
{\bf t},
\bar {\bf t}\! -\! [z^{-1}]_{\gamma})
\tau ({\bf n}', \bar {\bf n}'\! +\! {\bf e}_{\gamma}, 
{\bf t}',
\bar {\bf t}'\! +\! [z^{-1}]_{\gamma})}
\\ \\
+\, 
\displaystyle{
\sum_{\gamma =1}^N \epsilon_{\gamma}(\bar {\bf n})
\epsilon_{\gamma}(\bar {\bf n}')
\oint_{C_{\infty}}\frac{dz}{z^2} 
z^{\bar n_{\gamma}'-\bar n_{\gamma}}
e^{-\xi (\bar {\bf t}_{\gamma}-\bar {\bf t}_{\gamma}', z)}}
\\ \\
\phantom{aaaaaaaaaaa}\displaystyle{\times \, 
\tau ({\bf n}, \bar {\bf n}\! +\! {\bf e}_{\gamma}, 
{\bf t},
\bar {\bf t}\! +\! [z^{-1}]_{\gamma})
\tau ({\bf n}', \bar {\bf n}'\! -\! {\bf e}_{\gamma}, 
{\bf t}',
\bar {\bf t}'\! -\! [z^{-1}]_{\gamma})}.
\end{array}
$$
It is valid for all ${\bf t}, \bar {\bf t}$, ${\bf t}', \bar {\bf t}'$
and ${\bf n}, \bar {\bf n}$,  ${\bf n}', \bar {\bf n}'$
such that 
\beq\label{par1}
|{\bf n}|-|\bar {\bf n}|\in 2\ZZ +1, \qquad
|{\bf n}'|-|\bar {\bf n}'|\in 2\ZZ +1,
\eeq
otherwise the parity condition is not satisfied 
(because $|{\bf n}\pm {\bf e}_{\gamma}|$ and $|\bar {\bf n}|$
as well as  $|{\bf n'}\pm {\bf e}_{\gamma}|$ and $|\bar {\bf n'}|$
are then of different parities).
At $N=1$ equation (\ref{main}) coincides (after a linear change 
of the discrete variables) with
the equation for the tau-function of the one-component Pfaff-Toda
hierarchy obtained by Takasaki 
in \cite{Takasaki09}. 
Note that equation (\ref{main}) has the following obvious 
symmetries:
\beq\label{sym1b}
({\bf n}, {\bf t}, \bar {\bf n}, \bar {\bf t}),
({\bf n}', {\bf t}', \bar {\bf n}', \bar {\bf t}')
\longleftrightarrow
({\bf n}', {\bf t}', \bar {\bf n}', \bar {\bf t}'),
({\bf n}, {\bf t}, \bar {\bf n}, \bar {\bf t}),
\eeq
and
\beq\label{sym2}
({\bf n}, {\bf t}, {\bf n}', {\bf t}'),
(\bar {\bf n},\bar {\bf t}, \bar {\bf n}', \bar {\bf t}')
\longleftrightarrow
(\bar {\bf n},\bar {\bf t}, \bar {\bf n}', \bar {\bf t}'),
({\bf n}, {\bf t}, {\bf n}', {\bf t}').
\eeq 

For any $N\geq 1$, after setting 
$\bar {\bf n}'=\bar {\bf n}$, $\bar {\bf t}'=\bar {\bf t}$
in (\ref{main}) the bar-variables do not
participate in the equation entering as parameters. 
The right-hand side of (\ref{main}) vanishes identically
and the rest becomes the integral bilinear
equation for the tau-function of the $N$-component 
DKP hierarchy (\ref{mDKP-bil-rel}). In this hierarchy,
the independent variables are ${\bf n}$ and ${\bf t}$, and
the tau-function will be denoted as
$\tau^{\rm DKP}({\bf n}, {\bf t})$. So, DKP can be regarded as
a subhierarchy of the Pfaff-Toda.
On the other hand, the $2N$-component DKP
is equivalent to the $N$-component Pfaff-Toda, in the way 
which is similar to the relation between the $2N$-component KP
and $N$-component Toda  lattice described in Section 
\ref{section:multi-Toda}.

Let us present here the main points of this identification,
following \cite{SZ25b}.
For the case of $M$-component DKP hierarchy the integral
bilinear equation reads:
\beq\label{main1}
\begin{array}{l}
\displaystyle{
\sum_{\gamma =1}^M\epsilon_{\gamma}({\bf n}\! -\! {\bf n}')
\oint_{C_{\infty}}\! \frac{dz}{z^2} 
z^{n_{\gamma}-n_{\gamma}'}
e^{\xi ({\bf t}_{\gamma}-{\bf t}_{\gamma}', z)}}
\\ \\
\phantom{aaaaaaaaaaaaa}
\displaystyle{
\times \tau^{\rm DKP} ({\bf n}\! -\! {\bf e}_{\gamma}, 
{\bf t}\! -\! [z^{-1}]_{\gamma})\tau^{\rm DKP} ({\bf n}'\! +\! {\bf e}_{\gamma}, 
{\bf t}'\! +\! [z^{-1}]_{\gamma})}
\\ \\
+\, 
\displaystyle{
\sum_{\gamma =1}^M \epsilon_{\gamma}({\bf n}\! -\! {\bf n}')
\oint_{C_{\infty}}\! \frac{dz}{z^2} 
z^{n_{\gamma}'-n_{\gamma}}
e^{-\xi ({\bf t}_{\gamma}-{\bf t}_{\gamma}', z)}}
\\ \\
\phantom{aaaaaaaaaaaaa}
\displaystyle{
\times \tau^{\rm DKP} ({\bf n}\! +\! {\bf e}_{\gamma}, 
{\bf t}\! +\! [z^{-1}]_{\gamma})\tau^{\rm DKP} ({\bf n}'\! -\! {\bf e}_{\gamma}, 
{\bf t}'\! -\! [z^{-1}]_{\gamma})}=0
\end{array}
\eeq
(see (\ref{mDKP-bil-rel})).

\begin{proposition}\cite{SZ25b}
Equation 
(\ref{main}) is equivalent to (\ref{main1}) at $M=2N$,
and the tau-functions $\tau^{\rm PT} $ 
and $\tau^{\rm DKP}$ of the Pfaff-Toda and DKP hierarchies
are related as
\beq\label{tautau1}
\tau^{\rm PT} ({\bf n}, \bar {\bf n}, {\bf t}, \bar {\bf t})=
(-1)^{\frac{1}{2}|\bar {\bf n}| (|\bar {\bf n}|-1)}
\tau^{\rm DKP}(\tilde {\bf n}, \tilde {\bf t}).
\eeq
The sets of variables $\tilde {\bf n}, \tilde {\bf t}$ 
are $\tilde {\bf n}=\{n_1, \ldots , n_N, \bar n_1, \ldots , \bar n_N\},
\;
 \tilde {\bf t}=\{{\bf t}_1, \ldots , {\bf t}_N, 
 \bar {\bf t}_1, \ldots , \bar {\bf t}_N\}.$
\end{proposition}

\noindent
To see this, we re-denote the variables in (\ref{main1}) in 
the same way as this was done in Section \ref{section:multi-Toda}, 
i.e. $n_{N+\mu}=\bar n_{\mu}$, 
${\bf t}_{N+\mu}=\bar {\bf t}_{\mu}$, where 
the index $\mu$ runs from $1$ to $N$. It remains to
divide each sum over $\gamma$ in (\ref{main1}) in two
(one from $1$ to $N$,
the other one from $N+1$ to $2N$) and repeat the arguments from
Section \ref{section:multi-Toda}.

\begin{remark}
The parity condition on the DKP side
is $|{\bf n}|+|\bar {\bf n}|\in 2\ZZ$,
while that on the Pfaff-Toda side is
$|{\bf n}|-|\bar {\bf n}|\in 2\ZZ$ which is the same.
\end{remark}

\subsubsection{The dispersionless limit}

The equivalence established in the previous subsection allows one
to obtain 
the dispersionless version of the $N$-component 
Pfaff-Toda hierarchy
by rearranging the equations for the $2N$-component dDKP.
The details are explained in Section
\ref{section:disp-Toda}. Here we present only the result.
Equation (\ref{E11}) for the $F$-function of the $2N$-component 
dDKP hierarchy, being rewritten in terms of the Pfaff-Toda
variables, have the form of the system
\beq\label{final1}
\left \{
\begin{array}{rcl}
E_{\beta \alpha}(a,b) e^{\nabla_{\alpha}(a)\nabla_{\beta}(b)F}
&=& \displaystyle{ \frac{\theta_1 (u_{\alpha}(a)- 
u_{\beta}(b))}{\theta_4 (u_{\alpha}(a)- 
u_{\beta}(b))},}
\\ && \\
e^{\nabla_{\alpha}(a)\bar \nabla_{\beta}(b)F}& =& \displaystyle{
\frac{\theta_1 (u_{\alpha}(a)+ 
\bar u_{\beta}(b))}{\theta_4 (u_{\alpha}(a)+ 
\bar u_{\beta}(b))},}
\\ && \\
E_{\beta \alpha}(a,b) e^{\bar \nabla_{\alpha}(a)\bar \nabla_{\beta}(b)F}
& =& \displaystyle{\frac{\theta_1 (\bar u_{\alpha}(a)- 
\bar u_{\beta}(b))}{\theta_4 (\bar u_{\alpha}(a)- 
\bar u_{\beta}(b))}},
\end{array} \right.
\eeq
where $F=F_0$ is the $F$-function for the 
dispersionless Pfaff-Toda hierarchy.
The indices $\alpha , \beta$ run from $1$ to $N$.
Note that equations (\ref{final1}) look like an elliptic 
deformation of (\ref{final}): the trigonometric function
$\sin u$ is replaced by its elliptic counterpart
$\theta_1(u)/\theta_4(u)$ which is (up to some technical
details) the ``elliptic sinus'' function.
In particular, at $N=1$ the system (\ref{final1}) has the form
\beq\label{final3}
\left \{
\begin{array}{rcl}
(a^{-1}-b^{-1})e^{\nabla(a)\nabla(b)F}
&=& \displaystyle{ \frac{\theta_1 (u(a)- 
u(b))}{\theta_4 (u(a)- 
u(b))},}
\\ && \\
e^{\nabla(a)\bar \nabla(b)F}& =& \displaystyle{
\frac{\theta_1 (u(a)+ 
\bar u(b))}{\theta_4 (u(a)+ 
\bar u(b))},}
\\ && \\
(a^{-1}-b^{-1})e^{\bar \nabla(a)\bar \nabla (b)F}
& =& \displaystyle{\frac{\theta_1 (\bar u(a)- 
\bar u(b))}{\theta_4 (\bar u(a)- 
\bar u(b))}},
\end{array} \right.
\eeq
which coincides with the result first obtained in \cite{AZ15}.

\subsection{Multi-component large BKP and its dispersionless versions}

The set of independent variables is the same as the one
for the $N$-component KP hierarchy: 
$$
{\bf t}=\{{\bf t}_1 , \ldots , {\bf t}_N \}, \qquad
{\bf t}_{\alpha}=\{t_{\alpha , 1}, t_{\alpha , 2}, t_{\alpha , 3}, 
\ldots \}
$$
and $N$ discrete variables ${\bf n}=\{n_1, \ldots , n_N\}$, 
$n_{\alpha}\in \ZZ$. 
The only difference is that their sum $|{\bf n}|$ can be arbitrary.

\subsubsection{Multi-component large BKP}

The bilinear equation for the tau-function $\tau ({\bf n}, {\bf t})$ 
has the form
\beq\label{L1a}
\begin{array}{ll}
&\displaystyle{ 
\frac{1}{2\pi i} \sum_{\gamma =1}^N \epsilon_{\gamma}
({\bf n}\! - \! {\bf n}')
\oint_{C_{\infty}} \! \frac{dz}{z^2} 
z^{n_{\gamma}-n'_{\gamma}} e^{\xi ({\bf t}_{\gamma}-{\bf t}'_{\gamma}, z)}
\tau \Bigl ( {\bf n}\! -\! {\bf e}_{\gamma}, 
{\bf t}\! -\! [z^{-1}]_{\gamma}\Bigr )
\tau \Bigl ( {\bf n}' \! +\! {\bf e}_{\gamma}, 
{\bf t}' \! +\! [z^{-1}]_{\gamma}\Bigr )}
\\ &\\
+& \displaystyle{
\frac{1}{2\pi i} 
\sum_{\gamma =1}^N \epsilon_{\gamma}({\bf n}\! -\! {\bf n}')
\oint_{C_{\infty}} \! \frac{dz}{z^2} 
z^{n'_{\gamma}-n_{\gamma}} e^{-\xi ({\bf t}_{\gamma}-{\bf t}'_{\gamma}, z)}
\tau \Bigl ({\bf n}\! +\! {\bf e}_{\gamma} , 
{\bf t}\! +\! [z^{-1}]_{\gamma}\Bigr )
\tau \Bigl ( {\bf n}'\! -\! {\bf e}_{\gamma},
 {\bf t}'\! -\! [z^{-1}]_{\gamma}\Bigr )}
\\ &\\
=& \frac{1}{2} \Bigl (1-(-1)^{|{\bf n}-{\bf n}'|}\Bigr )
\tau ({\bf n}, {\bf t})\tau ({\bf n}', {\bf t}'),
\end{array}
\eeq
which is valid for all ${\bf n}, {\bf n}'$, ${\bf t}, {\bf t}'$.
If $|{\bf n}|$ and $|{\bf n}'|$ are of the same parity 
(both even or both odd),
the right-hand side vanishes and (\ref{L1a}) becomes the 
integral bilinear equation for the $N$-component 
DKP hierarchy (\ref{mDKP-bil-rel}). 
Similarly to the one-component case,
the full set of equations can
be divided into three groups: 
the ``even'' sector consisting of equations that connect
tau-functions with even $|{\bf n}|$ (the ``even'' copy of DKP), 
the ``odd'' sector consisting of equations that connect
tau-functions with odd $|{\bf n}|$ (the ``odd'' copy of DKP) and
equations that ``intertwine'' the even and odd sectors
(they connect tau-functions $\tau ({\bf n}, {\bf t})$,
$\tau ({\bf m}, {\bf t})$ with $|{\bf n}-{\bf m}|\in 2\ZZ +1$).
Note the obvious symmetry of (\ref{L1a}):
\beq\label{sym1a}
({\bf n}, {\bf t}) \longleftrightarrow ({\bf n}', {\bf t}').
\eeq 

From the proof of Proposition
\ref{proposition:simplest11} it follows that
the function (\ref{simplest1})
is simultaneously
a solution (the simplest one) to (\ref{L1a}). 
Moreover, the following statement establishes a more general
connection between the
(multi-component) DKP and large BKP hierarchies.

\begin{proposition}
Fix some natural number $M$ such that $1 \leq M \leq N-1$ and
divide the set of independent variables 
$\{ {\bf n}, {\bf t} \}$ into two subsets:
$
\{ {\bf n}, {\bf t} \}=
\{ {\bf n}_{\rm I}, {\bf t}_{\rm I} \}\cup
\{ {\bf n}_{\rm II}, {\bf t}_{\rm II} \}
$,
where
$$
\begin{array}{l}
{\bf n}_{\rm I}=\{n_1, \ldots , n_M\}, \quad
{\bf n}_{\rm II}=\{n_{M+1}, \ldots , n_N\}
\\ \\
{\bf t}_{\rm I}=\{{\bf t}_1, \ldots , {\bf t}_M\},\;\; \quad
{\bf t}_{\rm II}=\{{\bf t}_{M+1}, \ldots , {\bf t}_N\}.
\end{array}
$$
If $\tau ({\bf n}, {\bf t})$ of the form
\beq\label{tautau111}
\tau ({\bf n}, {\bf t})=\tau ({\bf n}_{\rm I}, {\bf t}_{\rm I})
\exp \left ( \frac{1}{2}\sum_{\gamma =M+1}^N 
\sum_{k\geq 1} kt_{\gamma , k}^2 \right )
\eeq
solves the
$N$-component DKP hierarchy (\ref{mDKP-bil-rel}), then
the function $\tau ({\bf n}_{\rm I}, {\bf t}_{\rm I})$ 
satisfies the bilinear equation (\ref{L1a}) of the 
$M$-component large BKP hierarchy.
\end{proposition}

\noindent
{\it Proof.}
For simplicity, we present the proof for $M=N-1$ (the other cases
can be considered in a similar way). Substituting the 
tau-function of the form (\ref{tautau111}) into 
(\ref{mDKP-bil-rel}), we have:
\beq\label{tautau12}
\begin{array}{l}
\displaystyle{
\sum_{\gamma =1}^{N-1}\epsilon_{\gamma}({\bf n}_{\rm I})
\epsilon_{\gamma}({\bf n}'_{\rm I})
\oint_{C_{\infty}}\! \frac{dz}{z^2}
\left ( 
z^{n_{\gamma}-n_{\gamma}'}
e^{\xi ({\bf t}_{\gamma}-{\bf t}_{\gamma}', z)}\vphantom{\frac{A}{B}}
\right. }
\\ \\
\displaystyle{ \phantom{aaaaaaaaaaaaaaaaaaaaaa} \times \,
\tau ({\bf n}_{\rm I}\! -\! {\bf e}_{\gamma}, 
{\bf t}_{\rm I}\! -\! [z^{-1}]_{\gamma})
\tau ({\bf n}_{\rm I}'\! +\! {\bf e}_{\gamma}, 
{\bf t}_{\rm I}'\! +\! [z^{-1}]_{\gamma})  }
\\ \\
\phantom{aaaaaaa} 
\displaystyle{+\,  \left.  \vphantom{\int^a_b}
z^{n_{\gamma}'-n_{\gamma}}
e^{-\xi ({\bf t}_{\gamma}-{\bf t}_{\gamma}', z)}
\tau ({\bf n}_{\rm I}\! +\! {\bf e}_{\gamma}, 
{\bf t}_{\rm I}\! +\! [z^{-1}]_{\gamma})
\tau ({\bf n}_{\rm I}'\! -\! {\bf e}_{\gamma}, 
{\bf t}_{\rm I}'\! -\! [z^{-1}]_{\gamma}\right )}
\\ \\
+\, \displaystyle{(-1)^{n_N +n'_N}
\oint_{C_{\infty}}\! \frac{dz}{z^2-1}\left (
\vphantom{\frac{A}{B}}
z^{n_N -n'_N}e^{\xi ({\bf t}_N -{\bf t}_N', z)-
\xi ({\bf t}_N -{\bf t}_N', z^{-1})}\right. }
\\ \\
\displaystyle{ \phantom{aaaaaaaaaaaaa}+\, \left.
z^{n'_N -n_N}e^{-\xi ({\bf t}_N -{\bf t}_N', z)+
\xi ({\bf t}_N -{\bf t}_N', z^{-1})}\right )
\tau ({\bf n}_{\rm I}, {\bf t}_{\rm I})
\tau ({\bf n}_{\rm I}',
{\bf t}_{\rm I}')=0
}
\end{array}
\eeq
(the sign factor in front of the last term appears because
$\epsilon_{\gamma}({\bf n})=(-1)^{n_N} 
\epsilon_{\gamma}({\bf n}_{\rm I})$).
Calculating the integral in the last term as in the
proof of Proposition
\ref{proposition:simplest11}, we find that this term is
equal to
$$
-\, \pi i \Bigl (1- (-1)^{n_N -n'_N}\Bigr )
\tau ({\bf n}_{\rm I}, {\bf t}_{\rm I})
\tau ({\bf n}_{\rm I}',
{\bf t}_{\rm I}').
$$
It remains to recall that for the DKP hierarchy
$$|{\bf n}|-|{\bf n'}|=
|{\bf n}_{\rm I}|-|{\bf n'}_{\rm I}|+ n_N -n'_N \in 2\ZZ ,
$$
hence the last term can be written in the form
$$
-\, \pi i \Bigl (1- (-1)^{|{\bf n}_{\rm I}|-|{\bf n'}_{\rm I}|}\Bigr )
\tau ({\bf n}_{\rm I}, {\bf t}_{\rm I})
\tau ({\bf n}_{\rm I}',
{\bf t}_{\rm I}'),
$$
and so (\ref{tautau12}) coincides with the bilinear equation
of the $(N-1)$-component large BKP hierarchy
for the function $\tau ({\bf n}_{\rm I}, {\bf t}_{\rm I})$.
\square

Similar to how it was done in the single-component case, 
it is convenient to introduce different notation for the
tau-functions in the even and odd sectors:
$$
\mbox{$\tau ({\bf n}, {\bf t})$ for even $|{\bf n}|$ and
$\sigma ({\bf n}, {\bf t})$ for odd $|{\bf n}|$}.
$$
Then equation (\ref{L1a}) for $|{\bf n}|-|{\bf n'}|\in 2\ZZ +1$
connecting the two sectors acquires the form
\beq\label{L1b}
\begin{array}{ll}
&\displaystyle{ 
\frac{1}{2\pi i} \sum_{\gamma =1}^N \epsilon_{\gamma}
({\bf n}\! - \! {\bf n}')
\oint_{C_{\infty}} \! \frac{dz}{z^2} 
z^{n_{\gamma}-n'_{\gamma}} e^{\xi ({\bf t}_{\gamma}-{\bf t}'_{\gamma}, z)}
\sigma \Bigl ( {\bf n}\! -\! {\bf e}_{\gamma}, 
{\bf t}\! -\! [z^{-1}]_{\gamma}\Bigr )
\tau \Bigl ( {\bf n}' \! +\! {\bf e}_{\gamma}, 
{\bf t}' \! +\! [z^{-1}]_{\gamma}\Bigr )}
\\ &\\
+& \displaystyle{
\frac{1}{2\pi i} 
\sum_{\gamma =1}^N \epsilon_{\gamma}({\bf n}\! -\! {\bf n}')
\oint_{C_{\infty}} \! \frac{dz}{z^2} 
z^{n'_{\gamma}-n_{\gamma}} e^{-\xi ({\bf t}_{\gamma}-{\bf t}'_{\gamma}, z)}
\sigma \Bigl ({\bf n}\! +\! {\bf e}_{\gamma} , 
{\bf t}\! +\! [z^{-1}]_{\gamma}\Bigr )
\tau \Bigl ( {\bf n}'\! -\! {\bf e}_{\gamma},
 {\bf t}'\! -\! [z^{-1}]_{\gamma}\Bigr )}
\\ &\\
=& \frac{1}{2} \Bigl (1-(-1)^{|{\bf n}-{\bf n}'|}\Bigr )
\tau ({\bf n}, {\bf t})\sigma ({\bf n}', {\bf t}').
\end{array}
\eeq

As before, the substitution 
\beq\label{L2a}
\left \{\begin{array}{l}
\displaystyle{
{\bf n}-{\bf n}'=\sum_{i=1}^{P^+} {\bf e}_{\alpha_i} -
\sum_{k=1}^{P^-} {\bf e}_{\beta_k},}
\\ \\
\displaystyle{
{\bf t}-{\bf t}' =\sum_{i=1}^{P^+}[a_i^{-1}]_{\alpha_i}-
\sum_{k=1}^{P^-}[b_k^{-1}]_{\beta_k}}
\end{array}\right.
\eeq
allows one to apply residue calculus.
The case $P^+-P^-\in 2\ZZ$ corresponds to DKP and was already
considered. Here we are interested in the case
$P^+-P^-\in 2\ZZ +1$. The residue calculus 
yields the following general
Hirota-Miwa equation:
\beq\label{L3a}
\begin{array}{l}
\displaystyle{
\sum_{s=1}^{P^+}
\prod_{{\scriptsize \begin{array}{l}i=1\\ i\neq s \end{array}}}^{P^+}
E_{\alpha_i \alpha_s}^{-1}(a_s, a_i)
\prod_{k=1}^{P^-} E_{\beta_k \alpha_s}(a_s, b_k)
\tau \Bigl ({\bf n}+\sum_{i\neq s}^{P^+}{\bf e}_{\alpha_i}, 
 {\bf t}+\sum_{i\neq s}^{P^+}[a_i^{-1}]_{\alpha_i}\Bigr )}
 \\
 \displaystyle{\phantom{aaaaaaaaaaaaaaaaaaaaaaaaaaa}
\times \, \sigma \Bigl ({\bf n}+{\bf e}_{\alpha_s}+
\sum_{k=1}^{P^-}{\bf e}_{\beta_k}, {\bf t}+[a_s^{-1}]_{\alpha_s}
+\sum_{k=1}^{P^-} [b_k^{-1}]_{\beta_k}\Bigr )}
\end{array}
\eeq
$$
\begin{array}{l}
\displaystyle{
+\, \sum_{s=1}^{P^-}
\prod_{{\scriptsize \begin{array}{l}k=1\\ k\neq s \end{array}}}^{P^-}
E_{\beta_k \beta_s}^{-1}(b_s, b_k)
\prod_{i=1}^{P^+}E_{\alpha_i \beta_s}(b_s, a_i)
\tau \Bigl ( {\bf n}+{\bf e}_{\beta_s}+
\sum_{i=1}^{P^+}{\bf e}_{\alpha_i}, {\bf t}+[b_s^{-1}]_{\alpha_s}+
\sum_{i=1}^{P^+}[a_i^{-1}]_{\alpha_i}, \Bigr )}
\\
\displaystyle{\phantom{aaaaaaaaaaaaaaaaaaaaaaaaaaa}
\times \, \sigma \Bigl (\sum_{k\neq s}^{P^-}{\bf e}_{\beta_k},  {\bf t}+
\sum_{k\neq s}^{P^-} [b_k^{-1}]_{\beta_k}
\Bigr )}
\\ \\
\phantom{aaaaaaaaaaaaa}
=\displaystyle{\sigma \Bigl ({\bf n}+
\sum_{i=1}^{P^+}{\bf e}_{\alpha_i}, 
{\bf t}+\sum_{i=1}^{P^+}[a_i^{-1}]_{\alpha_i}
\Bigr )\, \tau \Bigl ({\bf n}+
\sum_{k=1}^{P^-}{\bf e}_{\beta_k}, {\bf t}+
\sum_{k=1}^{P^-}[b_k^{-1}]_{\beta_k} \Bigr ).}
\end{array}
$$
This equation contains $P^+ + P^- +1$ bilinear terms.

The simplest
nontrivial case of (\ref{L3a}) is $P^+ + P^-=3$ 
leading to 4-term relations. 
Taking into account the symmetry (\ref{sym1a}), we should consider
two cases: $(P^+ , P^-)=(3,0)$ and $(P^+ , P^-)=(2,1)$.
In the first case, 
\beq\label{L4a}
\left \{\begin{array}{l}
\displaystyle{
{\bf n}-{\bf n}'=
{\bf e}_{\alpha_1}+{\bf e}_{\alpha_2}+{\bf e}_{\alpha_3},}
\\ \\
\displaystyle{
{\bf t}-{\bf t}' =[a_1^{-1}]_{\alpha_1}+[a_2^{-1}]_{\alpha_2}
+[a_3^{-1}]_{\alpha_3}},
\end{array}\right.  
\eeq
the corresponding 3-point 4-term equation is
\beq\label{(3,0)m}
\begin{array}{l}
\phantom{a}\,
E_{\alpha_2 \alpha_1}^{-1}(a_1, a_2)
E_{\alpha_3 \alpha_1}^{-1}(a_1, a_3)
\tau \Bigl ({\bf n}\! +\! {\bf e}_{\alpha_2}\! +\! {\bf e}_{\alpha_3},
{\bf t}\! +\! [a_2^{-1}]_{\alpha_2}\! +\! 
[a_3^{-1}]_{\alpha_3}\Bigr )\tau 
\Bigl ({\bf n}+{\bf e}_{\alpha_1}, {\bf t} +[a_1^{-1}]_{\alpha_1}\Bigr )
\\ \\
+\, 
E_{\alpha_1 \alpha_2}^{-1}(a_2, a_1)
E_{\alpha_3 \alpha_2}^{-1}(a_2, a_3)
\tau \Bigl ({\bf n}\! +\! {\bf e}_{\alpha_1}\! +\! {\bf e}_{\alpha_3},
{\bf t}\! +\! [a_1^{-1}]_{\alpha_1} \! +\! 
[a_3^{-1}]_{\alpha_3}\Bigr )\tau 
\Bigl ({\bf n}+{\bf e}_{\alpha_2}, {\bf t} +[a_2^{-1}]_{\alpha_2}\Bigr )
\\ \\
+\, 
E_{\alpha_1 \alpha_3}^{-1}(a_3, a_1)
E_{\alpha_2 \alpha_3}^{-1}(a_3, a_2)
\tau \Bigl ({\bf n}\! +\! {\bf e}_{\alpha_1}\! +\! {\bf e}_{\alpha_2},
{\bf t}\! +\! [a_1^{-1}]_{\alpha_1}\! +\! 
[a_2^{-1}]_{\alpha_2}\Bigr )\tau 
\Bigl ({\bf n}+{\bf e}_{\alpha_3}, {\bf t} +[a_3^{-1}]_{\alpha_3}\Bigr )
\\ \\
=\, \tau \Bigl ({\bf n}\! +\! {\bf e}_{\alpha_1}\! +\! {\bf e}_{\alpha_2}
\! +\! {\bf e}_{\alpha_3}, {\bf t}\! +\! [a_1^{-1}]_{\alpha_1}\! +\! 
[a_2^{-1}]_{\alpha_2} \!+ \! [a_3^{-1}]_{\alpha_3}\Bigr )
\tau (n ,{\bf t}),
\end{array}
\eeq
where we have returned to the original notation $\sigma ({\bf n},
{\bf t})=\tau ({\bf n}, {\bf t})$ for odd $|{\bf n}|$.
Like in the one-component case, the second choice $(P^+, P^-)=(2,1)$
leads to the same equation.
Denoting 
$$
\begin{array}{l}
\tau^{[a_i]}=\tau \Bigl ({\bf n}+{\bf e}_{\alpha_i}, {\bf t}+
[a_i^{-1}]\Bigr ),
\\ \\
\tau^{[a_ia_j]}=\tau \Bigl ({\bf n}+{\bf e}_{\alpha_i} 
+{\bf e}_{\alpha_j}, {\bf t}+
[a_i^{-1}]+[a_j^{-1}]\Bigr ),
\end{array}
$$
we can write equation (\ref{(3,0)m}) in a more explicit form:
\beq\label{(3,0)m1}
\begin{array}{l}
\epsilon_{\alpha_2 \alpha_1}(a_1^{-1}-
a_2^{-1})^{\delta_{\alpha_1 \alpha_2}}
\tau^{[a_1a_2]} \tau^{[a_3]}+
\epsilon_{\alpha_3 \alpha_2}(a_2^{-1} -
a_3^{-1})^{\delta_{\alpha_2 \alpha_3}}
\tau^{[a_2a_3]} \tau^{[a_1]}
\\ \\
\phantom{aaaaaaaaaaaaaaaaaaaaaaaaaaaaaaaaaaa}+
\epsilon_{\alpha_1 \alpha_3}(a_3^{-1} -
a_1^{-1})^{\delta_{\alpha_3 \alpha_1}}
\tau^{[a_3a_1]} \tau^{[a_2]}
\\ \\
+ \,\, \epsilon_{\alpha_2 \alpha_1}\epsilon_{\alpha_3 \alpha_2}
\epsilon_{\alpha_1 \alpha_3}
(a_1^{-1}-a_2^{-1})^{\delta_{\alpha_1 \alpha_2}}
(a_2^{-1} -a_3^{-1})^{\delta_{\alpha_2 \alpha_3}}
(a_3^{-1} -a_1^{-1})^{\delta_{\alpha_3 \alpha_1}}
\tau^{[a_1a_2a_3]} \tau =0.
\end{array}
\eeq

\subsubsection{The dispersionless limit: a general form}

As before, the first step is to rescale all the times dividing 
them by a small parameter $\hbar$.
Like in the one-component case, we should
assume that, generally speaking, the functions $\tau$ and $\sigma$
may behave, as $\hbar \to 0$, in different ways. 
To take this possibility into account, we set
\beq\label{FG1}
\tau \Bigl (\hbar^{-1}{\bf T}\Bigr )=
\exp \left ( \frac{1}{\hbar^2}\,
F({\bf T}, \hbar )\right ), \quad
\sigma \Bigl (\hbar^{-1}{\bf T}\Bigr )=
\exp \left ( \frac{1}{\hbar^2}\,
G({\bf T}; \hbar )\right ),
\eeq
where by ${\bf T}$ we understand the full set of times (including 
discrete and continuous ones),
and assume that the functions $F$, $G$ have $\hbar$-expansions 
of the form 
$$
\begin{array}{l}
\displaystyle{
F({\bf T}; \hbar )=F_0({\bf T})+
\sum_{k\geq 1}F_k({\bf T})\hbar^k,} 
\\ \\
\displaystyle{
G({\bf T}; \hbar )=G_0({\bf T})+
\sum_{k\geq 1}G_k({\bf T})\hbar^k}.
\end{array}
$$

The dispersionless limit can be performed in the same way
as in the one-component case. Namely, 
expanding the $\hbar$-version of equation 
(\ref{L3a}) in powers of $\hbar$ as $\hbar \to 0$, 
one can see that the limit exists only if the leading terms 
of the expansions for $F$ and $G$ coincide,
i.e., $G_0=F_0$. However, in general the functions $F$ and $G$ may
differ in the next order. Taking this into account, we set
\beq\label{FGf1}
G_1({\bf T})=F_1({\bf T})-f({\bf T}).
\eeq
Then the $\hbar \to 0$ limit of (\ref{L3a}) reads
\beq\label{dL1c}
\begin{array}{l}
\displaystyle{
\sum_{s=1}^{P^+}
\prod_{{\scriptsize \begin{array}{l}i=1\\ i\neq s \end{array}}}^{P^+}
E_{\alpha_i , \alpha_s}^{-1}(a_s, a_i)
\prod_{k=1}^{P^-} E_{\beta_k , \alpha_s}(a_s, b_k)}
\\ \\
\displaystyle{\phantom{aaaaaaaaaaaa}
\times \, \exp \!
\left (\nabla_{\alpha_s} (a_s) 
\Bigl (\sum_{k=1}^{P^-}\nabla_{\beta_k} (b_k)
-  \sum_{i\neq s}^{P^+}\nabla_{\alpha_i} (a_i)\Bigr )
F_0 -\nabla_{\alpha_s} (a_s)f\Bigr )\right )
}
\\ \\
\displaystyle{
+\,\, \sum_{s=1}^{P^-}
\prod_{{\scriptsize \begin{array}{l}k=1\\ k\neq s \end{array}}}^{P^-}
E_{\beta_k , \beta_s}^{-1}(b_s, b_k)
\prod_{i=1}^{P^+} E_{\alpha_i , \beta_s}(b_s, a_i)}
\\ \\
\displaystyle{\phantom{aaaaaaaaaaaa}
\times \, \exp \! \left (\nabla_{\beta_s} (b_s) \Bigl 
(\sum_{i=1}^{P^+}\nabla_{\alpha_i} (a_i)
-  \sum_{k\neq s}^{P^-}\nabla_{\beta_k} (b_k)\Bigr )F_0 +
\nabla_{\beta_s} (b_s)f\Bigr )\right )
}
\\ \\
\phantom{aaaaa}\displaystyle{
=\, \exp \! \left (\sum_{k=1}^{P^-}\nabla_{\beta_k} (b_k)f-
\sum_{i=1}^{P^+}\nabla_{\alpha_i} (a_i)f \right ).
}
\end{array}
\eeq
It is important to note that at least two essentially
different dispersionless versions of the hierarchy exist
depending on whether the function $f$ is identically zero
or not.
In the next two subsections they are considered separately.

\subsubsection{The dispersionless limit: version I}

We begin with 
the simpler possibility to put $f=0$ (or $f={\rm const}$).
In this case equation 
(\ref{dL1c}) simplifies:
\beq\label{dL3a}
\begin{array}{l}
\displaystyle{
\sum_{s=1}^{P^+}\left (
\prod_{i\neq s}^{P^+} E_{\alpha_i \alpha_s}^{-1}(a_s, a_i)
e^{-\nabla_{\alpha_s}(a_s)\nabla_{\alpha_i}(a_i)F} \right )
\left (
\prod_{k=1}^{P^-} E_{\beta_k \alpha_s}(a_s, b_k)
e^{\nabla_{\alpha_s}(a_s)\nabla_{\beta_k}(b_k)F} \right )}
\\ \\
\displaystyle{
+\,\, \sum_{s=1}^{P^-}\left (
\prod_{k\neq s}^{P^-} E_{\beta_k \beta_s}^{-1}(b_s, b_k)
e^{-\nabla_{\beta_s}(b_s)\nabla_{\beta_k}(b_k)F} \right )
\left (
\prod_{i=1}^{P^+} E_{\alpha_i \beta_s}(b_s, a_i)
e^{\nabla_{\beta_s}(b_s)\nabla_{\alpha_i}(a_i)F} \right )=1}.
\end{array}
\eeq
Recall that $P^++P^-$ is assumed to be odd here.
In particular, for $(P^+, P^-)=(3,0)$ we obtain the dispersionless
analog of equation (\ref{(3,0)m1}), which at $a_3=\infty$ acquires
the form
\beq\label{dmL1}
\begin{array}{c}
\epsilon_{\gamma \beta}b^{-\delta_{\gamma \beta}}
e^{\nabla_{\beta}(b)\p_{\gamma}F}
+\epsilon_{\alpha \gamma}(-a)^{-\delta_{\gamma \alpha}}
e^{\nabla_{\alpha}(a)\p_{\gamma}F}+
\epsilon_{\beta \alpha}(a^{-1}-b^{-1})^{\delta_{\alpha \beta}}
e^{\nabla_{\alpha}(a)\nabla_{\beta}(b)F}
\\ \\
+\, \epsilon_{\beta \alpha}\epsilon_{\alpha \gamma}
\epsilon_{\gamma \beta}(a^{-1}-b^{-1})^{\delta_{\alpha \beta}}
b^{-\delta_{\gamma \beta}}(-a)^{-\delta_{\gamma \alpha}}
e^{\nabla_{\alpha}(a)\nabla_{\beta}(b)F+\nabla_{\alpha}(a)\p_{\gamma}F
+\nabla_{\beta}(b)\p_{\gamma}F}\, =\, 0,
\end{array}
\eeq
where we have put $a_1=a$, $a_2=b$, $\alpha_1=\alpha$, 
$\alpha_2=\beta$, $\alpha_3=\gamma$.

\begin{theorem}
The equation (\ref{dL3a}) is equivalent to the equation
\beq\label{first2}
E_{\beta \alpha}(a,b)
e^{\nabla_{\alpha}(a)\nabla_{\beta}(b)F}=
\tanh \Bigl (u_{\alpha}(a)-u_{\beta}(b)\Bigr ),
\eeq
where
\beq\label{ualpha}
u_{\alpha}(z)=\eta_{\alpha}({\bf t}) +\sum_{k\geq 1}
c^{(\alpha )}_k({\bf t}) z^{-k}.
\eeq
\end{theorem}

\noindent
{\it Proof.}
We should consider particular cases of equation (\ref{dmL1}) 
corresponding to different non-equivalent choices of
$\alpha , \beta , \gamma$. They are:

\begin{itemize}
\item[I)] \; \underline{$\alpha =\beta =\gamma \equiv \alpha$}:
\beq\label{LI}
(a^{-1}-b^{-1})e^{\nabla_{\alpha}(a)\nabla_{\alpha}(b)F} 
\Bigl (1-(ab)^{-1} e^{\nabla_{\alpha}(a)\p_{\alpha}F +
\nabla_{\alpha}(b)\p_{\alpha}F}\Bigr )=
a^{-1}e^{\nabla_{\alpha}(a)\p_{\alpha}F}-
b^{-1}e^{\nabla_{\alpha}(b)\p_{\alpha}F},
\eeq
\item[II)] \; \underline{$\beta =\alpha \neq \gamma$}:
\beq\label{LII}
(a^{-1}-b^{-1})e^{\nabla_{\alpha}(a)\nabla_{\alpha}(b)F} 
\Bigl (1-e^{\nabla_{\alpha}(a)\p_{\gamma}F +
\nabla_{\alpha}(b)\p_{\gamma}F}\Bigr )=
\epsilon_{\gamma \alpha}e^{\nabla_{\alpha}(a)\p_{\gamma}F}-
\epsilon_{\gamma \alpha}e^{\nabla_{\alpha}(b)\p_{\gamma}F},
\eeq
\item[III)] \; \underline{$\alpha \neq \beta =\gamma$}:
\beq\label{LIII}
\epsilon_{\beta \alpha}
e^{\nabla_{\alpha}(a)\nabla_{\beta}(b)F} 
\Bigl (1-\epsilon_{\beta \alpha} b^{-1}
e^{\nabla_{\beta}(b)\p_{\beta}F +
\nabla_{\alpha}(a)\p_{\beta}F}\Bigr )=
\epsilon_{\beta \alpha}e^{\nabla_{\alpha}(a)\p_{\beta}F}-
b^{-1}e^{\nabla_{\beta}(b)\p_{\beta}F},
\eeq
The right-hand side of the last equation is not symmetric under the simultaneous permutation $\alpha \leftrightarrow 
\beta$, $a\leftrightarrow b$.
So, we should add to the list the equation that is obtained 
from (\ref{LIII}) by this permutation:
\item[IV)] \; \underline{$\beta \neq \alpha =\gamma$}:
\beq\label{LIV}
\epsilon_{\beta \alpha}
e^{\nabla_{\alpha}(a)\nabla_{\beta}(b)F} 
\Bigl (1+\epsilon_{\beta \alpha} a^{-1}
e^{\nabla_{\alpha}(a)\p_{\alpha}F +
\nabla_{\beta}(b)\p_{\alpha}F}\Bigr )=
a^{-1}e^{\nabla_{\alpha}(a)\p_{\alpha}F}+
\epsilon_{\beta \alpha}e^{\nabla_{\beta}(b)\p_{\alpha}F}.
\eeq
\end{itemize}
Introducing the functions
\beq\label{ww}
w_{\alpha}(z)=z^{-1}e^{\nabla_{\alpha}(z)\p_{\alpha}F}, \qquad
w_{\alpha \beta}(z)=e^{\nabla_{\alpha}(z)\p_{\beta}F} \quad
(\alpha \neq \beta),
\eeq
we rewrite these equations in the form

\beq\label{LIa}
\hspace{-2.7cm}
(a^{-1}-b^{-1})e^{\nabla_{\alpha}(a)\nabla_{\alpha}(b)F}
\Bigl (1-w_{\alpha}(a)w_{\alpha}(b)\Bigr ) =
w_{\alpha}(a)-w_{\alpha}(b),
\eeq

\beq\label{LIIa}
(a^{-1}-b^{-1})e^{\nabla_{\alpha}(a)\nabla_{\alpha}(b)F}
\Bigl (1-w_{\alpha \beta}(a)w_{\alpha \beta}(b)\Bigr ) =
\epsilon_{\beta \alpha} \Bigl (w_{\alpha \beta}(a)-
w_{\alpha \beta}(b)\Bigr ),
\eeq

\beq\label{LIIIa}
\hspace{-2.7cm}
\epsilon_{\beta \alpha}e^{\nabla_{\alpha}(a)\nabla_{\beta}(b)F}
\Bigl (1-\epsilon_{\beta \alpha}
w_{\beta}(b)w_{\alpha \beta}(a)\Bigr ) =
\epsilon_{\beta \alpha}w_{\alpha \beta}(a)-
w_{\beta}(b),
\eeq

\beq\label{LIVa}
\hspace{-2.7cm}
\epsilon_{\beta \alpha}e^{\nabla_{\alpha}(a)\nabla_{\beta}(b)F}
\Bigl (1+\epsilon_{\beta \alpha}
w_{\alpha}(a)w_{\beta \alpha}(b)\Bigr ) =
w_{\alpha}(a)+
\epsilon_{\beta \alpha}w_{\beta \alpha}(b).
\eeq

\noindent
The first equation,
\beq\label{first}
(a^{-1}-b^{-1})e^{\nabla_{\alpha}(a)\nabla_{\alpha}(b)F}=
\frac{w_{\alpha}(a)-w_{\alpha}(b)}{1-w_{\alpha}(a)w_{\alpha}(b)},
\eeq
has the same form (\ref{L6a}) 
as in the one-component case,
and thus can be pa\-ra\-met\-ri\-zed by hyperbolic functions in the same 
way:
\beq\label{first1}
(a^{-1}-b^{-1})e^{\nabla_{\alpha}(a)\nabla_{\alpha}(b)F}=
\tanh \Bigl (u_{\alpha}(a)-u_{\alpha}(b)\Bigr ),
\eeq
with $u_{\alpha}$ as in (\ref{ualpha}).
In this parametrization,
\beq\label{ww1}
w_{\alpha}(z)=\tanh \Bigl (u_{\alpha}(z)-\eta_{\alpha}\Bigr ),
\qquad
w_{\alpha \beta}(z)=\epsilon_{\beta \alpha}
\tanh \Bigl (u_{\alpha}(z)-\eta_{\beta}\Bigr ).
\eeq
Plugging this into the other equations (\ref{LIIIa}), ({\ref{LIVa}), 
we can see that they are equivalent 
to the single equation (\ref{first2}).

It remains to show that equation (\ref{first2}) is equivalent 
to the whole hierarchy. Indeed, the substitution 
(\ref{first2}) converts
(\ref{dL3a}) into the equality
\beq\label{dL3b}
\begin{array}{l}
\displaystyle{
\phantom{+\,}\sum_{s=1}^{P^+}
\prod_{i\neq s}^{P^+} 
\coth \Bigl (u_{\alpha_s}(a_s)-u_{\alpha_i}(a_i)\Bigr )
\prod_{k=1}^{P^-} 
\tanh \Bigl (u_{\alpha_s}(a_s)-u_{\beta_k}(b_k)\Bigr )}
\\ \\
\displaystyle{
+\,\, \sum_{s=1}^{P^-}
\prod_{k\neq s}^{P^-} 
\coth \Bigl (u_{\beta_s}(b_s)-u_{\beta_k}(b_k)\Bigr )
\prod_{i=1}^{P^+} 
\tanh \Bigl (u_{\beta_s}(b_s)-u_{\alpha_i}(a_i)\Bigr )\, =\, 1.}
\end{array}
\eeq
Putting $u_{\alpha_i}(a_i)=u_i$, $u_{\beta_k}(b_k)=v_k$,
we rewrite it in the form
\beq\label{dL3c}
\begin{array}{l}
\displaystyle{
\sum_{s=1}^{P^+}
\prod_{i\neq s}^{P^+} 
\coth \Bigl (u_s-u_i\Bigr )
\prod_{k=1}^{P^-} 
\tanh \Bigl (u_s-v_k\Bigr )
+\! \sum_{s=1}^{P^-}
\prod_{k\neq s}^{P^-} 
\coth \Bigl (v_s-v_k\Bigr )
\prod_{i=1}^{P^+} 
\tanh \Bigl (v_s-u_i\Bigr ) = 1.}
\end{array}
\eeq
Using the well known formula
$$
\tanh (u-v)=\frac{\tanh u -\tanh v}{1-\tanh u \, \tanh v},
$$
we see that (\ref{dL3c}) is nothing else than the already
proved identity (\ref{L10a}) valid for odd $P^+ + P^-$.
Moreover, from that proof 
it is clear that in general case, when $P^+$ and $P^-$
are allowed to be arbitrary natural numbers, the identity
reads
\beq\label{dL3d}
\begin{array}{l}
\displaystyle{
\phantom{+\,}\sum_{s=1}^{P^+}
\prod_{i\neq s}^{P^+} 
\coth \Bigl (u_{\alpha_s}(a_s)-u_{\alpha_i}(a_i)\Bigr )
\prod_{k=1}^{P^-} 
\tanh \Bigl (u_{\alpha_s}(a_s)-u_{\beta_k}(b_k)\Bigr )}
\\ \\
\displaystyle{
+\,\, \sum_{s=1}^{P^-}
\prod_{k\neq s}^{P^-} 
\coth \Bigl (u_{\beta_s}(b_s)-u_{\beta_k}(b_k)\Bigr )
\prod_{i=1}^{P^+} 
\tanh \Bigl (u_{\beta_s}(b_s)-u_{\alpha_i}(a_i)\Bigr )}
\\ \\
=\, \frac{1}{2} \Bigl (1-(-1)^{P^+ +P^-}\Bigr ).
\end{array}
\eeq
This means that the equation (\ref{first2}) is indeed equivalent
to the whole $N$-component large dBKP hierarchy.
\square

\begin{example}
In the dispersionless limit, the simplest solution (\ref{simplest2})
reads:
\beq\label{simplest3}
F=\frac{1}{2}\sum_{\gamma =1}^N \sum_{k\geq 1}kt^2_{\gamma , k},
\eeq
so all second order derivatives in this example 
do not depend on times.
Namely, we have:
$$
\nabla_{\alpha}(a)\nabla_{\beta}(b)F=-\delta_{\alpha \beta}
\log \Bigl (1-(ab)^{-1}\Bigr ),
\quad w_{\alpha}(a)=a^{-1}, \quad w_{\alpha \beta}=1.
$$
In the trigonometric parametrization this solution can be
written in the
form
$$
u_{\alpha}(z)=\eta_{\alpha} + {\rm arctanh} (z^{-1}),
$$
where $\eta_{\alpha}\to \infty$ in the following way:
$$
\eta_{\alpha}=\alpha M, \quad M\to +\infty , \quad
\alpha =1, \ldots , N.
$$
In this case
$$
\begin{array}{l}
\eta_{\alpha}-\eta_{\beta}=(\alpha -\beta )M\to -\infty \quad
\mbox{for $\alpha <\beta $},
\\ \\
\eta_{\alpha}-\eta_{\beta}=(\alpha -\beta )M\to +\infty \quad
\mbox{for $\alpha >\beta $},
\end{array}
$$
and $\tanh (u_{\alpha}(a)-u_{\beta}(b))=
\tanh (\eta_{\alpha}-\eta_{\beta})\to \epsilon_{\beta \alpha}$
for $\alpha \neq \beta$, as it should be according to (\ref{first2}).
\end{example}

\begin{remark}
Like in Section
\ref{section:modular}, equation (\ref{first2})
can be regarded as the degeneration of (\ref{E11}) as $\tau \to +i0$.
Since in this case
$$
V_{\alpha \beta}=-2(R^2_{\alpha \beta}-1),
$$
the elliptic curve (\ref{E9}), 
$
R^2_{\alpha \beta}(w_{\alpha}^2 w_{\alpha \beta}^2 +1)-
(w_{\alpha}^2 +w_{\alpha \beta}^2) +V_{\alpha \beta}w_{\alpha}
w_{\alpha \beta}=0,
$
splits into two rational components
\beq\label{split2}
R_{\alpha \beta} (w_{\alpha}w_{\alpha \beta}+1)=\pm
(w_{\alpha} +w_{\alpha \beta} ),
\eeq
which can be uniformized by hyperbolic functions.
\end{remark}

\subsubsection{The dispersionless limit: version II}

Our aim in this section is to find more general solutions
to equation (\ref{dL1c}) for which the $f$-function is not
identically zero. We will argue that such solutions do exist
and, in contrast to version I
of the limit, are essentially ``elliptic'', i.e., they admit a 
parametrization via functions on a smooth elliptic curve,
like solutions to the dDKP hierarchy. This can be
done using the fact that the $N$-component hierarchy is 
contained in the $(N+1)$-component 
hierarchy (i.e., can be regarded as its subhierarchy).
Namely, we will show that the dispersionless
limit of the $N$-component large BKP hierarchy is essentially 
the same as that of the $(N+1)$-component DKP hierarchy.

To be more precise, consider the $(N+1)$-component DKP hierarchy,
with the $(N+1)$th component being numbered by the index $0$.
Instead of (\ref{m1a}) consider the Miwa substitution of the form
\beq \label{m1ab}
\left \{ 
\begin{array}{l}
\displaystyle{
{\bf n}-{\bf n}' =\sum_{i=1}^{P^+}{\bf e}_{\alpha_i}-
\sum_{k=1}^{P^-}{\bf e}_{\beta_k}-{\bf e}_{\beta_0},}
\\  \\
\displaystyle{
{\bf t}-{\bf t}' =\sum_{i=1}^{P^+}[a_i^{-1}]_{\alpha_i}-
\sum_{k=1}^{P^-}[b_k^{-1}]_{\beta_k}-[b_0^{-1}]_{\beta_0}},
\end{array} \right.
\eeq
where the additional index $\beta_0$ is equal to $0$ and hence 
can not coincide with any one of the indices $\alpha_i$, $\beta_k$
for $i=1, \ldots , N$. Since $P^+ +P^-$ is odd, the numbers
$\tilde P^+ =P^+$, $\tilde P^- =P^- +1$ corresponding to the
substitution (\ref{m1ab})
satisfy the parity
condition: $\tilde P^+ +\tilde P^- \in 2\ZZ$.
The Hirota-Miwa equation (\ref{gen-H-M}) for the extended
hierarchy is then written as
\beq\label{dL1d}
\begin{array}{l}
\displaystyle{
\sum_{s=1}^{P^+}
\prod_{{\scriptsize \begin{array}{l}i=1\\ i\neq s \end{array}}}^{P^+}
E_{\alpha_i , \alpha_s}^{-1}(a_s, a_i)
\prod_{k=1}^{P^-} E_{\beta_k , \alpha_s}(a_s, b_k)}\, 
E_{\beta_0 \alpha_s}(a_s, b_0)
\\ \\
\displaystyle{\phantom{aaaaaaaaaaaa}
\times \, \exp \!
\left (\nabla_{\alpha_s} (a_s) 
\Bigl (\sum_{k=1}^{P^-}\nabla_{\beta_k} (b_k)
-  \sum_{i\neq s}^{P^+}\nabla_{\alpha_i} (a_i)\Bigr )
F_0 \Bigr )+\nabla_{\alpha_s}(a_s)\nabla_{\beta_0}(b_0)F_0\right )
}
\end{array}
\eeq
$$
\begin{array}{l}
\displaystyle{
+\,\, \sum_{s=1}^{P^-}
\prod_{{\scriptsize \begin{array}{l}k=1\\ k\neq s \end{array}}}^{P^-}
E_{\beta_k , \beta_s}^{-1}(b_s, b_k)
\prod_{i=1}^{P^+} E_{\alpha_i , \beta_s}(b_s, a_i)}\,
E^{-1}_{\beta_0 \beta_s}(b_s, b_0)
\\ \\
\displaystyle{\phantom{aaaaaaaaaaaa}
\times \, \exp \! \left (\nabla_{\beta_s} (b_s) \Bigl 
(\sum_{i=1}^{P^+}\nabla_{\alpha_i} (a_i)
-  \sum_{k\neq s}^{P^-}\nabla_{\beta_k} (b_k)\Bigr )F_0 
-\nabla_{\beta_s}(b_s)\nabla_{\beta_0}(b_0)F_0
\Bigr )\right )
}
\end{array}
$$
$$
\begin{array}{l} \displaystyle{
+\,
\prod_{k=1}^{P^-}
E^{-1}_{\beta_k  \beta_0}(b_0, b_k)
\prod_{i=1}^{P^+} E_{\alpha_i  \beta_0}(b_0, a_i)}
\\ \\
\phantom{aaaaaaaaaaaaa}\displaystyle{
\times \, \exp \! \left (\nabla_{\beta_0} (b_0) \Bigl 
(\sum_{i=1}^{P^+}\nabla_{\alpha_i} (a_i)
-  \sum_{k=1}^{P^-}\nabla_{\beta_k} (b_k)\Bigr )F_0 
\Bigr )\right )=0.}
\end{array}
$$
Note that
$$
E_{\beta_0 \alpha_s}(a_s, b_0)=
E_{\beta_0 \beta_s}(b_s, b_0)=1, \quad
E_{\alpha_i \beta_0 }(b_0, a_i)=
E_{\beta_k \beta_0 }(b_0, b_k)=-1,
$$
hence the coefficient in front of the exponential
function in the last term is simply a sign
factor equal to $(-1)^{P^+ + P^-}=-1$. Now, putting
\beq\label{f111}
f=-\nabla_{\beta_0}(b_0)F_0,
\eeq
we see that equations (\ref{dL1d}) and (\ref{dL1c}) become
the same, hence the function $f$ 
solves equation (\ref{dL1c}).
In accordance with (\ref{E11}), in the
elliptic parametrization the function $f$ satisfies the equation
\beq\label{f112}
e^{-\nabla_{\alpha}(a)f}=
\frac{\theta_1(u_{\alpha}(a)-
u_{\beta_0}(b_0))}{\theta_4(u_{\alpha}(a)-u_{\beta_0}(b_0))}.
\eeq
Moreover, the function $f$ defined by equation (\ref{f111}), 
being the first order 
derivative of $F_0$ (i.e., being of the form (\ref{F1v1})), 
satisfies the linear equation (\ref{exp5}) for the $F_1$-function
for the $(N+1)$-component dDKP hierarchy. This perfectly agrees
with the original definition of the function $f$ (\ref{FGf1}):
$f=F_1-G_1$.

All this can be summarized as the following theorem.

\begin{theorem}
Dispersionless version II of the $N$-component large
BKP hierarchy is reduced to the $(N+1)$-component dDKP hierarchy:
the function $f$ defined by (\ref{f111}), (\ref{f112}) is a solution to
equation (\ref{dL1c}).
\end{theorem}

\begin{remark}
Specifying this theorem to the one-component case, we see that
one can construct a general solution to equation (\ref{dL1})
via embedding the one-component hierarchy
into the 2-component one. 
\end{remark}

\section{Conclusion and further problems}

We have considered various integrable hierarchies in their
dispersionless limits and have shown that in all cases there is
an algebraic curve built in the structure of the hierarchy.
The appearance of such a curve turns out to be a universal phenomenon.
Parameters of the curve are dynamical variables, i.e., they
are functions of times, and for this reason it is natural to call
it the dynamical curve. More precisely, 
the curve enters the game along with
a finite set of its marked points, which depend on the times as well.
(In the elliptic parametrization, the marked points are just
the variables $\eta_{\alpha}$, $\alpha =1, \ldots , N$.)
The change of dynamical variables based on uniformization 
of the dynamical curve allows one to represent equations of 
the hierarchy in a simple and nice form. This advantage is 
especially important in the multi-component case. 

What's most amazing about dynamical curves is that 
their genus ${\sf g}$ 
is not necessarily zero 
(what is customary for dispersionless hierarchies 
and could be expected) but also can be equal to 1. 
Dynamical curves of genus 1 emerge in (dispersionless) 
hierarchies of the Pfaff type: DKP, large BKP,
Pfaff-Toda and their multi-component generalizations. The uniformization 
of the (elliptic) dynamical curve by means of elliptic functions
allows one to significantly clarify the structure of hierarchies of the
Pfaff type in all the cases and represent them in a compact nice form.

Dispersionless versions of ``usual'' hierarchies 
(those of the type A) such as KP, modified KP,
Toda lattice and their multi-component generalizations have a
dynamical curve, too, but in all these cases it turns out to be rational
(of genus 0) and allows uniformization by means of elementary functions
(trigonometric or hyperbolic). For one-component hierarchies this does not
give anything new. A significant benefit of the curve 
is shown in the $N$-component hierarchies, especially for
$N\geq 3$. In these cases, 
trigonometric parameterization drastically simplifies and clarifies the 
structure of the equations.

\begin{table}[h]
\begin{center}
\begin{tabular}{|l|c|p{2cm}|p{3cm}|}
\hline
$\displaystyle{\phantom{\frac{A}{B}}}$
& degree & curve & uniformization \\
\hline
$\; \, \mbox{dKP}$ & $\begin{array}{l} \\ 1\end{array}$ 
& sphere, ${\sf g}=0$ & rational \\
\cline{1-1}
$\; \, \mbox{small dBKP}$ & & $\displaystyle{\phantom{\frac{A}{B}}}$ &  \\
\hline
\hline
$\begin{array}{l}
\mbox{$N$-dKP}, \, N\geq 2,\\
\mbox{$N$-dmKP}, \, N\geq 1\end{array}$
 & $\begin{array}{l} \\ 2\end{array}$ 
 & cylinder, ${\sf g}=0$ & trigonometric or hyperbolic\\
\cline{1-1}
$\begin{array}{l}
\mbox{$N$-comp. large dBKP,}\\
\mbox{type I}, \, N\geq 1
\end{array}$, & &  & \\
\hline
\hline
$\; \mbox{$N$-dDKP}, \, N\geq 1$ & $\begin{array}{l} \\ 4\end{array}$ 
& torus, $\;\; {\sf g}=1$ & elliptic \\
\cline{1-1}
$\begin{array}{l}
\mbox{$N$-comp. large dBKP,}\\
\mbox{type II}, \, N\geq 1
\end{array}$  &   &   & \\
\hline
\end{tabular}
\caption{Dynamical curves in various dispersionless hierarchies}
\label{table:curves}
\end{center}
\end{table}

We have also considered the multi-component large BKP hierarchy
and have shown (presumably, for the first time in the literature) that 
it admits two essentially different dispersionless versions. One
of them (version I) leads to rational dynamical curves, which can be
regarded as degenerations of elliptic curves that emerge in the case
of DKP. Namely, the former can be formally obtained from the latter
when the modular parameter $\tau$ tends to zero: $\tau \to +i0$.
(Amusingly, rational curves for the $N$-component KP are formally obtained
in the opposite limit $\tau \to +i\infty$.) The curve 
that emerge in the other dispersionless
version of large BKP (verion II) is elliptic, and the hierarchy
itself is basically equivalent to the $(N+1)$-component dDKP.

The results related to dynamical curves for various hierarchies are
summarized in Table \ref{table:curves}.``Degree'' in 
the second column of the table means 
the total degree of the polynomial that defines the curve.

Passing to problems that deserve further study, we should mention
the following natural question: whether there is some kind of 
hidden continuous parameter
that would control the curve, ``interpolating'' between the two opposite
limits of $\tau$. If it existed, it would mean 
a possibility of some kind of continuous interpolation 
between hierarchies of type A and Pfaff hierarchies.
Here is another related question: can degenerate elliptic curves 
with singularities in general position (double points) be realized
as dynamical curves for any hierarchy?

The question of how the approach and methods we have developed 
can be applied to the CKP hierarchy, including its 
multi-component version suggested in \cite{Z24b}, is also interesting. 
Any direct generalization of our approach to CKP is problematic, since 
the CKP tau-function is characterized not by bilinear equations, but by equations of the fourth degree (see \cite{AHH23,YZF25}).

Among other interesting questions waiting to be answered, 
there are two long-standing problems related to
dispersionless hierarchies of the Pfaff type.
The first one    
is to develop a Lax-Sato type formalism for the dDKP hierarchy
(and other hierarchies of the Pfaff type)
in the elliptic parametrization. The second one is to find
their geometric interpretation in the spirit of the works
\cite{MWWZ00}-\cite{MWZ02}, where it has been shown that 
the dispersionless Toda hierarchy controls 
conformal maps of simply connected plane domains 
with a smooth boundary. It is possible that 
these two problems are actually related to each other.

Lastly, the most challenging problem is 
to find out what a role  (if any) dynamical curves could play
in the theory of dispersionfull hierarchies.

\section*{Appendix A: Free fermions}

\addcontentsline{toc}{section}{Appendix A: Free fermions}
\def\theequation{A\arabic{equation}}
\def\theHequation{\theequation}
\setcounter{equation}{0}

In this appendix, we present some basic facts 
of the theory of free (multi-component) fermions.
(For a more comprehensive treatment 
see \cite{DJKM83,JM83,AZ13}.)

In the multi-component theory,
the fermionic operators are $\psi_{j}^{(\alpha )}$, 
$\psi_{j}^{*(\alpha )}$, where 
$j\in \ZZ$ and 
$\alpha =1, \ldots , N$ numbers different components.
These operators obey the standard anti-com\-mu\-ta\-ti\-on relations
$$
[\psi_{j}^{(\alpha )}, \psi_{k}^{*(\beta )}]_+=\delta_{\alpha \beta}\delta_{jk},
\qquad
[\psi_{j}^{(\alpha )}, \psi_{k}^{(\beta )}]_+=
[\psi_{j}^{*(\alpha )}, \psi_{k}^{*(\beta )}]_+=0.
$$
We also introduce 
free fermionic fields constructed as series in the variable $z\in \CC$:
$$
\psi^{(\alpha )}(z)=\sum_{j\in \z}\psi^{(\alpha )}_j z^j,
\qquad
\psi^{*(\alpha )}(z)=\sum_{j\in \z}\psi^{*(\alpha )}_j z^{-j}.
$$
The fermionic operators carry a charge: by definition, the charge
of $\psi^{(\alpha )}$ is $1$ and the charge
of $\psi^{*(\alpha )}$ is $-1$. The charge of any product
of $\psi$- and $\psi^{*}$-operators is 
product of charges of the multipliers. Linear combinations
of such products in general do not have any definite charge.

The Fock and dual Fock spaces are generated by action 
of creation operators to the vacuum states 
$\left | {\bf 0}\rbr$, $\lbr {\bf 0} \right |$ that satisfy the conditions
$$
\psi_{j}^{(\alpha )}\left | {\bf 0}\rbr =0 \quad (j<0), \qquad
\psi_{j}^{*(\alpha )}\left | {\bf 0}\rbr =0 \quad (j\geq 0),
$$
$$
\lbr {\bf 0}\right | \psi_{j}^{(\alpha )} =0 \quad (j\geq 0), \qquad
\lbr {\bf 0}\right | \psi_{j}^{*(\alpha )} =0 \quad (j< 0),
$$
so $\psi_{j}^{(\alpha )}$ with $j<0$ and $\psi_{j}^{*(\alpha )}$ with
$j\geq 0$ are annihilation operators while 
$\psi_{j}^{(\alpha )}$ with $j\geq 0$ and
$\psi_{j}^{*(\alpha )}$ with
$j<0$ are creation operators. 
Let ${\bf n}=\{ n_1, n_2, \ldots , n_N\}$ be a set
of $N$ integer numbers. The right and left vacuum states 
$\left | {\bf n}\rbr$, $\lbr {\bf n} \right |$
are defined as
$$
\left | {\bf n}\rbr =\Psi_{n_N}^{*(N)}\ldots \Psi_{n_2}^{*(2)}
\Psi_{n_1}^{*(1)}\left | {\bf 0}\rbr , \qquad
\lbr {\bf n} \right |=\lbr {\bf 0} \right |\Psi_{n_1}^{(1)}\Psi_{n_2}^{(2)}\ldots
\Psi_{n_N}^{(N)},
$$
where
$$
\Psi_{n}^{*(\alpha )}=\left \{ \begin{array}{l}
\psi^{(\alpha )}_{n-1}\ldots \psi^{(\alpha )}_{0} \quad \,\,\, (n >0)
\\
1 \phantom{\psi^{(\alpha )}_{n-1}\ldots \psi^{(\alpha )}_{0}}
\quad (n=0)
\\
\psi^{*(\alpha )}_{n}\ldots \psi^{*(\alpha )}_{-1} \quad (n <0),
\end{array}
\right.
$$
$$
\Psi_{n}^{(\alpha )}=\left \{ \begin{array}{l}
\psi^{*(\alpha )}_{0}\ldots \psi^{*(\alpha )}_{n-1} \quad  \, (n >0)
\\
1 \phantom{\psi^{(\alpha )}_{n-1}\ldots \psi^{(\alpha )}_{0}}
\quad (n=0)
\\
\psi^{(\alpha )}_{-1}\ldots \psi^{(\alpha )}_{n} \quad \,\,\,\,\, (n <0).
\end{array}
\right.
$$

The modes of the current operators $J^{(\alpha )}(z)=
\normord \psi^{(\alpha )}(z)\psi^{*(\alpha )}(z)\normord$ have the form
$$
J_{k}^{(\alpha )}=\sum_{j\in \z}\normord 
\psi^{(\alpha )}_{j} \psi^{*(\alpha )}_{j+k}.
\normord 
$$
The normal ordering $\normord (\ldots )\normord$ 
(which is essential only at $k=0$)
is defined by moving the annihilation operators 
to the right and creation operators to the left with 
the minus sign emerging each time
when two fermionic operators are permuted.  
The commutation relations of these operators are
\beq\label{com1}
[J_k^{(\alpha )} , J_l^{(\beta )}]=k\delta_{\alpha \beta}\delta_{k, -l}.
\eeq

Let
\beq\label{times}
\begin{array}{l}
{\bf t}=\{{\bf t}_1, {\bf t}_2, \ldots , {\bf t}_N\}, \qquad
{\bf t}_{\alpha}=\{t_{\alpha , 1}, t_{\alpha , 2}, t_{\alpha , 3}, 
\ldots \, \},
\\ \\
\bar {\bf t}=\{\bar {\bf t}_1, \bar {\bf t}_2, \ldots , 
\bar {\bf t}_N\}, \qquad
\bar {\bf t}_{\alpha}=\{\bar t_{\alpha , 1}, \bar t_{\alpha , 2}, 
\bar t_{\alpha , 3}, \ldots \, \},
\end{array}
\qquad \alpha = 1, \ldots , N
\eeq
be $2N$ infinite sets of the independent time
variables (in general complex numbers). 
We introduce the operators
$$
J({\bf t})=\sum_{\alpha =1}^N \sum_{k\geq 1} t_{\alpha , k}J_k^{(\alpha )},
\qquad
\bar J(\bar {\bf t})=\sum_{\alpha =1}^N 
\sum_{k\geq 1} \bar t_{\alpha , k}J_{-k}^{(\alpha )}.
$$
Their commutation relations with the fermionic fields are as follows:
\beq\label{comm}
\begin{array}{l}
e^{J({\bf t})}\psi^{(\gamma )}(z)=e^{\xi ({\bf t}_{\gamma}, z)}
\psi^{(\gamma )}(z)e^{J({\bf t})}, \quad
e^{J({\bf t})}\psi^{*(\gamma )}(z)=e^{-\xi ({\bf t}_{\gamma}, z)}
\psi^{*(\gamma )}(z)e^{J({\bf t})},
\\ \\
e^{\bar J(\bar {\bf t})}\psi^{(\gamma )}(z)=e^{\xi 
(\bar {\bf t}_{\gamma}, z^{-1})}
\psi^{(\gamma )}(z)e^{\bar J({\bf t})}, \quad
e^{\bar J(\bar {\bf t})}\psi^{*(\gamma )}(z)=
e^{-\xi (\bar {\bf t}_{\gamma}, z^{-1})}
\psi^{*(\gamma )}(z)e^{\bar J(\bar {\bf t})},
\end{array}
\eeq
where
\beq\label{f5}
\xi ({\bf t}_{\gamma}, z)=\sum_{k\geq 1}t_{\gamma , k}z^k.
\eeq

Clifford group elements of the fermionic algebra 
have the general form
\beq\label{int1a}
g=\exp \left ( \sum_{\alpha , \beta}
\sum_{j,k}\Bigl ( A_{jk}^{(\alpha \beta )}
\psi^{(\alpha )}_{j}\psi^{*(\beta )}_{k}+
B_{jk}^{(\alpha \beta )}
\psi^{(\alpha )}_{j}\psi^{(\beta )}_{k}+
C_{jk}^{(\alpha \beta )}
\psi^{*(\alpha )}_{j}\psi^{*(\beta )}_{k}\Bigr )
\right )
\eeq
with some infinite matrices $A_{jk}^{(\alpha \beta )}$, 
$B_{jk}^{(\alpha \beta )}$, $C_{jk}^{(\alpha \beta )}$, i.e.,
it is exponent of a quadratic form in the fermionic
operators $\psi_{j}^{(\alpha )}$, 
$\psi_{j}^{*(\alpha )}$. If $B=C=0$, 
the Clifford group elements
\beq\label{int1b}
g=\exp \left ( \sum_{\alpha , \beta}
\sum_{j,k} A_{jk}^{(\alpha \beta )}
\psi^{(\alpha )}_{j}\psi^{*(\beta )}_{k}
\right )
\eeq
have zero charge. In this case they are called 
{\it neutral}. If the matrices $B$ or $C$ (or both) are nonzero,
the elements $g$ do not have a definite charge. In this case only
its parity is definite: as is seen from (\ref{int1a}), the parity 
is even.

A characteristic property of the Clifford group elements 
of the general form (\ref{int1a}) is the following 
operator bilinear identity:
\beq\label{f2}
\sum_{\gamma =1}^N \sum_{j\in \z}\Bigl (
\psi_{j}^{(\gamma )}g \otimes
\psi_{j}^{*(\gamma )}g +
\psi_{j}^{*(\gamma )}g \otimes
\psi_{j}^{(\gamma )}g \Bigr )
=\sum_{\gamma =1}^N \sum_{j\in \z}\Bigl (
g\psi_{j}^{(\gamma )}\otimes g \psi_{j}^{*(\gamma )}+
g\psi_{j}^{*(\gamma )}\otimes g \psi_{j}^{(\gamma )}\Bigr ).
\eeq
The proof can be found in Appendix A of the paper \cite{SZ24}.
For neutral elements of the form (\ref{int1b}) the
identity simplifies:
\beq\label{f2a}
\sum_{\gamma =1}^N \sum_{j\in \z}
\psi_{j}^{(\gamma )}g \otimes
\psi_{j}^{*(\gamma )}g 
=\sum_{\gamma =1}^N \sum_{j\in \z}
g\psi_{j}^{(\gamma )}\otimes g \psi_{j}^{*(\gamma )}.
\eeq

Tau-functions of integrable hierarchies are realized 
as expectation values of the following general form:
\beq\label{T3a}
\tau ({\bf n}, \bar {\bf n}, {\bf t}, \bar {\bf t})=
\lbr {\bf n}\bigl | e^{J({\bf t})} g e^{-\bar J(\bar {\bf t})} 
\bigr |-\bar {\bf n}\rbr ,
\eeq
where $g$ is a Clifford group element,
with certain conditions on ${\bf n}$ and $\bar {\bf n}$ that
depend on a particular hierarchy. The operator bilinear
identities (\ref{f2}) or (\ref{f2a}) allow one to obtain 
bilinear relations for the tau-functions (\ref{T3a}).

For one-component fermions ($N=1$) the notations are simpler.
The fermionic operators are $\psi_{j}$, 
$\psi_{j}^{*}$, with the
anti-com\-mu\-ta\-ti\-on relations
$$
[\psi_{j}, \psi_{k}^{*}]_+=\delta_{jk},
\qquad
[\psi_{j}, \psi_{k}^{*}]_+=
[\psi_{j}^{*}, \psi_{k}^{*}]_+=0.
$$
The vacuum states 
$\left | 0\rbr$, $\lbr 0 \right |$ satisfy the conditions
$$
\psi_{j}\left | 0\rbr =0 \quad (j<0), \qquad
\psi_{j}^{*}\left | 0\rbr =0 \quad (j\geq 0),
$$
$$
\lbr 0\right | \psi_{j} =0 \quad (j\geq 0), \qquad
\lbr 0\right | \psi_{j}^{*} =0 \quad (j< 0),
$$
The other vacuum states $\left | n\rbr$, $\lbr n \right |$ 
(for any $n\in \ZZ$)
are defined as
$
\left | n\rbr =
\Psi_{n}^{*}\left | 0\rbr , \;\,
\lbr n \right |=\lbr 0 \right |\Psi_{n},
$
where
$$
\Psi_{n}^{*}=\left \{ \begin{array}{l}
\psi_{n-1}\ldots \psi_{0} \quad \,\,\, (n >0)
\\
1 \phantom{\psi_{n-1}\ldots \psi_{0}}
\quad (n=0)
\\
\psi^{*}_{n}\ldots \psi^{*}_{-1} \quad \;\;\,\, (n <0),
\end{array}
\right.
\qquad
\Psi_{n}=\left \{ \begin{array}{l}
\psi^{*}_{0}\ldots \psi^{*}_{n-1} \quad  \, (n >0)
\\
1 \phantom{\psi_{n-1}\ldots \psi_{0}}
\quad (n=0)
\\
\psi_{-1}\ldots \psi_{n} \quad \,\,\,\,\, (n <0).
\end{array}
\right.
$$
The modes of the current operator $J(z)=
\normord \psi(z)\psi^{*}(z)\normord$ have the form
$$
J_{k}=\sum_{j\in \z}\normord 
\psi_{j} \psi^{*}_{j+k}.
\normord 
$$ 
The commutation relations of these operators are
$[J_k , J_l]=k\delta_{k, -l}$.
The operators $J({\bf t}), \, \bar J({\bf t})$ are defines as
$$
J({\bf t})=\sum_{k\geq 1} t_{k}J_k,
\qquad
\bar J(\bar {\bf t})= 
\sum_{k\geq 1} \bar t_{k}J_{-k}.
$$
Their commutation relations with the fermionic fields are as follows:
\beq\label{comm1}
\begin{array}{l}
e^{J({\bf t})}\psi(z)=e^{\xi ({\bf t}, z)}
\psi(z)e^{J({\bf t})}, \quad
e^{J({\bf t})}\psi^{*}(z)=e^{-\xi ({\bf t}, z)}
\psi^{*}(z)e^{J({\bf t})},
\\ \\
e^{\bar J(\bar {\bf t})}\psi(z)=e^{\xi 
(\bar {\bf t}, z^{-1})}
\psi(z)e^{\bar J({\bf t})}, \quad
e^{\bar J(\bar {\bf t})}\psi^{*}(z)=
e^{-\xi (\bar {\bf t}, z^{-1})}
\psi^{*}(z)e^{\bar J(\bar {\bf t})}.
\end{array}
\eeq
Note also that 
$
J({\bf t})\bigl |n\bigr > =\bigl < n\bigr |
\bar J(\bar {\bf t})=0,
$
so $\bigl < n\bigr |e^{\bar J(\bar {\bf t})}=
\bigl < n\bigr |$, $e^{J({\bf t})}\bigl | n\bigr >=
\bigl | n\bigr >$.

The general and neutral Clifford group elements (\ref{int1a}) and
(\ref{int1b}) in the one-component case are:
\beq\label{int1a1}
g=\exp \left ( 
\sum_{j,k}\Bigl ( A_{jk}
\psi_{j}\psi^{*}_{k}+
B_{jk}
\psi_{j}\psi_{k}+
C_{jk}
\psi^{*}_{j}\psi^{*}_{k}\Bigr ),
\right ), \quad
g=\exp \left ( 
\sum_{j,k} A_{jk}
\psi_{j}\psi^{*}_{k}
\right ).
\eeq
The operator bilinear identities for them are
\beq\label{f21}
\sum_{j\in \z}\Bigl (
\psi_{j}g \otimes
\psi_{j}^{*}g +
\psi_{j}^{*}g \otimes
\psi_{j}g \Bigr )
=\sum_{j\in \z}\Bigl (
g\psi_{j}\otimes g \psi_{j}^{*}+
g\psi_{j}^{*}\otimes g \psi_{j}\Bigr )
\eeq
and
\beq\label{f2a1}
\sum_{j\in \z}
\psi_{j}g \otimes
\psi_{j}^{*}g 
= \sum_{j\in \z}
g\psi_{j}\otimes g \psi_{j}^{*}
\eeq
respectively. The tau-function is defined as
\beq\label{T3a1}
\tau (n, \bar n, {\bf t}, \bar {\bf t})=
\lbr n\bigl | e^{J({\bf t})} g e^{-\bar J(\bar {\bf t})} 
\bigr |-\bar n\rbr .
\eeq

At last, it deserves noting that the algebra of $N$-component
fermions is in fact isomorphic to the algebra of one-component
ones. The isomorphism is given by 
$\psi^{(\alpha )}_j=\psi_{Nj +\alpha -1}$,
$\psi^{*(\alpha )}_j=\psi^{*}_{Nj +\alpha -1}$, where
$\alpha =1, \ldots , N$.

\section*{Appendix B: Theta-functions}
\label{appendix:theta}

\addcontentsline{toc}{section}{Appendix B: Theta-functions}
\def\theequation{B\arabic{equation}}
\def\theHequation{\theequation}
\setcounter{equation}{0}

The Jacobi's theta-functions $\theta_a (u)=
\theta_a (u|\tau )$, $a=1,2,3,4$, are defined by the absolutely 
convergent infinite sums as follows:
\beq\label{Bp1}
\begin{array}{l}
\theta _1(u)=-\displaystyle{\sum _{k\in \z}}
\exp \left (
\pi i \tau (k+\frac{1}{2})^2 +2\pi i
(u+\frac{1}{2})(k+\frac{1}{2})\right ),
\\ \\
\theta _2(u)=\displaystyle{\sum _{k\in \z}}
\exp \left (
\pi i \tau (k+\frac{1}{2})^2 +2\pi i
u(k+\frac{1}{2})\right ),
\\ \\
\theta _3(u)=\displaystyle{\sum _{k\in \z}}
\exp \left (
\pi i \tau k^2 +2\pi i u k \right ),
\\ \\
\theta _4(u)=\displaystyle{\sum _{k\in \z}}
\exp \left (
\pi i \tau k^2 +2\pi i
(u+\frac{1}{2})k\right ),
\end{array}
\eeq where $\tau$ is a complex parameter (the modular parameter) 
such that ${\rm Im}\, \tau >0$. The function 
$\theta_1(u)$ is odd, the other three functions are even.
The infinite product representation for the theta-functions reads: 
\beq\label{infprod}
\begin{array}{ll}
\theta_1(u|\tau)&=\displaystyle{2q^{\frac{1}{4}}\sin\pi u
\prod_{n=1}^\infty(1-q^{2n})(1-q^{2n}e^{2\pi i u})(1-q^{2n}e^{-2\pi i u}),}
\\ & \\
\theta_2(u|\tau)&=\displaystyle{2q^{\frac{1}{4}}\cos\pi u
\prod_{n=1}^\infty(1-q^{2n})(1+q^{2n}e^{2\pi i u})(1+q^{2n}
e^{-2\pi i u}),}
\\ & \\
\theta_3(u|\tau)&=\displaystyle{
\prod_{n=1}^\infty(1-q^{2n})(1+q^{2n-1}e^{2\pi i u})(1+q^{2n-1}
e^{-2\pi i u}),}
\\ & \\
\theta_4(u|\tau)&=\displaystyle{
\prod_{n=1}^\infty(1-q^{2n})(1-q^{2n-1}e^{2\pi i u})
(1-q^{2n-1}e^{-2\pi i u})}.
\end{array}
\eeq
where $q=e^{\pi i \tau}$.
In the limit $\tau \to +i\infty$ they are:
$\theta_1(u|\tau )=2q^{\frac{1}{4}}\sin \pi u + O(q^{\frac{9}{4}})$,
$\theta_2(u|\tau )=2q^{\frac{1}{4}}\cos \pi u + O(q^{\frac{9}{4}})$,
$\theta_3(u|\tau )=1+O(q)$, $\theta_4(u|\tau )=1+O(q)$.

The theta-functions satisfy a lot of nontrivial identities.
Here we mention two of them:
\beq\label{id}
\theta_2^4(0)\, 
\frac{\theta_2^2 (u)\, \theta_3^2 (u)}{\theta_1^2 (u)\, \theta_4^2 (u)}=
\theta_2^2 (0)\, \theta_3^2 (0)\left (
\frac{\theta_4^2(u)}{\theta_1^2(u)}+ \frac{\theta_1^2(u)}{\theta_4^2(u)}
\right ) - \Bigl (\theta_2^4(0) +\theta_3^4(0)\Bigr )
\eeq
and
\beq\label{theta1prime}
\theta_1'(0)=\pi \theta_2(0) \theta_3(0) \theta_4(0).
\eeq
They are used in the main text.

Next, we list the transformation properties of the theta functions.

\smallskip

\noindent
Shifts by periods:
\beq\label{p}
\begin{array}{l}
\theta_1(u+1)=-\theta_1(u),\\ \\
\theta_2(u+1)=-\theta_2(u),\\ \\
\theta_3(u+1)=\theta_3(u),\\ \\
\theta_4(u+1)=\theta_4(u).
\end{array}\hspace{2cm}
\begin{array}{l}
\theta_1(u+\tau)=-e^{-\pi i(2u+\tau)}\theta_1(u),\\ \\
\theta_2(u+\tau)=e^{-\pi i(2u+\tau)}\theta_2(u),\\ \\
\theta_3(u+\tau)=e^{-\pi i(2u+\tau)}\theta_3(u),\\ \\
\theta_4(u+\tau)=-e^{-\pi i(2u+\tau)}\theta_4(u).
\end{array}\vspace{0.3cm}
\eeq
Shifts by half-periods:
\beq\label{hp}
\begin{array}{l}
\theta_1(u+{\textstyle\frac{1}{2}})=\theta_2(u),\\ \\
\theta_2(u+{\textstyle\frac{1}{2}})=-\theta_1(u),\\ \\
\theta_3(u+{\textstyle\frac{1}{2}})=\theta_4(u),\\ \\
\theta_4(u+{\textstyle\frac{1}{2}})=\theta_3(u).
\end{array}\hspace{2cm}
\begin{array}{l}
\theta_1(u+{\textstyle\frac{\tau}{2}})=
ie^{-\pi i(u+\tau/4)}\theta_4(u),\\ \\
\theta_2(u+{\textstyle\frac{\tau}{2}})=
e^{-\pi i(u+\tau/4)}\theta_3(u),\\ \\
\theta_3(u+{\textstyle\frac{\tau}{2}})=
e^{-\pi i(u+\tau/4)}\theta_2(u),\\ \\
\theta_4(u+{\textstyle\frac{\tau}{2}})=i
e^{-\pi i(u+\tau/4)}\theta_1(u).
\end{array}\vspace{0.3cm}
\eeq

We also need properties of the tau-functions
under modular transformation $\tau \to -1/\tau$:
\beq\label{mod2d}
\begin{array}{ll}
\theta_1\left (u/\tau|-1/\tau\right)&=-i\sqrt{-i\tau}\,
e^{\pi iu^2/\tau}\theta_1(u|\tau),\\ & \\
\theta_2\left (u/\tau|-1/\tau\right)&=\sqrt{-i\tau}\,
e^{\pi iu^2/\tau}\theta_4(u|\tau), \\ & \\
\theta_3\left (u/\tau|-1/\tau\right)&=\sqrt{-i\tau}\,
e^{\pi iu^2/\tau}\theta_3(u|\tau),\\ & \\
\theta_4\left (u/\tau|-1/\tau\right)&=\sqrt{-i\tau}\,
e^{\pi iu^2/\tau}\theta_2(u|\tau).
\end{array}
\eeq
The branch of the square root here
is such that $\Re \sqrt{-i\tau}>0$.

For a more detailed account of properties of the theta-functions
see \cite{Akhiezer,Takebe-book,KZ15}. 

\section*{Appendix C: Uniformization of algebraic curves}
\label{appendix:uniformization}

\addcontentsline{toc}{section}{Appendix C: Uniformization of  
algebraic curves}
\def\theequation{C\arabic{equation}}
\def\theHequation{\theequation}
\setcounter{equation}{0}

Given a complex algebraic curve $\Gamma$ defined by an equation of the
form
\beq\label{P}
P(x,y)=0, \quad x,y \in \CC ,
\eeq
where $P(x,y)$ is a polynomial, a natural question is how this curve
can be uniformized. The uniformization means that there are two 
functions, $x(u)$ and $y(u)$, of some complex variable $u$
such that: 
\begin{itemize}
\item[a)] They are single-valued in a domain 
${\sf D}\in \CC$, 
\item[b)] The equation $P(x(u),y(u))=0$ is 
satisfied identically for all $u\in {\sf D}$, 
\item[c)] Any solution
to the equation $P(x,y)=0$ is obtained in this way 
for some $u\in {\sf D}$.
\end{itemize}

The simplest example is the curve $x^2 +y^2=1$ which can be uniformized
by trigonometric functions: $x(u)=\sin u$, $y(u)=\cos u$.
More generally,
if the polynomial $P(x,y)$ is quadratic in $x,y$, i.e., of the
general form 
$$
P(x,y)= Ax^2 +By^2 +Cxy +Dx +Ey +V,
$$
equation (\ref{P}) defines a rational curve (of genus 0). In this case
the uniformization can be achieved by trigonometric or hyperbolic 
functions.
For example, consider the rational curve
\beq\label{curve1a}
xy +Ax +By +C=0
\eeq
that has appeared in Section
\ref{section:dmmKP}. Its uniformization is
\beq\label{curve1b}
\begin{array}{l}
x(u)=\gamma_1 \cot (u-\eta_1), \qquad y(u)=\gamma_2\cot (u-\eta_2)
\end{array}
\eeq
together with
\beq\label{curve1c}
A=-\gamma_2 \cot (\eta_1 -\eta_2), \qquad
B=\gamma_1 \cot (\eta_1 -\eta_2),  \qquad
C=\gamma_1 \gamma_2 .
\eeq
Here $\eta_{1,2}$ and $\gamma_{1,2}$ are parameters that
parametrize the constants $A, B, C$. Indeed, the equation defining the 
curve is satisfied identically due to the identity
\beq
\cot (u-\eta_1)\cot (u-\eta_2)+\cot (u-\eta_2)\cot (\eta_1-\eta_2)+
\cot (u-\eta_1)\cot (\eta_2 -\eta_1)+1=0,
\eeq
which can be easily proved.

Smooth curves defined by polynonial
equations of degree higher than 2 can not be uniformized
by elementary functions. For example, uniformization of curves
defined by an equation of the form  $y^2 =Q(x)$, where $Q(x)$ is
a polynomial of degree 3 or 4, requires elliptic functions.
In the case of degree 4 the canonical form of the equation is
\beq\label{uni1}
y^2 =(1-x^2)(1-k^2 x^2),
\eeq
where $k$ is a parameter called elliptic modulus. If $k\neq 0,1$,
the curve is a smooth elliptic curve (a torus).
It can be uniformized by the elliptic functions ${\rm sn}(w)$
(the ``elliptic sinus''), ${\rm cn}(w)$
(the ``elliptic cosinus'') and ${\rm dn}(w)$:
\beq\label{uni2}
x(w)={\rm sn}(w), \quad y(w)={\rm cn}(w)\, {\rm dn}(w)=x'(u).
\eeq
They are expressed through the 
Jacobi theta-functions from Appendix B
in the following way:
\beq\label{uni3}
\begin{array}{l}
\displaystyle{
{\rm sn}(w)=\frac{\theta_3(0)\, 
\theta_1(u)}{\theta_2(0)\, \theta_4(u)}},
\quad
\displaystyle{
{\rm cn}(w)=\frac{\theta_4(0)\, 
\theta_2(u)}{\theta_2(0)\, \theta_4(u)}},
\quad
\displaystyle{
{\rm dn}(w)=\frac{\theta_4(0)\, 
\theta_3(u)}{\theta_3(0)\, \theta_4(u)}},
\end{array}
\eeq
where 
\beq\label{uni5}
u=\frac{w}{\pi \, \theta_3^2(0)}
\eeq
and
the modular parameter $\tau$ of the theta-functions is connected
with the elliptic modulus $k$ by the formula
\beq\label{uni4}
k =\frac{\theta_2^2(0|\tau )}{\theta_3^2(0|\tau )}.
\eeq

If the curve is defined by an 
equation $P(x,y)=0$, where $P(x,y)$ is a bi-quadratic polynomial
in the variables $x,y$, it is in general 
a smooth elliptic curve (of genus 1). 
For its uniformization one needs 
elliptic functions or Jacobi theta-functions.
(See, for example, the last section of Baxter's book \cite{Baxter}.)

Our first example is the curve
\beq\label{B5}
R^2(x^2y^2+1)-(x^2+y^2)+Vxy=0.
\eeq
The rational change of variables $(x,y)\to (X,Y)$, where
$$
Y=Rx^2y +\frac{Vx-2Ry}{2R}, \quad
X=k^{-1/2}x \quad \mbox{with 
$\displaystyle{k+k^{-1}=R^2 +R^{-2} -\frac{V^2}{4R^2}}$}
$$
brings equation (\ref{B5}) to the canonical form (\ref{uni1}),
i.e., $Y^2 =(1-X^2)(1-k^2 X^2)$.
In the original variables the uniformization (\ref{uni2})
acquires the form
\beq\label{B6}
x(u)=\frac{\theta_1(u|\tau )}{\theta_4(u |\tau )}, \qquad
y(u)=\frac{\theta_1(u+\eta |\tau )}{\theta_4(u+\eta |\tau )}.
\eeq
The two constants $R$, $V$ are expressed in terms 
of two parameters $\eta$, $\tau$ as follows:
\beq\label{B7}
R=\frac{\theta_1(\eta |\tau )}{\theta_4(\eta |\tau )}, \qquad
V=2\frac{\theta_4^2(0|\tau )
\theta_2(\eta |\tau )\theta_3(\eta|\tau  )}{\theta_2(0|\tau )
\theta_3(0|\tau )\theta_4^2(\eta |\tau )}.
\eeq
To verify the uniformization formulas 
directly, one should prove the identity
\beq\label{B8}
\begin{array}{c}
\displaystyle{
\frac{\theta_1^2(\eta )}{\theta_4^2(\eta )}\left (
\frac{\theta_4^2(u)\theta_4^2(u+\eta )}{\theta_1^2(u)\theta_1^2(u+\eta )}
+1\right ) -\left (\frac{\theta_4^2(u)}{\theta_1^2(u)}+
\frac{\theta_4^2(u+\eta )}{\theta_1^2(u+\eta )}\right )}
\\ \\
\displaystyle{
+2\, \frac{\theta_4^2(0)\theta_2(\eta )\theta_3(\eta )\theta_4(u)
\theta_4 (u+\eta )}{\theta_2(0)\theta_3(0)\theta_4^2(\eta )
\theta_1(u)\theta_1(u+\eta )}=0}
\end{array}
\eeq
which is equation (\ref{B5}) after the substitutions (\ref{B6}), (\ref{B7}).
(For notational simplicity, we omit the modular 
parameter $\tau$, which is the same for all
theta-functions in (\ref{B8}).)
The left-hand side is an elliptic function of $u$ with possible poles
at $u=0$ and $u=-\eta$ of at most second order. It is easy to see that 
the highest singularities (second order poles) cancel.
Therefore, the left-hand side is an elliptic function of $u$ with
possible simple poles at $u=0$ and $u=-\eta$. Therefore, it is enough
to establish the equality at three distinct points. It is easy to see that
the left-hand side equals 0 at $u=\frac{\tau}{2}$ 
and $u=-\eta +\frac{\tau}{2}$. As the third point we
take $u=\frac{\tau +1}{2}$. At this point, the left-hand side is
\beq\label{B9}
\mbox{L.h.s.}=\frac{\theta_1^2(\eta )}{\theta_4^2(\eta )}
\left (\frac{\theta_2^2(0)\theta_2^2(\eta )}{\theta_3^2(0)\theta_3^2(\eta )}
+1\right ) -\left (\frac{\theta_2^2(0)}{\theta_3^2(0)}+
\frac{\theta_2^2(\eta )}{\theta_3^2(\eta )}\right )+2\,
\frac{\theta_4^2(0)\theta_2^2(\eta )}{\theta_3^2(0)\theta_4^2(\eta )}.
\eeq
It is an even elliptic function of $\eta -\frac{\tau}{2}$ and 
$\eta -\frac{\tau +1}{2}$ with possible second order poles at $\eta =
\frac{\tau}{2}$ and $\eta =\frac{\tau +1}{2}$. The expansion around 
these points shows that the singular terms cancel. Therefore, 
expression (\ref{B9}) 
does not depend on $\eta$. Substituting $\eta =0$, we see that
it is equal to zero. This proves the identity (\ref{B8}).

It is worth noting that equation of the curve in the form
(\ref{B5}) contains two parameters ($R$ and $V$) while in the
canonical equation (\ref{uni1}) there is only one, $k$. 
The explanation of this apparent discrepancy is that 
equation (\ref{B5}) defines not only the curve itself, but 
the curve with a marked point on it. In the elliptic parametrization,
this point is just $\eta$.

The second example is the curve
\beq\label{B1}
y^2 -R^2 (x^2 +x^{-2})+V=0,
\eeq
or, in the polynomial form,
\beq\label{B1a}
x^2y^2 -R^2x^4 +Vx^2 -R^2=0.
\eeq
In this case the rational change of variables that brings
it to the canonical form is
$$
Y=\frac{xy}{R}, \quad
X=k^{-1/2}x \quad \mbox{with 
$\displaystyle{k+k^{-1}=\frac{V}{R^2}}$}.
$$
In terms of the original variables the uniformization 
(\ref{uni2}) reads:
\beq\label{B2}
x(u)=\frac{\theta_1(u)}{\theta_4(u)}, \qquad
y(u)=\gamma \theta_4^2(0)\frac{\theta_2(u)\theta_3(u)}{\theta_1(u)
\theta_4(u)}=\frac{\gamma}{\pi}\, \frac{x'(u)}{x(u)}
\eeq
and
\beq\label{B3}
R=\gamma 
\theta_2(0)\theta_3(0), \qquad V=\gamma^2
\Bigl (\theta_2^4(0)+\theta_3^4(0)\Bigr ),
\eeq
where $\gamma$ is an arbitrary constant. 
To verify validity of these formulas directly, one should
prove the identity
\beq\label{B4}
\theta_4^4(0)\frac{\theta_2^2(u)\theta_3^2(u)}{\theta_1^2(u)
\theta_4^2(u)}-\theta_2^2(0)\theta_3^2(0)\left (
\frac{\theta_4^2(u)}{\theta_1^2(u)}+\frac{\theta_1^2(u)}{\theta_4^2(u)}
\right )+\theta_2^4(0)+\theta_3^4(0)=0.
\eeq
For the proof we note that the left-hand side is an even elliptic
function of $u$ with possible poles at $u=0$ and $u=\frac{\tau}{2}$.
However, the expansion around these points shows that the singular terms
cancel and the function is regular everywhere. This means that it is
a constant. To find the constant one can substitute any value of $u$.
It is convenient to take $u=\frac{1}{2}$. Using the transformation 
properties (\ref{hp}), one finds that the constant is zero.


\section*{Appendix D: $N$-component versus 1-component dmKP}

\addcontentsline{toc}{section}{Appendix D: 
$N$-component versus 1-component dmKP}
\def\theequation{D\arabic{equation}}
\def\theHequation{\theequation}
\setcounter{equation}{0}

In this appendix we show how the general approach 
developed in Section \ref{section:dmmKP} for the $N$-component
dmKP hierarchy can be applied to the case $N=1$.
To do this, we should take into account 
that in the one-component mKP hierarchy
the discrete variable $n$ is frozen to the value
$n=0$ because of the
condition (\ref{mmkp1a}), and only the $m$-variable is alive.

We recall that the dmKP hierarchy is equivalent to the system of two
equations (\ref{mkp13})\footnote{As we have seen in Section \ref{section:dmKP}, these two
equations are actually equivalent.}:
\beq\label{mkp13a}
\left \{
\begin{array}{l}
(a^{-1}-b^{-1})e^{\tilde \nabla (a)\tilde \nabla (b)F}
=\tilde w(a)-\tilde w(b),
\\ \\
(a^{-1}-b^{-1})e^{\tilde \nabla (a)\tilde \nabla (b)F}=
e^{-\tilde \p_0^2 F}
\tilde w(a)\tilde w(b)
\Bigl (\tilde p(b)-\tilde p(a)\Bigr ).
\end{array}
\right.
\eeq
In particular, from the first equation 
it follows that
\beq\label{t15b}
\tilde w(z)=z^{-1}e^{\tilde \nabla (z)\tilde \p_{0} F}.
\eeq
In fact, there are two different ways to obtain the 
(one-component) dmKP hierarchy from the $N$-component one.

\subsubsection*{The dmKP from the $N$-component dmKP 
in the trigonometric pa\-ra\-met\-ri\-za\-ti\-on at $N=1$}

We should identify (\ref{mkp13a}) with the system (\ref{t15}) at $N=1$:
\beq\label{t15a}
\left \{
\begin{array}{l}
(a^{-1}-b^{-1})
e^{\nabla (a)\nabla (b)F}=
\sin \Bigl (u (a)-u(b)\Bigr ),
\\ \\
e^{\nabla (a)\bar \p_{0} F}=\sin \Bigl (u (a)
-\bar \eta \Bigr ) =\bar w(a),
\end{array} \right.
\eeq
where
$\displaystyle{
u(z)=\eta +\sum_{k\geq 1}c_k z^{-k}.}
$
In particular,
the second equation in (\ref{t15a}) implies that
\beq\label{t15c}
e^{\p_{0} \bar \p_{0} F}=\sin (\eta -\bar \eta ) .
\eeq
Recall the connection 
(\ref{mmkp17}) between the derivatives $\p_{\alpha}$,
$\bar \p_{\alpha}$ and $\tilde \p_{\alpha}$, which in the
one-component case is simply
\beq\label{c1}
\tilde \p_{0} =\p_{0} -\bar \p_{0} \qquad (\tilde \p_{0} =\p_{\tilde t_0},
\; \p_{0} =\p_{t_0}, \; \bar \p_{0} =\p_{\bar t_0}).
\eeq
Therefore, $\nabla (z)=\tilde \nabla (z)+\bar \p$.
Using this, we can write:
\beq\label{c2}
\begin{array}{lll}
\tilde w(z)& = &z^{-1}e^{\tilde \nabla (z)\tilde \p_{0} F}
=z^{-1} e^{\nabla (z)\p_{0} F -\nabla (z) 
\bar \p_{0} F -\p_{0} \bar \p_{0} F
+\bar \p^2_{0} F}
\\ && \\
&=& w(z)(\bar w(z))^{-1} \, e^{-\p_{0} \bar \p_{0} F
+\bar \p^2_{0} F}
\\ && \\
&=& \displaystyle{
\frac{\sin (u(z)-\eta )}{\sin (u(z)-\bar \eta )}\, e^{\bar \p_{0} (
\bar \p_{0} -\p_{0})F}}.
\end{array}
\eeq
This equation establishes the relation between
the functions $\tilde w(z)$ and $u(z)$. Let us show 
that this relation means that 
the first equations in (\ref{mkp13a}) and (\ref{t15a}) 
are equivalent.
Indeed, 
\beq\label{c3}
\begin{array}{lll}
\tilde w(a)-\tilde w(a)&=& \displaystyle{
e^{\bar \p_{0} (\bar \p_{0} -\p_{0} )F}\, 
\left (
\frac{\sin (u(a)-\eta )}{\sin (u(a)-\bar \eta )}
-\frac{\sin (u(b)-\eta )}{\sin (u(b)-\bar \eta )} \right )}
\\ && \\
&=& e^{\bar \p_{0} (\bar \p_{0} -\p_{0} )F}\, 
\displaystyle{
\frac{\sin (u(a)-u(b))
\sin (\eta -\bar \eta )}{\sin (u(a)-\bar \eta )\sin (u(b)-\bar \eta )}},
\end{array}
\eeq
while the left-hand side of (\ref{mkp13a}) is
\beq\label{c4}
(a^{-1}-b^{-1})e^{\tilde \nabla (a)\tilde \nabla (b)F}=
(a^{-1}-b^{-1})e^{\bar \p^2_{0} F-\nabla (a)\bar \p_{0} F-
\nabla (b)\bar \p_{0} F +\nabla (a)\nabla (b)F}.
\eeq
Equating the right-hand sides of (\ref{c3}) and (\ref{c4}),
we get:
\beq\label{c5}
(a^{-1}-b^{-1})e^{\nabla (a)\nabla (b)F}=
e^{\nabla (a)\bar \p_{0} F+\nabla (b)\bar \p_{0} F 
-\p_{0} \bar \p_{0} F}\,
\frac{\sin (u(a)-u(b))
\sin (\eta -\bar \eta )}{\sin (u(a)-\bar \eta )
\sin (u(b)-\bar \eta )},
\eeq
which is satisfied identically due to the second equation in (\ref{t15a})
and equation (\ref{t15c}).

One can also show that the curve (\ref{mkp15}) after a simple change
of variables coincides with the curve (\ref{mmkp24}).

\subsubsection*{The 1-component dmKP from the 2-component dmKP}

In the 2-component dmKP there are two sets of continuous times,
${\bf t}_1$ and ${\bf t}_2$ and the former discrete variables
$t_{1,0}$ and $t_{2,0}$ such that $t_{2,0}=-t_{1,0}$.
The ``zeroth'' variable in the one-component dmKP will be
$t_0=t_{1,0}$. The corresponding vector fields are, therefore,
related as
\beq\label{v1}
\p_0 =\p_1-\p_2,
\eeq
where $\p_0=\p_{t_0}$,  $\p_1=\p_{t_{1,0}}$, $\p_2=\p_{t_{2,0}}$.

The 1-component dmKP is equivalent not to the whole 
2-component dmKP but to its ``half'' obtained by freezing
the continuous times ${\bf t}_2$ and identifying ${\bf t}_1={\bf t}$.
Having this in mind, we write the relevant 
equations of the 2-component
dmKP in the trigonometric form as follows:
\beq\label{v2}
\left \{
\begin{array}{l}
(a^{-1}-b^{-1})e^{\nabla_1(a)\nabla_1(b)F} =
\sin \Bigl (u_1(a)-u_1(b)\Bigr ), 
\\ \\
e^{\nabla_1(a)\p_2 F} =\sin \Bigl (u_1(a)-\eta_2\Bigr ) =w_{12}(a).
\end{array}\right.
\eeq
Putting $b=\infty$ in the first equation, we have:
\beq\label{v4}
a^{-1}e^{\nabla_1(a)\p_1 F} =\sin \Bigl (u_1(a)-\eta_1 \Bigr )
=w_1(a).
\eeq
In a similar way, from the second equation it follows that
$e^{\p_1 \p_2 F}=\sin (\eta_1 -\eta_2).$

Dividing equation (\ref{v4}) by the second equation in (\ref{v2}),
we get:
\beq\label{v5}
a^{-1}e^{\nabla_1(a)\p_0 F} =
\frac{\sin \Bigl (u_1(a)-\eta_1\Bigr )}{\sin \Bigl (u_1(a)-\eta_2\Bigr )}.
\eeq
We will use the operators\footnote{The latter 
one is going to be the $\nabla$-operator for 
the dmKP.}
$
\nabla_1(z)=\p_1 +D_1(z), \;
\nabla_0(z)=\p_0 +D_1(z)
$
(the notation $D_1$ means here that this operator
contains derivatives with respect to the times ${\bf t}_1$).
Let us rewrite (\ref{v5}) in the form
\beq\label{v6}
a^{-1}\sin (\eta_1 -\eta_2)e^{\nabla_0(a)\p_0 F}=
e^{\p_0^2F} \, 
\frac{\sin \Bigl (u_1(a)-\eta_1\Bigr )}{\sin \Bigl 
(u_1(a)-\eta_2\Bigr )}.
\eeq
Now we can write the following chain of equalities:
\beq\label{v7}
\begin{array}{lll}
(a^{-1}-b^{-1})e^{\nabla_0(a)\nabla_0(b)F}& =&
\displaystyle{
(a^{-1}-b^{-1})\frac{ab w_1(a)w_1(b)}{w_{12}(a)w_{12}(b)}\,
e^{\p_2^2F -\p_1^2 F +D_1(a)D_1(b)F}}
\\ && \\
&=&
\displaystyle{
(a^{-1}-b^{-1})
\, \frac{e^{\p_2^2F} \, e^{\nabla_1(a)\nabla_1(b)F}}{\sin \Bigl (
u_1(a)-\eta_2 \Bigr )\sin \Bigl (
u_1(b)-\eta_2 \Bigr )}}
\\ && \\
&=&
\displaystyle{
\, \frac{e^{\p_2^2F} \,\sin \Bigl (
u_1(a)-u_1(b) \Bigr )}
{\sin \Bigl (u_1(a)-\eta_2 \Bigr )
\sin \Bigl (u_1(b)-\eta_2 \Bigr )}}.
\end{array}
\eeq
Writing the $\sin$-function in the numerator as
$$
\sin \Bigl (u_1(a)-u_1(b) \Bigr )=
\sin \Bigl ((u_1(a)-\eta_2) -(u_1(b)-\eta_2) \Bigr )
$$
and using the identity $\sin (x-y)=\sin x \cos y -\cos y \sin x$,
we arrive at the equation that expresses 
the left-hand side of (\ref{v7}) as a difference of the form
$g(b)-g(a)$ with a function $g(z)$:
\beq\label{v8}
(a^{-1}-b^{-1})e^{\nabla_0(a)\nabla_0(b)F}=
e^{\p_2^2F}\, \Bigl (
\cot (u_1(b)-\eta_2 )-\cot (u_1(a)-\eta_2 )\Bigr ).
\eeq
Using the identity
$$
\cot (u-\eta_2)-\cot (v-\eta_2)=-
\frac{1}{\sin (\eta_1-\eta_2)}\left (
\frac{\sin (u-\eta_1)}{\sin (u-\eta_2)} -
\frac{\sin (v-\eta_1)}{\sin (v-\eta_2)}\right ),
$$
we rewrite equation (\ref{v8}) in the form
\beq\label{v9}
(a^{-1}-b^{-1})e^{\nabla_0(a)\nabla_0(b)F}=w_0(a)-w_0(b),
\eeq
where
\beq\label{v10}
w_0(z)=\frac{e^{\p_2^2F}\, 
\sin \Bigl (u_1(z)-\eta_1\Bigr )}{\sin \Bigl (u_1(z)-\eta_2\Bigr )
\sin \Bigl (\eta_1-\eta_2\Bigr )}.
\eeq
The following simple calculation shows that
the $w_0(z)$ defined in this way does equal
$z^{-1}e^{\nabla_0(z)\p_0 F}$:
$$
w_0(z)=e^{\p_2 (\p_2-\p_1)F} \, \frac{w_1(z)}{w_{12}(z)}=
z^{-1}e^{\p_0^2 F+D_1(z)\p_0F}=
z^{-1}e^{\nabla_0(z)\p_0 F}.
$$
So, we have shown that equations (\ref{v2}) do contain 
the one-component dmKP hierarchy.

\section*{Acknowledgments}
\addcontentsline{toc}{section}{Acknowledgments}

The work of A.S. (Sections 4, 5.2, 5.3, 6.1, 6.2)
and A.Z (Sections 2, 3, 5.1, 6.3, 7) 
was implemented in the framework 
of the Basic Research Program at HSE University (HSE-BR-2025-84).


\end{document}